\documentclass[a4paper,11pt]{article} 

\usepackage[english]{babel}
\usepackage[a4paper]{geometry}
\usepackage{enumerate}  
\usepackage{amsmath}
\usepackage{dsfont}
\usepackage{graphicx}
\usepackage{color}
\usepackage{amsthm}
\usepackage{amssymb}
\usepackage{accents}

\usepackage{mathtools}
\usepackage{bbm}
\usepackage{cancel}
\usepackage{soul}

\usepackage{hyperref}    

\newcommand{\hc}{\mbox{h.c.}}

\usepackage{xcolor}
\usepackage{ulem}  

\let\a=\alpha     \let\g=\gamma     \let\d=\delta     \let\e=\varepsilon
        \let\k=\kappa     \let\l=\lambda
    \let\n=\nu      \let\x=\xi                
\let\s=\sigma \let\t=\tau         \let\ph=\varphi   \let\c=\chi
        
 \let\D=\Delta   \let\Th=\Theta        \let\X=\Xi

\def\cE{{\cal E}}\def\cV{{\cal V}}
\def\cC{{\cal C}}\def\cF{{\cal F}}\def\cH{{\cal H}}
\def\cN{{\cal N}}
\def\cL{{\cal L}}\def\cJ{{\cal J}}
\def\cD{{\cal D}}\def\cA{{\cal A}}\def\cG{{\cal G}}
\def\cO{{\cal O}}\def\cK{{\cal K}}\def\cU{{\cal U}}

 \def\bP{{\bf P}}

  \def\v0{{\vec 0}}

\def\bal{{\bar \l}}

\def\tl#1{{\tilde{#1}}}

\newcommand{\dbtilde}[1]{\hat #1}

\def\bR{\mathbb{R}}
\def\cU{\mathcal{U}}
\def\cV{\mathcal{V}}
\def\cF{\mathcal{F}}
\def\cG{\mathcal{G}}
\def\cL{\mathcal{L}}
\def\cN{\mathcal{N}}
\def\cE{\mathcal{E}}
\def\cK{\mathcal{K}}
\def\cH{\mathcal{H}}

\def\ph{\varphi}

\def\NNN{\mathbb{N}}  
\def\ZZZ{\mathbb{Z}} 
\def\RRR{\mathbb{R}}

\def\bN{\mathbb{N}}  
\def\bZ{\mathbb{Z}} 
\def\bR{\mathbb{R}} 
\def\bP{\mathbb{P}}

\def\indic{\hbox{\raise-2pt \hbox{\indbf 1}}}

\let\==\equiv

\let\io=\infty
\let\0=\noindent

\def\*{{\hfill\break\null\hfill\break}}

\def\la{{\langle}}
\def\ra{{\rangle}}
\def\norm#1{{\left|\hskip-.05em\left|#1\right|\hskip-.05em\right|}}

\def\tende#1{\,\vtop{\ialign{##\crcr\rightarrowfill\crcr
			\noalign{\kern-1pt\nointerlineskip}
			\hskip3.pt${\scriptstyle #1}$\hskip3.pt\crcr}}\,}
\def\otto{\,{\kern-1.truept\leftarrow\kern-5.truept\to\kern-1.truept}\,}

\def\wt#1{\widetilde{#1}}

\def\sqt[#1]#2{\root #1\of {#2}}

\def\hc{{\rm h.c.}\,}

\def\tr{{\rm tr}}

\def\wt{\widetilde}

\def\wt{\widetilde}
\def\gxt{\gamma_{x}}
\def\sxt{\sigma_{x}}

\def\be{\begin{equation}}
	\def\ee{\end{equation}}
\def\bea{\begin{eqnarray}}\def\eea{\end{eqnarray}}
\def\bean{\begin{eqnarray*}}\def\eean{\end{eqnarray*}}
\def\bfr{\begin{flushright}}\def\efr{\end{flushright}}
\def\bc{\begin{center}}\def\ec{\end{center}}
\def\bal{\begin{align}} 
	\def\eal{\end{align}}

\def\spl#1\spl{\[ \begin{split}#1\end{split} \]}

\def\bd{\begin{description}}\def\ed{\end{description}}
\def\bv{\begin{verbatim}}\def\ev{\end{verbatim}}

\def\Halmos{\hfill\vrule height10pt width4pt depth2pt \par\hbox to \hsize{}}

\newtheorem{theorem}{Theorem}[section]
\newtheorem{prop}[theorem]{Proposition}
\newtheorem{lemma}[theorem]{Lemma} 
\theoremstyle{remark}

\numberwithin{equation}{section}
\def \aa{{\mathfrak a}}

\newcommand{\wln}{w_{N,\ell}}
\newcommand{\wlnxy}{w_{N,\ell}(x-y)}
\newcommand{\fln}{f_{N,\ell}}

\newcommand{\wmnxy}{w_{N,m}(x-y)}

\newcommand{\fmnxy}{f_{N,m}(x-y)}
\newcommand{\VN}{V_N}
\newcommand{\VNxy}{V_N(x-y)}

\newcommand{\VNxs}{V_N(x-s)}

\newcommand{\VNys}{V_N(y-s)}
\newcommand{\VNrs}{V_N(r-s)}

\newcommand{\axx}{a_x}
\newcommand{\axxs}{a_x^*}
\newcommand{\ayy}{a_y}
\newcommand{\ayys}{a_y^*}
\newcommand{\azz}{a_z}
\newcommand{\azzs}{a_z^*}

\newcommand{\arr}{a_r}
\newcommand{\arrs}{a_r^*}
\newcommand{\ass}{a_s}
\newcommand{\asss}{a_s^*}

\newcommand{\bxx}{b_x}
\newcommand{\bxxs}{b_x^*}

\newcommand{\byys}{b_y^*}
\newcommand{\bzz}{b_z}
\newcommand{\bzzs}{b_z^*}

\newcommand{\brr}{b_r}
\newcommand{\brrs}{b_r^*}
\newcommand{\bss}{b_s}
\newcommand{\bsss}{b_s^*}

\newcommand{\Bxx}{B_x}

\newcommand{\Brr}{B_r}
\newcommand{\Brrs}{B_r^*}

\newcommand{\bgx}{b(\g_x)}
\newcommand{\bgxs}{b^*(\g_x)}
\newcommand{\bsx}{b(\s_x)}
\newcommand{\bsxs}{b^*(\s_x)}

\newcommand{\bgys}{b^*(\g_y)}
\newcommand{\bsy}{b(\s_y)}

\newcommand{\agx}{a(\g_x)}

\newcommand{\asxs}{a^*(\s_x)}

\newcommand{\agr}{a(\g_r)}
\newcommand{\agrs}{a^*(\g_r)}
\newcommand{\asr}{a(\s_r)}
\newcommand{\asrs}{a^*(\s_r)}
\newcommand{\bgr}{b(\g_r)}
\newcommand{\bgrs}{b^*(\g_r)}
\newcommand{\bsr}{b(\s_r)}

\newcommand{\bxxtilde}{\tilde{b}_x}
\newcommand{\bxxstilde}{\tilde{b}_x^*}
\newcommand{\byytilde}{\tilde{b}_y}
\newcommand{\byystilde}{\tilde{b}_y^*}
\newcommand{\bgxtilde}{\tilde{b}(\g_x)}

\newcommand{\bsxstilde}{\tilde{b}^*(\s_x)}

\newcommand{\Bxxtilde}{\tilde{B}_x}

\newcommand{\bgu}{b(\g_u)}
\newcommand{\bgus}{b^*(\g_u)}
\newcommand{\bsu}{b(\s_u)}
\newcommand{\bsus}{b^*(\s_u)}

\newcommand{\dxxs}{d_x^*}

\newcommand{\dyys}{d_y^*}

\newcommand{\intxy}{\int dxdy}
\newcommand{\intrs}{\int drds}

\newcommand{\cVN}{\cV_N}
\newcommand{\cVh}{\cV_N^{1/2}}
\newcommand{\cHN}{\cH_N}

\newcommand{\cHNplusoneh}{(\cH_N+1)^{1/2}}

\newcommand{\cNplusoneto}[1]{(\cN + 1)^{#1}}
\newcommand{\cNplusone}{(\cN + 1)}
\newcommand{\cNplusoneh}{(\cN + 1)^{1/2}}
\newcommand{\cKh}{\cK^{1/2}}
\newcommand{\cHNplusNh}{(\cHN + \cN +1)^{1/2}}
\newcommand{\cHNplusN}{(\cHN + \cN +1)}

\newcommand{\ptt}{\wt{\ph}_t}
\def\ptx{\wt{\ph}_t(x)}
\def\pty{\wt{\ph}_t(y)}

\def\ptr{\wt{\ph}_t(r)}
\def\pts{\wt{\ph}_t(s)}
\def\dptt{\dot{\wt{\ph}}_t}
\def\dptx{\dptt (x)}
\def\dpty{\dptt (y)}

\def\gyt{\gamma_{y}}
\def\syt{\sigma_{y}}

\newcommand{\Ftru}{\cF^{\leq N}}
\newcommand{\Fperpt}{\cF^{\leq N}_{\perp\ptt}}

\newcommand{\cutoffNlessM}{\mathbbm{1}^{\leq M}}

\newcommand{\expt}{e^{c\vert t \vert}}
\newcommand{\expexpt}{e^{c\expt}}

\newcommand{\abs}[1]{\lvert #1 \rvert}

\renewcommand{\norm}[1]{\| #1 \|}

\title{Quantum Fluctuations of Many-Body Dynamics  \\ around the Gross-Pitaevskii Equation}

\author{
Cristina Caraci$^{*,}$\footnote{Electronic mail: cristina.caraci@math.uzh.ch}\;, 
Jakob Oldenburg$^{*,}$\footnote{Electronic mail: jakob.oldenburg@math.uzh.ch}\;, and Benjamin Schlein$^{*,}$\footnote{Electronic mail: benjamin.schlein@math.uzh.ch} \\[0.2cm]
{\footnotesize $^{*}$Institute of Mathematics, University of Zurich, Winterthurerstrasse 190, 8057 Zurich.}\\
}

\date{\today}

\begin{document}
	\maketitle

\begin{abstract}
We consider the evolution of a gas of $N$ bosons in the three-dimensional Gross-Pitaevskii regime (in which particles are initially trapped in a volume of order one and interact through a repulsive potential with scattering length of the order $1/N$). We construct a quasi-free approximation of the many-body dynamics, whose distance to the solution of the Schr\"odinger equation converges to zero, as $N \to \infty$, in the $L^2 (\bR^{3N})$-norm. To achieve this goal, we let the Bose-Einstein condensate evolve according to a time-dependent Gross-Pitaevskii equation. After factoring out the microscopic correlation structure, the evolution of the orthogonal excitations of the condensate is governed instead by a Bogoliubov dynamics, with a time-dependent generator quadratic in creation and annihilation operators. As an application, we show a central limit theorem for fluctuations of bounded observables around their expectation with respect to the Gross-Pitaevskii dynamics. 
\end{abstract}

\section{Introduction}

In the Gross-Pitaevskii regime, we consider systems of $N$ bosons confined by an external field in a volume of order one (after appropriate choice of the length unit) and interacting through a repulsive potential with small effective range of the order $1/N$. The corresponding Hamilton operator is given by 
\begin{equation}\label{eq:trapped} H_N^\text{trap} = \sum_{j=1}^N \left[ -\Delta_{x_j} + V_\text{ext} (x_j) \right]   + \sum_{i<j}^N N^2 V (N (x_i - x_j)) \end{equation} 
with $V_\text{ext} (x) \to \infty$ as $|x| \to \infty$ and $V \geq 0$ compactly supported. According to bosonic statistics, $H_N^\text{trap}$ acts as a self-adjoint operator on $L^2_s (\bR^{3N})$, the subspace of $L^2 (\bR^{3N})$ consisting of wave functions that are symmetric with respect to permutations of the $N$ particles. In \cite{LSY,NRS} it was proven that, to leading order in $N$, the ground state energy $E_N^\text{trap}$ of (\ref{eq:trapped}) satisfies  
\begin{equation}\label{eq:conv1}  \lim_{N \to \infty} \frac{E_N^\text{trap}}{N} = \min_{\ph \in L^2 (\bR^3) : \| \ph \|_2 =1} \cE_\text{GP} (\ph) \end{equation} 
with the Gross-Pitaevskii energy functional 
\begin{equation}\label{eq:EGP} \cE_\text{GP} (\ph) = \int \left[ |\nabla \ph |^2 + V_\text{ext} |\ph|^2 + 4 \pi \frak{a} |\ph|^4 \right] dx\,. \end{equation} 
Here $\frak{a} > 0$ denotes the scattering length of the interaction potential $V$, which is defined through the solution of the zero-energy scattering equation 
\[ \left[ -\Delta + \frac{1}{2} V \right] f = 0 \]
with the boundary condition $f(x) \to 1$, as $|x| \to \infty$, by requiring that $f(x) = 1 - \frak{a}/ |x|$ outside the support of $V$.

Let $\ph_\text{GP} \in L^2 (\bR^3)$ denote the (unique, up to a phase) normalized minimizer of (\ref{eq:EGP}). It turns out that the ground state vector of (\ref{eq:trapped}) and, in fact, every sequence of approximate ground states, exhibit complete Bose-Einstein condensation in the one-particle state $\ph_\text{GP}$. In other words, let us consider a normalized sequence $\psi_N \in L^2_s (\bR^{3N})$ satisfying 
\[ \frac{1}{N} \langle \psi_N, H_N^\text{trap} \psi_N \rangle \to \cE_\text{GP} (\ph_\text{GP}) \]
as $N \to \infty$ (ie. $\psi_N$ is a sequence of approximate ground states). Let $\gamma_N$ denote the one-particle reduced density matrix associated with $\psi_N$, which is defined as the non-negative operator on $L^2 (\bR^3)$ with the integral kernel 
\[ \gamma_N (x;y) = \int dx_2 \dots dx_N \, \overline{\psi}_N (x , x_2, \dots , x_N) \psi_N (y, x_2, \dots, x_N) \, , \]
normalized so that $\tr \, \gamma_N = 1$. Then, as first proven in \cite{LS1,LS2,NRS}, 
\begin{equation}\label{eq:BEC} \lim_{N \to \infty} \langle \ph_\text{GP} , \gamma_N \ph_\text{GP} \rangle = 1\,. \end{equation} 
The convergence (\ref{eq:BEC}) implies that, in the states $\psi_N$, the fraction of particles orthogonal to $\ph_\text{GP}$ vanishes, in the limit $N \to \infty$. 

Recently, the estimates (\ref{eq:conv1}), (\ref{eq:BEC}) have been improved in \cite{NNRT,BSS,NT,BSS1} for integrable interaction potentials, through a rigorous version of Bogoliubov theory. In the translation invariant setting, these improvements have been previously achieved in \cite{BBCS4,BBCS3}. This approach determines the ground state energy $E_N^\text{trap}$ of (\ref{eq:trapped}), up to errors vanishing as $N \to \infty$. Moreover, it gives precise information on the low-energy excitation spectrum of (\ref{eq:trapped}) and on the depletion of the Bose-Einstein condensate (in particular, it shows that the rate of convergence in (\ref{eq:BEC}) is proportional to $1/N$) and it provides a norm-approximation for the ground state vector.  It is interesting to remark that, in the Gross-Pitaevskii regime, such precise estimates on the spectrum of the Hamiltonian cannot be obtained restricting the attention on quasi-free states. Instead, it is important to take into account corrections that can be described by the action of a unitary operator on the Fock space of excitations of the condensate, given by the exponential of a cubic expression in creation and annihilation operators. 
For related recent results concerning equilibrium properties of Bose gases in the (translation invariant) Gross-Pitaevskii regime, beyond the Gross-Pitaevskii regime and in the thermodynamic limit, see \cite{HST, ABS, BCaS, BCOPS1,BCOPS2, HHNST, BoS, CCS1, CCS2, FS1,FS2, BaCS}.

Gross-Pitaevskii theory is not only useful to predict the ground state energy of Bose gases described by the Hamilton operator (\ref{eq:trapped}). It can also be used to approximate their time-evolution. From the point of view of physics, it is relevant to study the dynamics of an equilibrium state of (\ref{eq:trapped}), after the external fields are switched off (so that the system is no longer at equilibrium). At (or very close to) zero temperature, it is therefore interesting to study the solution of the time-dependent Schr\"odinger equation 
\begin{equation}\label{eq:schro}  i \partial_t \psi_{N,t} = H_N \psi_{N,t} \end{equation} 
with the translation invariant Hamiltonian (obtained after turning off the trap) 
\begin{equation}\label{eq:HN} 
H_N = \sum_{j=1}^N -\Delta_{x_j} + \sum_{i<j}^N N^2 V(N (x_i - x_j)) 
\end{equation}
for initial data $\psi_{N,0}$ approximating the ground state of (\ref{eq:trapped}). In \cite{ESY1,ESY2,ESY3,ESY4}, it was first proven that the time-evolution $\psi_{N,t}$ of an initial data $\psi_{N,0}$ exhibiting Bose-Einstein condensate in a one-particle state $\ph \in L^2 (\bR^3)$ still exhibits Bose-Einstein condensation, in a new one-particle state $\ph_t$, given by the solution of the nonlinear time-dependent Gross-Pitaevskii equation 
\begin{equation}\label{eq:GPtd}  i\partial_t \ph_t = -\Delta \ph_t+ 8\pi \frak{a} |\ph_t|^2 \ph_t \end{equation} 
with the initial data $\ph_{t=0} = \ph$. More precisely, denoting by $\gamma_{N,t}$ the one-particle reduced density associated with the solution $\psi_{N,t} \in L^2_s (\bR^{3N})$ of the Schr\"odinger equation (\ref{eq:schro}), it turns out that 
\begin{equation}\label{eq:BECt} \lim_{N \to \infty} \langle \ph_t, \gamma_{N,t} \ph_t \rangle = 1 \end{equation} 
for any fixed $t \in \bR$, if (\ref{eq:BECt}) holds true at time $t=0$. Analogous stability results have been later established in \cite{P,BDS}. In \cite{BS}, the convergence (\ref{eq:BECt}) is shown to hold with the optimal rate $1/N$, for every fixed $t \in \bR$. It is easy to check that (\ref{eq:BECt}) implies the convergence $\gamma_{N,t} \to |\ph_t \rangle \langle \ph_t|$ in the trace-class topology (and also the convergence of the $k$-particle reduced density matrix associated with $\psi_{N,t}$ towards the product $|\ph_t \rangle \langle \ph_t|^{\otimes k}$, for every fixed $k \in \bN$). However, (\ref{eq:BECt}) does not provide an approximation for the many-body wave function $\psi_{N,t}$ in the strong $L^2 (\bR^{3N})$ topology. To obtain a norm-approximation, it is not enough to approximate the evolution of the condensate. It is instead crucial to take into account the evolution of its orthogonal excitations. 

Norm-approximations for many-body dynamics have been derived in the mean-field setting, where particles are initially trapped in a volume of order one and interact weakly through a potential whose range is comparable with the size of the trap (so that every particle interacts effectively with all other particles in the system). In this case, as shown in \cite{Hepp,GV,RoS,GMM,GMM2,XC,LNS}, the solution of the many-body Schr\"odinger equation can be approximated, after removing the condensate wave function (whose evolution is described here by the nonlinear Hartree equation, see also \cite{Spohn,EY,BGM,FKP,FKS}), by a unitary dynamics on the Fock space of excitations, with a generator quadratic in creation and annihilation operators, acting as a family of time-dependent Bogoliubov transformations (an approximation to arbitrary precision has been recently obtained in \cite{BPPS}). A similar norm-approximation has been derived in \cite{BCS,BNNS} for the many-body evolution generated by Hamilton operators having the form 
\begin{equation}\label{eq:Hbeta} H^\beta_N = \sum_{j=1}^N -\Delta_{x_j} + \frac{1}{N} \sum_{i<j}^N N^{3\beta} V (N^\beta (x_i - x_j)) \end{equation} 
with $\beta \in (0;1)$, interpolating between the mean-field and the Gross-Pitaevskii scaling (analogous results have been also obtained in \cite{GM,GM1,NN1} for $\beta <1/3$, and in \cite{K,NN}, for $\beta < 1/2$; higher order estimates have been derived, for sufficiently small $\beta > 0$, in \cite{BPaPS}).  To achieve this goal, it was important to combine the unitary dynamics with quadratic generator, describing the evolution of excitations on macroscopic scales, with an additional Bogoliubov transformation generating the correct microscopic correlation structure. 

In the present paper, we prove a norm-approximation for the many-body dynamics generated by the Hamilton operator (\ref{eq:HN}), in the Gross-Pitaevskii regime (ie. with   $\beta =1$ in (\ref{eq:Hbeta})). Compared with the techniques developed in \cite{BCS,BNNS} for $\beta < 1$, there is an important difference in the construction of the approximating wave function. In fact, as already observed in  \cite{BBCS4,NT,BSS1} in the time-independent setting, for $\beta =1$ microscopic correlations cannot be precisely modelled only through a Bogoliubov transformation; they require instead an additional unitary conjugation with a phase cubic in creation and annihilation operators. This makes our analysis significantly more involved.  While the inclusion of the cubic phase is crucial to compare the generators of the full many-body evolution and of the quadratic dynamics (and thus to establish convergence for the corresponding evolutions), at the end it does not substantially change the $L^2 (\bR^{3N})$-norm of the approximation and it can therefore be removed, providing a quasi-free norm-approximation to the many-body evolution, similar to those obtained in \cite{BCS,BNNS} for $\beta \in (0;1)$.

\medskip

{\it Acknowledgment.} We gratefully acknowledge partial support from the Swiss National Science Foundation through the Grant ``Dynamical and energetic properties of Bose-Einstein condensates'', from the European Research Council through the ERC-AdG CLaQS and from the NCCR SwissMAP. C.C. warmly acknowledges the GNFM Gruppo Nazionale per la Fisica Matematica - INDAM.

\section{Setting and Main Results} 

We aim at approximating the solution $\psi_{N,t}$ of the many-body Schr\"odinger equation (\ref{eq:schro}), for a class of initial data  exhibiting complete Bose-Einstein condensation in a normalized one-particle wave function $\ph \in L^2 (\bR^3)$. In view of  (\ref{eq:BEC}), from the point of view of physics it is interesting to choose $\ph$ as the minimizer of a Gross-Pitaevskii functional of the form (\ref{eq:EGP}). Here, we will keep the choice of $\ph$ open, requiring only sufficient regularity. 

First of all, we need to approximate the evolution of the condensate. While (\ref{eq:GPtd}) provides a good approximation at the level of reduced density matrices,  to derive a norm-approximation it is more convenient to consider a slightly modified, $N$-dependent, nonlinear equation. In the modified equation, the  interaction potential appearing in the Hamilton operator (\ref{eq:HN}) is corrected, to take into account correlations among particles. In order to describe correlations, we fix $\ell \in (0;1/2)$ and we consider the ground state solution of the Neumann problem
\begin{equation}\label{eq:scatl} \left[ -\Delta + \frac{1}{2} V \right] f_{\ell} = \lambda_{\ell} f_\ell \end{equation}
on the  ball $|x| \leq N\ell$ (we omit here the $N$-dependence in the notation for $f_\ell$ and for $\lambda_\ell$; notice that $\lambda_\ell$ scales as $N^{-3}$), with the normalization $f_\ell (x) = 1$ for $|x| = N \ell$. We extend $f_\ell$ to $\bR^3$, setting $f_\ell (x) = 1$ for all $|x| > N\ell$ and we also introduce the notation $w_\ell (x) = 1 - f_\ell (x)$. To describe correlations created by the rescaled interaction appearing in (\ref{eq:trapped}) and in  (\ref{eq:HN}), we will use the functions $f_{N,\ell} (x) = f_\ell (Nx)$, $w_{N,\ell} (x) = w_\ell (Nx) = 1- f_{N,\ell} (x)$. By scaling, we observe that 
\be\label{eq:scatlN} \left[ -\Delta + \frac{N^2}{2} V (N \cdot ) \right] f_{N,\ell} = N^2 \lambda_\ell f_{N,\ell} \ee
on the ball $|x| \leq \ell$. Some important properties of $\lambda_\ell, f_\ell, w_\ell$ are collected in Lemma \ref{3.0.sceqlemma} in Appendix \ref{app:eta}.

With $f_{N,\ell}$, we can now define the condensate wave function at time $t \in \bR$ as the solution $\wt\ph_t$ of the modified Gross-Pitaevskii equation 
	 \begin{equation}
		\label{eq:GPmod}
		i\partial_t\wt{\ph}_t = -\D\wt{\ph}_t + (N^3V(N\cdot)f_\ell(N\cdot)*|\wt{\ph}_t|^2)\wt{\ph}_t\,,
	\end{equation}
	with initial data $\wt{\ph}_{t=0}=\ph$. As discussed in Lemma \ref{3.0.sceqlemma}, we have
\[ \left|  \int  N^3 V(N x) f_\ell (Nx) dx - 8\pi \aa   \right|  = \left|  \int  V(x) f_\ell (x) dx - 8\pi \aa   \right| \leq \frac{C \aa^2}{\ell N} \, . \]
It is therefore easy to check that, as $N \to \infty$, $\wt{\ph}_t$ converges to the solution of the limiting Gross-Pitaevskii equation (\ref{eq:GPtd}) (with the same initial data). This convergence is part of the statement of 
Prop. \ref{prop:propertiesphit}, in Appendix \ref{app:eta}, where we also collect some standard properties of the solutions of (\ref{eq:GPtd}) and of (\ref{eq:GPmod}) which will be used throughout the paper.

As explained in the introduction, to obtain a norm-approximation for the time-evolution it is not enough to approximate the evolution of the condensate; we also need to take into account its excitations. To this end, it is convenient to factor out the (time-dependent) condensate wave function introducing, for every $t \in \bR$, the unitary map $U_{N,t} : L^2_s (\bR^{3N}) \to \cF^{\leq N}_{\perp \wt{\ph}_t}$ into the Fock space of excitations
\begin{equation*}\label{eq:FS}
 \cF^{\leq N}_{\perp \wt{\ph}_t} = \bigoplus_{n=0}^N L^2_{\perp \wt{\ph}_t} (\bR^3)^{\otimes_s n} \end{equation*} 
 over the orthogonal complement $L^2_{\perp \wt{\ph}_t} (\bR^3)$ of $\wt{\ph}_t$. The map $U_{N,t}$ is defined so that $U_{N,t} \psi_{N,t} = \{\alpha_{N,t}^{(0)},\dots, \alpha_{N,t}^{(N)}\}$, corresponding to the unique decomposition 
\[\psi_{N,t} = \sum_{n=0}^N\alpha_{N,t}^{(n)}\otimes_s\wt\ph_t^{\otimes(N-n)},\]
where $\alpha_{N,t}^{(n)} \in L^2_{\perp\wt\ph_t}(\bR^3)^{\otimes_sn}$ is symmetric with respect to permutation and orthogonal to $\wt\ph_t$ in each coordinate. Denoting by $a(f), a^*(g)$ the usual creation and annihilation operators, the action of $U_{N,t}$ is characterized (see \cite{LNSS,BS}) by the rules 
	\begin{equation}\label{eq:U-rules}
		\begin{split} 
			U_{N,t} \, a^*(\wt\ph_t) a(\wt\ph_t) \, U_{N,t}^* &= N - \cN   \\  
			U_{N,t} \, a^*(f) a(\wt\ph_t) \, U_{N,t}^* &= a^*(f) \sqrt{N-\cN} =: \sqrt{N} b(f)^* \\		
			U_{N,t} \, a^*(\wt\ph_t) a(g) \, U_{N,t}^* &= \sqrt{N-\cN} \, a(g) =: \sqrt{N} b(g) \\
			U_{N,t} \, a^*(f) a(g) \, U_{N,t}^* &= a^*(f) a(g) \,,   \end{split}
	\end{equation} 
	where $\cN  = \int a^*_x a_x \, dx $ is the number of particles operator, and $b^*(f), b(g)$ are modified creation and annihilation operators satisfying the commutation relations 
	\begin{equation}\label{eq:bcomm} \begin{split} [ b(f), b^*(g) ] &= \Big( 1 - \frac{\cN}{N} \Big) \langle f,g\rangle - \frac{1}{N} a^*(f) a(g), \quad
		 [ b(f), b(g) ] = [b^*(f) , b^*(g)] = 0 \,,
\end{split} \end{equation} 
for all $f, g \in L^2_{\perp\wt\ph_t}(\bR^3)$. Denoting by $b_x, b_x^*, a_x, a_x^*$ the corresponding operator valued distributions, we also find 
	\begin{equation} \label{eq:bcomm2} \begin{split} 
	[b_x, b^*_y ] &= \delta (x-y) \Big( 1- \frac{\cN}{N} \Big) - \frac{1}{N} a^*_y a_x ,  \\
			[b_x, a^*_y a_z] &= \d(x-y) b_z, \,\qquad [b^*_x, a^*_y a_z] = - \d(y-z) b^*_x\,,
	\end{split} \end{equation}
for all $x,y,z \in \bR^3$. 

After factoring out the condensate with the unitary operator $U_{N,t}$, we need to approximate the evolution of its orthogonal excitations in $\cF^{\leq N}_{\perp \wt{\ph}_t}$. Here, we need to distinguish between microscopic excitations, varying on small length scales between $1/N$ and $\ell$ (which is chosen of order one) and macroscopic excitations, varying on scales of order one. Let us first worry about the microscopic excitations, characterising all low-energy states. It is natural to include them on the initial data and to propagate them along the time-evolution. These excitations only depend on time through the time-dependence of the condensate wave function $\wt{\ph}_t$. We are going to describe them through a (generalized) Bogoliubov transformation. 

We define the integral kernel 
	\begin{equation}\label{eq:defk}
		k_{t}(x,y) = -Nw_\ell(N(x-y))\wt{\ph}_t(x)\wt{\ph}_t(y)\,.
	\end{equation}
With Lemma \ref{3.0.sceqlemma}, we find $k_{t} \in L^2(\bR^3\times \bR^3)$ (with bounded norm, uniform in $N$). Hence, $k_t$ defines a Hilbert-Schmidt operator on $L^2 (\bR^3)$, which we will again denote by $k_t$. To obtain a Bogoliubov transformation acting on the Hilbert space $\cF_{\perp\wt\ph_t}^{\leq N}$, defined on the orthogonal complement of the condensate wave function $\wt{\ph}_t$, we set $\wt{q}_t = 1 - |\wt{\ph}_t\rangle\langle \wt{\ph}_t|$ and 
\begin{equation}
		\label{eq:defeta}
		\eta_t = (\wt{q}_t \otimes \wt{q}_t) k_t	\,.
	\end{equation}
We also denote $\mu_t = \eta_t - k_t$. With $\eta_t$, we define the antisymmetric operator
\begin{equation}\label{eq:defB}
	B_t = \frac 12 \int \big[\eta_t(x;y) b_x^*b_y^* -\overline{\eta}_t(x;y)b_xb_y\big] dx dy 
\end{equation}
and we consider the unitary operator $e^{B_t}$ on $\cF^{\leq N}_{\perp \wt{\ph}_t}$. An important consequence of the bound $\| \eta_t \|_2 \leq \| k_t \|_2 \leq C$, uniformly in $N \in \bN$ and $t \in \bR$, is the fact that $e^{B_t}$ does not substantially change the number of excitations. The proof of the following lemma can be found, for example, in \cite[Lemma 3.1]{BS}. 
\begin{lemma}
	\label{lm:gron-B}
	Let $B_t$ be the anti-symmetric operator defined in \eqref{eq:defB}. Then for every $n \in \bZ$ there exists a constant $C>0$ (that depends only on $\|\eta_t\|$) such that 
	\[ e^{-B_t}(\cN+1)^n e^{B_t} \leq C(\cN+1)^n\]
	as an operator inequality on $\cF^{\leq N}$.
\end{lemma}

On states with few excitations (on which the commutation relations (\ref{eq:bcomm}) are almost canonical), $e^{B_t}$ approximately acts as a Bogoliubov transformation. In fact, we can write 
\begin{equation}\label{eq:defd} e^{B_t} b (f) e^{-B_t} = b (\gamma_t (f)) + b^* (\sigma_t (\bar{f})) + d_t (f) \end{equation} 
where we introduced the notation  
\begin{equation}\label{eq:defcoshsinh}
 \gamma_t = \cosh(\eta_t) = \sum_{j=0}^\infty \frac{(\eta_t\bar\eta_t)^j}{(2j)!}\,;\hspace{2cm}  \sigma_t = \sinh(\eta_t)=\sum_{j=0}^\infty \frac{(\eta_t\bar\eta_t)^j\eta_t}{(2j+1)!}
\end{equation} 
and where, from \cite[Lemma 3.4]{BBCS3}, we have the bounds (denoting with $d_x, d_x^*$ the operator valued distributions associated with the operators $d, d^*$) 
\begin{equation}\label{eq:d-bds} 
\begin{split}
\| (\cN+1)^k d^\sharp_t (f) \xi \| &\leq \frac{C}{N} \|f\| \| (\cN+1)^{k+3/2} \xi \| \\
\| (\cN + 1)^{n/2} a_y d_x \xi \| &\leq \frac{C}{N} \, \Big[  (\| \eta_t (x; \cdot) \|_2 \| \eta_t (\cdot ; y) \|_2  + |\eta_t (x;y)|)  \| (\cN +1)^{(n+2)/2}  \xi \|  \\  &\hspace{1cm} +  \| \eta_t (\cdot ; y) \|_2  \| a_x (\cN+1)^{(n+1)/2} \xi \| \\ &\hspace{1cm}  +  \| \eta_t (x; \cdot) \|_2   \|a_y (\cN + 1)^{(n+3)/2} \xi \| +  \| a_x a_y (\cN +1)^{(n+2)/2}  \xi \|   \, \Big] \\
		\| (\cN + 1)^{n/2} d_x d_y \xi \| 
			   &\leq \frac C {N^2} \Big[ (\| \eta_t (x; \cdot) \|_2 \| \eta_t (\cdot ; y) \|_2 +  |\eta_t (x;y)| ) \| (\cN + 1)^{(n+4)/2}  \xi \| \\ &\hspace{1cm} + \| \eta_t (\cdot ; y) \|_2  \| {a}_x (\cN + 1)^{(n+5)/2} \xi \| \\  &\hspace{1cm}  + \| \eta_t (x; \cdot) \|_2  \|{a}_y (\cN + 1)^{(n+5)/2} \xi \|+ \|{a}_x {a}_y (\cN +  1)^{(n+4)/2} \xi \| \; \Big] \,, \end{split} 
 \end{equation}  
 for $\sharp = *, \cdot$, for all $f \in L^2 (\bR^3)$, $x,y \in \bR^3$ and $n \in \bZ$. Some bounds on the operators $\eta_t, \mu_t, \gamma_t, \sigma_t$ are collected in Lemma \ref{lm:propeta}, in Appendix \ref{app:eta}. 

We are interested in the time-evolution of initial data having the form 
\begin{equation}\label{eq:psi-in} \psi_N = U_{N,0}^* e^{B_0} \xi_N, \end{equation} 
under appropriate conditions on the excitation vector $\xi_N \in \cF_{\perp \ph}^{\leq N}$ (we will make assumptions on moments of number of particles and kinetic energy operators, in the state $\xi_N$). This 
allows us to consider initial data which are expected to describe the ground state vector of trapped Hamiltonian like (\ref{eq:trapped}) (see the Remark after Theorem \ref{th:main} for more details). 

We consider the many-body evolution $\psi_{N,t} = e^{-iH_N t} \psi_N$ of (\ref{eq:psi-in}), generated by the translation invariant Hamilton operator (\ref{eq:HN}). At time $t \not = 0$, we expect $\psi_{N,t}$ to have again approximately the form (\ref{eq:psi-in}), but now with $\ph$ replaced by the solution $\wt{\ph}_t$ of (\ref{eq:GPmod}). For this reason, we introduce the excitation vector $\xi_{N,t} \in \cF_{\perp \wt\ph_t}^{\leq N}$ requiring that 
\[ e^{-i H_N t} \psi_N = U_{N,t}^* e^{B_t} \xi_{N,t}\,. \]
In other words, $\xi_{N,t} = \bar{\cU}_N (t;0) \xi_N$, with the fluctuation dynamics 
\begin{equation}\label{eq:fluct-dyn} \bar{\cU}_N (t;s) = e^{-B_t} U_{N,t} e^{-i H_N (t-s)} U_{N,s}^* e^{B_s}\,.  \end{equation} 

While $e^{B_t}$ takes care of the microscopic excitations of the condensate, to derive a norm approximation for 
$e^{-i H_N t} \psi_N$ we still need to take into account the evolution of the macroscopic excitations. To reach this goal, we introduce a unitary dynamics $\cU_{2,N} (t;0)$, whose generator is quadratic in creation and annihilation operators (a time-dependent Bogoliubov transformation), approximating the fluctuation dynamics $\bar{\cU}_N (t;0)$ and providing therefore an approximation of $\xi_{N,t} = \bar{\cU}_N (t;0) \xi_N$. To this end, let us introduce the ``projected'' modified creation and annihilation operators 
\be\label{eq:bxxtildedef}
\bxxtilde^* = b^* (\wt{q}_{x}) = \bxxs - \overline{\ptx} b^* (\ptt) , \qquad \bxxtilde = b(\wt{q}_{x}) = \bxx - \ptx b(\ptt) 
\ee
where $q_x (y) = q_t (y,x) = \delta (x-y) - \overline{\wt{\ph}}_t (x) \wt{\ph}_t (y) $ is the kernel of the projection orthogonal to the condensate wave function. Then, we define the time-dependent self-adjoint operator \begin{equation}\label{eq:generatorapprox} 
	\begin{split}
		\cJ_{2,N}(t)=&\, \cJ_{2,N}^K(t) + \cJ_{2,N}^V(t) \\
		& + N^3 \l_\ell\int dxdy\, \c_\ell(x-y)\wt\ph_t(x)\wt\ph_t(y)\tilde{b}^*_x\tilde{b}^*_y +\hc\\
		&+\frac12\int dxdy\, Nw_{N,\ell}(x-y)[\D\wt\ph_t(x)\wt\ph_t(y) +\wt\ph_t(x)\D\wt\ph_t(y)]\tilde{b}^*_x\tilde{b}^*_y +\hc\\
		&+ \int dxdy\, N\nabla w_{N,\ell}(x-y)[\nabla\wt\ph_t(x)\wt\ph_t(y) - \wt\ph_t(x)\nabla\wt\ph_t(y)]\tilde{b}^*_x\tilde{b}^*_y +\hc\\
		&-\int_0^1ds \int dxdy\, \dot{\eta}_t(x,y) \big[\tilde{b}^*(\g_x^{(s)})\tilde{b}^*(\g_y^{(s)})+\tilde{b}^*(\g_x^{(s)})\tilde{b}(\s_y^{(s)}) \\
		&\hskip4cm+\tilde{b}(\s_x^{(s)})\tilde{b}(\s_y^{(s)})+ \tilde{b}^*(\g_y^{(s)})\tilde{b}(\s_x^{(s)})  
		\big] +\hc \\
		&+\int dx N^3 (V_N \fln) * \abs{\ptt}^2 (x) (\axxs\axx - \bxxstilde\bxxtilde)\\
		\eqqcolon &\;\cK + \int dx N^3 (V_N \fln) * \abs{\ptt}^2 (x) \axxs\axx \\
        &+ \intxy \Big(G_t (x,y) \bxxstilde\byytilde + H_t (x,y)\bxxstilde\byystilde + \overline{H}_t (x,y) \bxxtilde\byytilde\Big)\,,
\end{split}\end{equation}
where 
\begin{equation}\label{eq:J2NK}
\begin{split}
		\cJ_{2,N}^K(t) =&\,\cK 
		+ \int dx \big[\tilde{b}^*(-\D_x p_x)\tilde{b}_x
		+ \frac{1}{2} \tilde{b}^*(\nabla_x p_x) \tilde{b}(\nabla_x p_x) 
		+ \bxxstilde \tl b^*(-\D_x \mu_x)\\
		&+ \tilde{b}^*(-\D_xp_x)\tilde{b}^*(\eta_x) 
		+ \tilde{b}^*(p_x)\tilde{b}^*(-\D_x r_x) 
		+ \tilde{b}_x^* \tilde{b}^*(-\D_x r_x)\\
		&+\frac{1}{2} \tilde{b}^*( \nabla_x \eta_x) \tilde{b}(\nabla_x \eta_x)
		+ \tilde{b}^*(\eta_x)\tilde{b}(-\D_x r_x)
		+ \frac{1}{2} \tilde{b}^*(\nabla_x r_x)\tilde{b}(\nabla_x r_x) +\hc \big]
\end{split}\end{equation}
and 
\begin{equation} \label{eq:cJ2NV}
	\begin{split}
		\cJ_{2,N}^V(t) =&\,
		\frac 12 \int dxdy\, N^3(V_Nf_{N,\ell})(x-y)\wt{\ph}_t(x)\wt{\ph}_t(y)\big[\tilde{b}^*(p_{x})\tilde{b}^*_{y} + \tilde{b}^*(\g_{x})\tilde{b}^*(p_{y})\big] + \hc \\
		&+\frac 12 \int dxdy\, N^3(V_Nf_{N,\ell})(x-y)\wt{\ph}_t(x)\wt{\ph}_t(y)\big[ \tilde{b}^*(\g_{y})\tilde{b}(\s_{x})  \\
		&\hspace{6cm}+ \tilde{b}^*(\g_{x})\tilde{b}(\s_{y})+ \tilde{b}(\s_{x})\tilde{b}(\s_{y}) \big]+ \hc\\
		&+\int dx\, (N^3(V_N f_{N,\ell})*|\wt{\ph}_t|^2)(x)\big(\tilde{b}^*(\gxt)\tilde{b}(\gxt) +\tilde{b}(\sxt)\tilde{b}(\gxt)  \\
		&\hspace{6cm}+ \tilde{b}^*(\gxt)\tilde{b}^*(\sxt)+ \tilde{b}^*(\sxt)\tilde{b}(\sxt) \big)\\
		& +\int dxdy\, N^3 (V_N f_{N,\ell})(x-y)\wt{\ph}_t(x)\overline{\wt{\ph}_t(y)}\big(\tilde{b}^*(\gxt)\tilde{b}(\g_{y}) +\tilde{b}(\sxt)\tilde{b}(\g_{y})\\
		&\hspace{6cm}+ \tilde{b}^*(\gxt)\tilde{b}^*(\s_{y})  +  \tilde{b}^*(\syt)\tilde{b}(\sxt)   \big)\,. \\
\end{split}\end{equation}

Here we used the kernels $\eta_t, \mu_t, \gamma_t, \sigma_t$ introduced in (\ref{eq:defeta}), (\ref{eq:defcoshsinh}) and, additionally, we defined $p_t = \gamma_t - 1$ and $r_t = \sigma_t - \eta_t$. Moreover, for $s \in [0;1]$, $\gamma_t^{(s)}$ and $\sigma_t^{(s)}$ are defined as $\gamma_t , \sigma_t$, but with $\eta_t$ replaced by $s \eta_t$. Furthermore, we used the notation $f_x (y) = f_t (y,x)$ for kernels of operators acting on $L^2(\bR^3)$ (here we drop the label $t$, to keep the notation as light as possible). Also, we set $V_N (x) = V(Nx)$ and we used the notation \[ \cK = \int dx \, \nabla_x a_x^* \nabla_x a_x \] for the kinetic energy operator. The self-adjoint operator $	\cJ_{2,N}(t)$ generates a two-parameter family of unitary transformations $\cU_{2,N}$, satisfying the equation
\be\label{eq:evolutionapprox}
i\partial_t \, \cU_{2,N}(t;s)= \cJ_{2,N}(t) \, \cU_{2,N}(t;s)
\ee
with $\cU_{2,N} (s;s) = 1$ for all $s \in \bR$ (the well-posedness of (\ref{eq:evolutionapprox}) is part of the statement of the next theorem; it will be established in Prop. \ref{prop:wellposedness}). 

In our first main theorem, we show that the Bogoliubov dynamics $\cU_{2,N}$ can be used to describe the evolution of macroscopic excitations of the condensate, providing a norm-approximation for the solution of the many-body Schr\"odinger equation. 
\begin{theorem}\label{th:main}
Let $V\in L^3(\RRR^3)$ be non-negative, spherically symmetric and compactly supported. Let $\varphi\in H^6(\RRR^3)$. Let $\eta_t$ be defined as in (\ref{eq:defeta}), with parameter $\ell > 0$ small enough. Then (\ref{eq:evolutionapprox}) defines a unique 2-parameter strongly continuous unitary group $\cU_{2,N}$, with $\cU_{2,N} (t;s) : \cF^{\leq N}_{\perp \wt{\ph}_s} \to \cF^{\leq N}_{\perp \wt{\ph}_t}$. Moreover, let $B_t$ be defined as in (\ref{eq:defB}) and  
\begin{equation}
	\label{eq:cubicphasek1}
	\begin{split}
		\k_{N}(t) =& \frac{N}{2} \la \ptt, [N^3\VN(1-2\fln)*|\ptt|^2] \ptt \ra
		-\frac{1}{2} \la \ptt, (N^3 V_Nf_{N,\ell}\ast \vert \ptt^2\vert ) \ptt\ra \\
		&+ \frac{1}{2}  \intxy N^2 \VNxy \vert \la \g_y,\s_x \ra\vert^2  \\
		&+\int dx\, (N^3\VN f_{N,\ell} *|\wt{\ph}_t|^2)(x)
		\langle \sxt, \sxt\rangle \\
		&+ \int dxdy\, N^3\VN f_{N,\ell} (x-y)\wt{\ph}_t(x)\overline{\wt{\ph}_t(y)} \langle \sxt, \s_y\rangle \\
		&+\frac 12 \int dxdy\, N^3\VN(x-y)\wt{\ph}_t(x)\wt{\ph}_t(y)\big \langle\s_x,\g_y\rangle +\hc +\|\nabla_1\s_t\|^2 \, \\
		&- \int_0^1ds \int dxdy\, \dot{\eta}_t(x,y)\la\s_x^{(s)},\g_y^{(s)}\ra\, +\hc
\end{split}\end{equation}
Consider a sequence of normalized initial data $\psi_N\in L^2(\RRR^3)^{\otimes_s N}$, with excitation vectors 
\be \label{eq:init-main}
\xi_N = e^{-B_0} U_{N,0} \psi_N
\ee
satisfying 
\begin{equation}\label{eq:init-bd} \la \x_N , (\cK^2 +\cN^6)\, \x_N \ra \leq C\,, \end{equation} 
uniformly in $N \in \bN$. Then there exist constants $C, c >0$ such that 
	\begin{equation}\label{eq:main} 
	    \norm{
	    e^{-i H_N t} \psi_N
	    - e^{-i\ \int_0^t \kappa_N (s) ds} U^*_{N,t} e^{B_t} \cU_{2,N}(t) \xi_N
	    }
	    \leq Ce^{c \expt}  N^{-1/8}
	\end{equation}
	for any $t \in \bR$, and $N \in \bN$ large enough. 
	\end{theorem}

{\it Remark.} From (\ref{eq:init-main}), we have $\psi_N = U_{N,0}^* e^{B_0} \xi_N$. Thus, (\ref{eq:main}) is equivalent to 
\begin{equation}\label{eq:main2} \| \bar{\cU}_N (t;0) \xi_N - e^{-i \int_0^t \kappa_N (s) ds} \, \cU_{2,N} (t;0) \xi_N \| \leq C e^{ce^{c|t|}} N^{-1/8} \end{equation}
with the fluctuation dynamics (\ref{eq:fluct-dyn}). In other words, the proof of Theorem \ref{th:main} reduces to the comparison of the evolution $\bar{\cU}_N$ with its quadratic approximation $\cU_{2,N}$. 

{\it Remark. } Observe that the choice $\xi_N = e^B \Omega$, with  
\[ B = \frac{1}{2} \int \big[ \tau (x;y) b_x^* b_y^* - \overline{\tau} (x;y) b_x b_y \big] dx dy \]
and with a kernel $\tau \in (q_0 \times q_0) H^2 (\bR^3 \times \bR^3)$ (with $q_0 = 1- |\ph \rangle \langle \ph|$ projecting orthogonally to the condensate wave function) is compatible with the condition (\ref{eq:init-bd}).
This allows us to consider initial many-body wave functions of the form $\psi_N = U_{N,0}^* e^{B_0} e^{B} \Omega$  which are expected to approximate (in the $L^2 (\bR^{3N})$ norm which is preserved over time) the ground state vector of Hamilton operators of the form (\ref{eq:trapped}), describing trapped Bose gases in the Gross-Pitaevskii regime. This fact has been proven in \cite{BBCS4} in the translation invariant setting, for a gas confined on the unit torus. Notice that, in \cite[Eq. (6.7)]{BBCS4}, the norm-approximation of the ground state vector includes also the unitary operator $e^A$, with $A$ cubic in creation and annihilation operators; in fact, this cubic phase can be removed, at the expense of another error of order smaller than $N^{-1/4}$ (arguing similarly as we do below, in the proof of Theorem \ref{th:main}, to show (\ref{eq:eA-insert})). If the trapping potential is sufficiently smooth, it is also easy to verify the condition $\ph \in H^6 (\bR^3)$, with $\ph$ minimizing the functional (\ref{eq:EGP}).

{\it Remark. } The double exponential in time on the r.h.s. of (\ref{eq:main}) is due to the fact that we use estimates on high Sobolev norms for the solution of the modified Gross-Pitaevskii equation (\ref{eq:GPmod}) and we apply a Gronwall argument for the propagation of many-body bounds. Assuming scattering for the solution of (\ref{eq:GPmod}), the dependence on $t$ could be substantially improved.

{\it Remark. } Theorem \ref{th:main} confirms a conjecture formulated in \cite[Sect. 10]{GM1}, for a different class of initial data (through the assumption (\ref{eq:init-bd}), we impose a correlation structure on the initial wave function, generated by $e^{B_0}$, which is absent from the initial data discussed in \cite{GM1}; it is not clear to us whether a norm approximation similar to (\ref{eq:main}) with an explicit rate can be obtained for ``flat'' data).

\medskip

In (\ref{eq:main}), after factoring out the evolving Bose-Einstein condensate and the microscopic correlation structure, we approximate the evolution of the macroscopic correlations by the Bogoliubov dynamics $\cU_{2,N} (t)$, which still depends on $N$. It is thus natural to ask whether $\cU_{2,N} (t)$ approaches a limiting, $N$-independent, quadratic evolution $\cU_{2,\infty} (t)$, as $N$ tends to infinity. To answer this question, we start by defining the pointwise limit of $N w_{N,\ell}(x)$, setting 
\begin{equation}
	\label{eq:defwio}
	w_{\io,\ell} (x) = 
	\begin{cases}
	 \aa\Big[\frac{1}{|x|} -\frac{3}{2\ell} + \frac{x^2}{2 \ell^3}\Big]\quad &|x|\leq \ell\\
	0 \quad &\text{otherwise}
	\end{cases} \, , 
\end{equation}
and the corresponding limiting integral kernel $k_{t,\io} \in L^2(\bR^3 \times \bR^3)$ by 
\begin{equation}
	\label{eq:limitk}
	k_{t,\io}(x;y) =  - w_{\io,\ell}(x-y)\ph_t(x)\ph_t(y)\,,
\end{equation}
where $\ph_t$ is the solution to the limiting ($N$-independent) Gross-Pitaevskii equation \eqref{eq:GPtd}. 
Similarly to (\ref{eq:defeta}), we define $\eta_{t,\io} = (q_t \otimes q_t) k_{t,\io}$, projecting along $q_t = 1 - |\ph_t \rangle \langle \ph_t |$, orthogonally to $\ph_t$ and $\mu_{t,\io} = \eta_{t,\io} - k_{t,\io}$. We also introduce the notation 
\be\label{eq:ginfty}
\g_{t,\io}= \cosh(\eta_{t,\io})= 1 + p_{t,\io} \textrm{ and } \s_{t,\io}= \sinh(\eta_{t,\io})= \eta_{t,\io} + r_{t,\io}.
\ee
With this notation,  we can introduce the generator of the limiting quadratic evolution, setting 
\begin{equation}\label{eq:generatorinfapprox}
	\begin{split}
		\cJ_{2,\io}(t)=&\, \cJ_{2,\io}^K(t) + \cJ_{2,\io}^V(t) \\
		& + \frac{3\aa}{\ell^3}\int dxdy\, \c_\ell(x-y)\ph_t(x)\ph_t(y)\dbtilde{a}^*_x\dbtilde{a}^*_y +\hc\\
		&+\frac12\int dxdy\, w_{\io,\ell}(x-y)[\D\ph_t(x)\ph_t(y) +\ph_t(x)\D\ph_t(y)]\dbtilde{a}^*_x\dbtilde{a}^*_y +\hc\\
		&+ \int dxdy\, \nabla w_{\io,\ell}(x-y)[\nabla\ph_t(x)\ph_t(y) - \ph_t(x)\nabla\ph_t(y)]\dbtilde{a}^*_x\dbtilde{a}^*_y +\hc\\
		&-\int_0^1ds \int dxdy\, 
		\dot{\eta}_{t, \io}(x,y) 
		\big[
		\dbtilde{a}^*(\g_{\io,x}^{(s)})\dbtilde{a}^*(\g_{\io,y}^{(s)})
		+\dbtilde{a}^*(\g_{\io,x}^{(s)})\dbtilde{a}(\s_{\io,y}^{(s)})
		 \\
		&\hskip3.5cm+\dbtilde{a}(\s_{\io,x}^{(s)})\dbtilde{a}(\s_{\io,y}^{(s)})+ \dbtilde{a}^*(\g_{\io,y}^{(s)})\dbtilde{a}(\s_{\io,x}^{(s)}) 
		\big] +\hc \\
		&+8\pi\aa \int dx \abs{\ph_t(x)}^2 (\axxs\axx-\dbtilde{a}_x^*\dbtilde{a}_x)
		\\ \eqqcolon &\; \cK + 8\pi\aa \int dx \abs{\ph_t(x)}^2 \axxs\axx \\
        &+ \intxy \Big(G_{t,\io} (x,y) \dbtilde{a}^*_x\dbtilde{a}_y 
        + H_{t,\io} (x,y) \dbtilde{a}^*_x\dbtilde{a}^*_y 
        + \overline{H}_{t,\io} (x,y) \dbtilde{a}_x\dbtilde{a}_y\Big)\\
\end{split}\end{equation}
where $\dbtilde{a}, \dbtilde{a}^*$ denote annihilation and creation operators, projected on $\cF_{\perp\ph_t}$ and where 
\begin{equation*}
\begin{split}
		\cJ_{2,\io}^K(t) =&\,\cK 
		+ \int dx \big[\dbtilde{a}^* (-\D_x p_{\io,x})\dbtilde{a}_x
		+ \frac{1}{2} \dbtilde{a}^*(\nabla_x p_{\io,x}) \dbtilde{a}(\nabla_x p_{\io,x}) 
		+ \dbtilde{a}^*_x \dbtilde{a}^*(-\D_x \mu_{\io,x})\\
		&+ \dbtilde{a}^*(-\D_x p_{\io,x})\dbtilde{a}^*(\eta_{\io,x}) 
		+ \dbtilde{a}^*(p_{\io,x})\dbtilde{a}^*(-\D_x r_{\io,x}) 
		+ \dbtilde{a}_x^* \dbtilde{a}^*(-\D_x r_{\io,x})\\
		&+\frac{1}{2} \dbtilde{a}^*( \nabla_x \eta_{\io,x}) \dbtilde{a}(\nabla_x \eta_{\io,x})
		+ \dbtilde{a}^*(\eta_{\io,x})\dbtilde{a}(-\D_x r_{\io,x})
		+ \frac{1}{2} \dbtilde{a}^*(\nabla_x r_{\io,x})\dbtilde{a}(\nabla_x r_{\io,x}) +\hc \big]
\end{split}\end{equation*}
and
\begin{equation}\label{eq:cJinftyV}
	\begin{split}
		\cJ_{2,\io}^V(t) =&\,
		4\pi\aa \int dx\, \ph_t(x)^2
		\big[\dbtilde{a}^*(p_{\io,x})\dbtilde{a}^*_{x} 
		    + \dbtilde{a}^*(\g_{\io,x})\dbtilde{a}^*(p_{y})\big] + \hc \\
		&+4\pi\aa \int dx\, \ph_t(x)^2
		\big[2 \dbtilde{a}^*(\g_{\io,x})\dbtilde{a}(\s_{\io,x})
		    + \dbtilde{a}(\s_{\io,x})\dbtilde{a}(\s_{\io,x}) \big]+ \hc\\
		&+ 16\pi\aa\int dx\, |\ph_t|^2(x)
		\big(\dbtilde{a}^*(\g_{\io,x})\dbtilde{a}(\g_{\io,x}) 
		    +\dbtilde{a}(\s_{\io,x})\dbtilde{a}(\g_{\io,x})\\
		&\hphantom{+ 16\pi\aa\int dx\, |\ph_t|^2(x)\big(}   + \dbtilde{a}^*(\g_{\io,x})\dbtilde{a}^*(\s_{\io,x})
		    + \dbtilde{a}^*(\s_{\io,x})\dbtilde{a}(\s_{\io,x}) \big)\,.
\end{split}\end{equation}
Here $\gamma^{(s)}_{t,\infty}, \sigma^{(s)}_{t,\infty}$ are defined like $\gamma_{t,\infty}, \sigma_{t,\infty}$, but with $\eta_{t,\io}$ replaced by $s \eta_{\io,t}$. The two-parameter unitary evolution generated by $\cJ_{2,\infty} (t)$ is denoted by $\cU_{2,\io} (t;s)$. It satisfies the Schr\"odinger equation 
\begin{equation}
	\label{eq:limitingdynamics}
	i \partial_t \cU_{2,\io}(t,s) = \cJ_{2,\io}(t) \, \cU_{2,\io}(t,s)\,,
\end{equation}
 with initial condition $\cU_{2,\io}(s,s) =1 $ for all $s \in \bR$ (the well-posedness of (\ref{eq:limitingdynamics}) is part of the statement of next theorem). Notice that $\cU_{2,\io} (t;s)$ maps the Fock space $\cF_{\perp \ph_s}$ into $\cF_{\perp \ph_t}$; here, there is no truncation on the number of particles. 
 \begin{theorem}
 	\label{th:limitingdynamics}
Let $V\in L^3(\RRR^3)$ be non-negative, spherically symmetric and compactly supported. Let $\varphi\in H^6(\RRR^3)$ and $\eta_{t,\infty}$ be defined as after \eqref{eq:limitk}. Then, there is a unique two-parameter strongly continuous unitary group $\cU_{2,\infty}$ satisfying (\ref{eq:limitingdynamics}), with $\cU_{2,\infty} (t;s) : \cF_{\perp \ph_s} \to \cF_{\perp \ph_t}$. For normalized $\xi \in\Ftru_{\perp\ph}$ satisfying
 	\begin{equation}\label{eq:exp-io}
 	\la \x , (\cK^2 +\cN^2)\, \x \ra < \infty \,,
 	\end{equation}
we find $C, c >0$ (only depending on the expectation (\ref{eq:exp-io})) such that 
 	\begin{equation} \label{eq:thm2}
 		\norm{
 	 \cU_{2,N} (t;0) \xi
 			-  \cU_{2,\io}(t;0) \xi 
 		}
 		\leq Ce^{c \expt}  N^{-1/4}
 	\end{equation}
 	for any $t \in \bR$, and $N \in \bN$ large enough. In (\ref{eq:thm2}), both $\cU_{2,N} (t;0) \xi$ and $\cU_{2,\io} (t;0) \xi$ are thought of as vectors in the full Fock space $\cF$, and $\| . \|$ denotes the norm in this space. 
 \end{theorem}

{\it Remark.} Since here we compare two quadratic evolutions, the condition (\ref{eq:exp-io}) is milder than the corresponding assumption (\ref{eq:init-bd}) in Theorem \ref{th:main}, where we compare the full many-body dynamics with a quadratic approximation. 

Because its generator is quadratic in creation and annihilation operators, the evolution $\cU_{2,\infty}$ acts as time-dependent Bogoliubov transformations. Thus, its action on annihilation and creation operators can be calculated explicitly. 
\begin{prop}
\label{prop:Bogtrasf}
Under the same assumptions as in Theorem \ref{th:main}, let $\cU_{2,\io}$ be the limiting quadratic evolution defined in Eq. \eqref{eq:limitingdynamics}. For every $t,s \in \bR$ there exists a bounded linear map  
\[
 \Theta(t;s) : L^2(\bR^3) \oplus L^2(\bR^3) \to L^2(\bR^3) \oplus L^2(\bR^3) 
\]
such that
\[
\cU_{2,\io}^*(t;s) A(f,g) \cU_{2,\io}(t;s) = A(\Theta(t;s) \, (f,g))
\]
for all $f,g \in L^2(\bR^3)$. Here, $A(f,g) = a(f)+a^*(\overline{g})$. The map $\Theta(t;s)$ satisfies 
\[
\Theta(t;s) \cJ = \cJ\Theta(t;s),\quad 
S = \Theta(t;s)^*S\Theta(t;s),
\]
where $\cJ = \begin{pmatrix}
0 & J \\
J & 0 
\end{pmatrix}$ with $J$ denoting complex conjugation on $L^2(\bR^3)$ and $S = \begin{pmatrix}
1 & 0 \\
0 & -1 
\end{pmatrix}$. It can be written as 
\begin{equation}\label{eq:ThetaUV}
\Theta (t;s) = \begin{pmatrix}
U(t,s) & JV(t,s)J \\
V(t,s) & JU(t,s)J
\end{pmatrix}
\end{equation} 
for bounded linear maps $U(t,s), V(t,s) : L^2(\bR^3) \to L^2(\bR^3)$ satisfying 
\[
U^*(t,s)U(t,s) -V^*(t,s)V(t,s) =1, \quad U^*(t,s)JV(t,s)J= V^*(t,s)JU(t,s)J\,.
\]
\end{prop}

\textit{Remark.} Differentiating the action of $\cU_{2,\infty}$ on $A(f,g)$ yields the differential equation $-i\partial_t \Theta(t,s)=\Theta(t,s) \cA(t)$ with
\[
\cA(t)= \begin{pmatrix}
-\Delta + 8\pi\aa \abs{\ph_t}^2 +q_t G_{t,\infty}q_t & -J(2 \overline{q_t H_{t,\infty}}q_t )J \\
2 \overline{q_t H_{t,\infty}}q_t & -J(-\Delta + 8\pi\aa \abs{\ph_t}^2 +q_t G_{t,\infty}q_t)J 
\end{pmatrix}
\]
This can be compared with \cite[Eq. (2.22)]{BAKS}, in the mean-field setting.

The proof follows from \cite[Theorem 2.2]{BAKS} and \cite[Prop.1.3]{Rademacher}; using the notation introduced in the last two lines of \eqref{eq:generatorinfapprox}, the bounds $\| G_{t,\infty} \|_\text{op}$ , $\| H_{t,\io} \|_2 < \infty$, which are assumed in \cite{Rademacher} follow from the analysis in the proof of Prop. \ref{prop:cJN-cJinfty}.

Using the approximation in terms of the Bogoliubov dynamics $\cU_{2,\infty}$, we can establish a central limit theorem for the evolution of initial data approximating ground states of the trapped Hamiltonian (\ref{eq:trapped}). Similar results have been obtained for the evolution generated by the Hamilton operator \eqref{eq:Hbeta} in  \cite{BAKS,BSS2} (if $\beta = 0$) and in \cite{Rademacher} (for $0 < \beta < 1$). In the time-independent setting, an analogous central limit theorem has been established in \cite{RS}, in the Gross-Pitaevskii regime.
\begin{theorem}
\label{th:clt}
Under the same assumptions as in Theorem \ref{th:main}, consider initial wave functions having the form \begin{equation*}\label{eq:clt-init} \psi_N=U_{N,0}^*e^{B_0} e^{B} \Omega , \end{equation*} where $B_0$ is defined as in \eqref{eq:defB}, and  \begin{equation}\label{eq:Btau} B = \frac{1}{2}  \int  \big[ \t (x,y)  b_x^* b_y^*  - \bar{\tau} (x;y) b_x b_y \big] dx dy  \end{equation} 
with kernel $\t\in (q_0 \otimes q_0) H^2(\bR^3\times\bR^3)$ (where $q_0 = 1 - |\ph \rangle \langle \ph|$ projects orthogonally to the condensate wave function $\ph$ used in $U_{N,0}^*$). For a bounded operator $O$ on $L^2 (\bR^3)$ we define 
\begin{equation}\label{eq:ONt}  \cO_{N,t} = \frac{1}{\sqrt{N}} \sum_{j=1}^N \big( O^{(j)} - \langle \ph_t, O \ph_t \rangle \big) \end{equation} 
where $O^{(j)} = 1 \otimes \dots \otimes O \otimes \dots \otimes 1$ is the operator $O$ acting only on the $j$-th particle. We set, with $U,V$ indicating the linear maps introduced in (\ref{eq:ThetaUV}), 
\begin{equation*}
\begin{split}
&h_{\infty, t} = \cosh(\eta_{t,\io}) q_{\ph_t} O \ph_t +\sinh(\eta_{t,\io})\overline{q_{\ph_t} O \ph_t} \\
&n_{t} = U(t,0) h_{\infty, t} + \overline{V(t,0)} \overline{h_{\infty, t}}\\
&f_{t}= \cosh(\t) n_{t}+ \sinh(\t)\overline{n_{t}}
\end{split}
\end{equation*}
Then there exists $c > 0$ and, for all $-\infty < a < b < \infty$, $C> 0$ such that, in the state $\psi_{N,t} = e^{-i H_N t} \psi_N$, 
\be \label{eq:CLT0} 
\abs{\bP_{\psi_{N,t}} (\cO_{N,t} \in [a, b]) - \bP (\frak{G}_t \in [a, b])}
\leq C \expexpt N^{-1/16}
\ee
for all $N$ large enough. Here $\frak{G}_t$ is a centered Gaussian with variance $\norm{f_t}^2$.
\end{theorem}

\textit{Remark.} The bound (\ref{eq:CLT0}) follows through standard arguments (see \cite[Proof of Corollary 1.2]{BSS2}) from an estimate of the form 
\begin{equation}\label{eq:exp-CLT2}
\begin{split}
\Big| \mathbb{E}_{\psi_{N,t}} g (\cO_{N,t}) &- \frac{1}{\sqrt{2 \pi} \| f_t \|} \int dx \, g(x) e^{-\frac{x^2}{2 \| f_t \|^2}} \Big| \\
&\leq  C e^{ce^{c|t|}}  \int |\hat{g} (s)| (N^{-1/8}+ N^{-1/2}|s|^3 \| O \|^3 + N^{-1} |s|^4 \| O \|^4) ds  
\end{split}
\end{equation} 
for the expectation of the random variable $g (\cO_{N,t})$, valid for any $g \in L^1 (\bR)$, with $(1+s^4) \hat{g} (s) \in L^1 (\bR)$. The proof of (\ref{eq:exp-CLT2}), which is based on the approximation (\ref{eq:main}) of the many body evolution, is given in Section \ref{sec:CLT}. Similarly to \cite{BSS2, Rademacher, RS}, we could also extend (\ref{eq:exp-CLT2}) to a multivariate central limit theorem, proving that expectations of products of observables of the form (\ref{eq:ONt}) approach a Gaussian limit, as $N \to \infty$.

\section{Fluctuation Dynamics} 

While the wave function $e^{-i\ \int_0^t \kappa_N (s) ds} U^*_{N,t} e^{B_t} \cU_{2,N}(t) \xi_N$ appearing in (\ref{eq:main}) provides a good norm-approximation for the full evolution $e^{-iH_N t} \psi_N = e^{-i H_N t} U_{N,0}^* e^{B_0} \xi_N$, the difference of their energy does not converge to zero, as $N \to \infty$. For this reason, it seems difficult to show (\ref{eq:main}) directly. To circumvent this problem, we introduce an alternative approximation for the many-body evolution, this time having the correct energy. At the end, we will show that the two approximations are close in norm. 

To define the new approximation of the many-body evolution, we will use a cubic phase. First, we define, as in \cite{NT}, a cutoff $\Th\colon \NNN \to \RRR $ in the number of particles, setting 
\begin{equation*}
\Th(n)= \begin{cases}
	1  & n \leq \frac{M}{2}+10 \\
	\frac{1}{2}\left( \frac{4n-3M}{40-M}+ 1\right) & \frac{M}{2}+10\leq n\leq  M-10\\
	0 &n \geq M-10
\end{cases}
\end{equation*}
for $M=N^{\e}\geq 50$ where $\e < 1$. Throughout this paper we always assume that $N$ is sufficiently large so that this last condition holds true.
Next we define the kernel
\be\label{eq:defnu}
\nu_t(x,y)= -Nw_m(N(x-y))\pty 
\ee
with $m=N^{-\a}$ and we introduce the antisymmetric operator  
\be\label{eq:defA}
\begin{split}
A_t &= \frac{\Th(\cN)}{\sqrt{N}}\int \nu_t(x,y) \bxxstilde\byystilde[\bgxtilde + \bsxstilde]  dx dy -\hc
\end{split}\ee
where we recall the definition (\ref{eq:bxxtildedef}) of the projected operators $\bxxtilde = b(\wt{q}_{x}) = \bxx - \ptx b(\ptt)$. Notice that, in contrast with the quadratic kernel (\ref{eq:defk}), we cut off (\ref{eq:defnu}) on a length scale $m = N^{-\alpha}$, vanishing as $N \to \infty$. This makes sure that, as discussed in Lemma \ref{lm:propnu}, $\| \nu \|_2 \leq C \sqrt{m} \leq C N^{-1/4}$ is small and therefore it allows us to compute the action of the phase $e^{A_t}$, which will be used in combination with the generalized Bogoliubov transformation $e^{B_t}$ to generate the microscopic correlation structure, expanding the exponential to first and second order (all higher order contributions will be negligible). In the following, we will choose $\a = \e = 1/2$. 

An important observation is that conjugation with $e^{A_t}$ does not substantially change the number of excitations and their energy. The proof of the next lemma is deferred to Section \ref{sec:gronA}. 
\begin{lemma}\label{lm:gron-A}
Let $A_t$ be the anti-symmetric operator defined in \eqref{eq:defA}. For every $k \in \bN$ there exists a constant $C>0$ such that 
\begin{equation} \label{eq:gron-NA} e^{-A_t} (\cN+1)^k e^{A_t} \leq C e^{c|t|} (\cN+1)^k\,. \end{equation} 
Let $\cH_N = \cK+ \cV_N$, with 
\begin{equation}\label{eq:KVN} \cK = \int \nabla_x a_x^* \nabla_x a_x, \quad \cV_N = \frac{1}{2} \int N^2 V(N(x-y)) a_x^* a_y^* a_x a_y \end{equation}
Then, for every $k \in \bN$, there is $C > 0$ such that 
\begin{equation}\label{eq:gron-A} e^{-A_t} (\cH_N + \cN + 1) (\cN+1)^k e^{A_t} \leq C e^{c|t|} (\cH_N + \cN+ 1) (\cN+1)^k \end{equation} 
holds true as an operator inequality on $\cF^{\leq N}_{\perp \wt{\ph}_t}$. 
\end{lemma}
{\it Remark.} In (\ref{eq:gron-A}) it is crucial that the exponent $k$ is the same on both sides of the inequality. In fact, this is the reason for the introduction of the cutoff $\Theta (\cN)$ in (\ref{eq:defA}). For cubic phases without cutoff, as the one used in the time-independent setting in \cite{BSS1}, an estimate similar to (\ref{eq:gron-A}) holds true, but only with an additional term $(\cN+1)^{k+2}$ on the r.h.s. of the equation; see \cite[Lemma 5.8]{BSS1}.  

With the cubic phase, we define a new fluctuation dynamics $\cU_N$, setting 
\begin{equation}\label{eq:fluct-dyn2}  \cU_N (t;s) = e^{-A_t} \bar{\cU}_N (t;s) e^{A_s} = e^{-A_t} e^{-B_t} U_{N,t} e^{-i H_N (t-s)} U_{N,s}^* e^{B_s} e^{A_s}\,. \end{equation} 
To prove Theorem \ref{th:main}, we will first show that the difference between $\bar{\cU}_N (t;0) \xi_N$ and $\cU_N (t;0) \xi_N$ is small in norm, in the limit $N \to \infty$. Afterwards, we will prove (\ref{eq:main2}), but with $\bar{\cU}_N$ replaced by $\cU_N$. To this end, we will need some properties of the cubically renormalized fluctuation dynamics $\cU_N$, which we establish in the rest of this section. 

First of all, we need to control the growth of the number of particles and of the energy along the evolution $\cU_N$. The proof of the next proposition is based on the estimates in \cite[Prop. 6.1]{BS} for the dynamics $\bar{\cU}_N$.  
\begin{prop}\label{prop:propenfulldyn}
Under the same assumptions as in Theorem \ref{th:main}, let $\cU_N$ be defined as in (\ref{eq:fluct-dyn2}).  
Then there exists $C,c > 0$ such that  
\begin{equation*}\label{eq:apri} 
		\la \cU_{N}(t;0) \x, (\cH_N+ \cN) \, \cU_{N}(t;0) \x\ra 
		\leq C \exp(c\exp(c\vert t \vert)) \la \x, (\cH_N+\cN+1)\x\ra\,,
	\end{equation*}
	for all $\x\in\Ftru_{\perp \ph}$ and all $t\in\RRR$.
\end{prop}
\begin{proof}
From Lemma \ref{lm:gron-A} (with $k=0$), we find
\[\begin{split} \langle \cU_N (t;0) \xi , (\cH_N+ \cN) \, \cU_{N}(t;0) \x\ra &=  \langle \bar{\cU}_N (t;0) e^{A_0} \xi , e^{A_t} (\cH_N+ \cN) \, e^{-A_t} \bar{\cU}_{N}(t;0) e^{A_0}  \x\ra \\ &\leq C e^{c|t|} \langle \bar{\cU}_N (t;0) e^{A_0} \xi , (\cH_N+ \cN+1) \, \bar{\cU}_{N}(t;0) e^{A_0}  \x\ra\,. \end{split} \]
From \cite[Prop. 6.1 and following remark]{BS} and applying again Lemma \ref{lm:gron-A}, we conclude 
\[ \begin{split}  \langle \cU_N (t;0) \xi , (\cH_N+ \cN) \, \cU_{N}(t;0) \x\ra &\leq C e^{c|t|} \langle e^{A_0} \xi , (\cH_N+ \cN+1) e^{A_0}  \x\ra \\ &\leq C e^{c|t|} \langle \xi, (\cH_N + \cN + 1) \xi \rangle\,. \end{split} \]
\end{proof}

While controlling the growth of the expectation of $\cN, \cH_N$ with respect to the fluctuation dynamics $\cU_N$ is enough to establish convergence of the reduced density, to obtain a norm approximation for the dynamics $\cU_N$ we need more precise information on its generator. To this end, we remark that (\ref{eq:fluct-dyn2}) satisfies the Schr\"odinger type equation 
\[i\partial_t \, \cU_{N}(t;s) = \cJ_{N} (t) \cU_{N}(t;s)\]
with the time-dependent generator $\cJ_{N} (t)$ given by 
\begin{equation}
\begin{split}\label{eq:genJ}	
			\cJ_{N} (t) =  [i\partial_te^{-A_t}]e^{A_t} +
			e^{-A_t}\Big[(i\partial_te^{-B_t})e^{B_t} +e^{-B_t}\big[(i\partial_tU_{N,t})U^*_{N,t} +  U_{N,t} H_N U^*_{N,t}\big]e^{B_t}\Big]e^{A_t}\,.
		\end{split}	
	\end{equation}
In the next proposition, we compute $\cJ_{N} (t)$ up to errors vanishing in the limit $N \to \infty$. The proof of this proposition is deferred to Section \ref{sec:proof-prop}.
\begin{prop} \label{prop:cJNt} 
Under the same assumptions as in Theorem \ref{th:main}, let $\cJ_N (t)$ be the generator of the fluctuation dynamics $\cU_N$, as defined in (\ref{eq:genJ}). Then we have, as a quadratic form on $\Fperpt\times\Fperpt$, 
	\begin{equation*}
		\cJ_{N}(t) = 	\k_{N}(t) + \cJ_{2,N}(t) + \cVN + \cE_{\cJ_N}(t),
	\end{equation*}
where $\k_{N}(t)$ and $ \cJ_{2,N}(t)$ are defined as in \eqref{eq:cubicphasek1}, \eqref{eq:generatorapprox} and where the error term $\cE_{\cJ_N}(t)$ satisfies 
	\[\begin{split}
		\abs{\la \xi_1, \cE_{\cJ_N}(t)\xi_2\ra}
		\leq &\,C\expt  N^{-1/4} 
		\norm{\cHNplusNh\x_1}\norm{(\cH_N + \cN^3+1)^{1/2} (\cN+1) \x_2} 
	\end{split}\]
	for any $\xi_1,\xi_2 \in \Fperpt$. 
\end{prop} 

\section{Quadratic evolution and proof of Theorem \ref{th:main}} 

In this section, we study the quadratic evolution defined in (\ref{eq:evolutionapprox}). First of all, we establish important properties of the time-dependent generator $\cJ_{2,N}$.
\begin{prop}\label{pr:boundsquadraticgenerator}
Let $\cJ_{2,N}$ be defined as in Eq. \eqref{eq:generatorapprox}. Under the same assumptions as in Theorem \ref{th:main}, there exist constant $C,c > 0$ such that 
\begin{align}
	    \label{eq:quadgeninequality1}
	\pm(\cJ_{2,N} (t) - \cK) & \leq C\expt \cNplusone,\\
	\label{eq:quadgeninequality1b}
	(\cJ_{2,N} (t) - \cK)^2 & \leq C\expt (\cN+1)^2  ,\\
	\label{eq:quadgeninequality2}
	\pm \dot{\cJ}_{2,N} (t) &\leq C\expt \cNplusone\\
	\label{eq:quadgeninequality2b}
	\dot{\cJ}^2_{2,N} (t) &\leq C\expt (\cN + 1)^2 \\
	\label{eq:quadgeninequality3}
	\abs{\la\x_1, [\cN, \cJ_{2,N} (t)] \x_2\ra}
	&\leq C\expt \norm{\cNplusone^{(1+j)/2}\x_1}\norm{\cNplusone^{(1-j)/2}\x_2}
	\end{align}
	for $j\in\ZZZ$ and $\x_1, \x_2\in\Ftru$. 
	\end{prop} 

To show Prop. \ref{pr:boundsquadraticgenerator} (and some of the other bounds discussed in this section), the following lemma will be useful. Its proof is a straighforward adaptation of \cite[Lemma 3.4, 3.6]{BCS}. 
\begin{lemma}
\label{lm:aprioribounds}
Let $F$ be a bounded operator and $J_1,J_2$ two Hilbert-Schmidt operators on $L^2 (\bR^3)$. We also denote by $F, J_1, J_2$ the integral kernels of the three operators ($J_1, J_2 \in L^2 (\bR^3 \times \bR^3)$, while $F$ is in general a distribution). Let 
\begin{equation}\label{eq:A1A3} 
\begin{split}
A_1 &= \int dx\, b^\# (J_{1,x}) b^\# (J_{2,x}) \, , \quad A_2 = \int dx \, b^\#(J_{1,x}) b_x \,,\quad 
\quad A_3= \int dx dy \, F (x,y) b_x^* b_y \end{split}\end{equation} 
where $\#= *, \cdot$. Then, we have 
\be
\label{eq:op1}
\begin{split}
    |\la \xi_1, A_1 \xi_2 \ra| &\leq C \|J_1\|_2 \norm{\cNplusone^{(1+p)/2}\x_1}
    \norm{\cNplusone^{(1-p)/2}\x_2}\\
    |\la \xi_1, A_2 \xi_2 \ra| &\leq C \|J_1\|_2\|J_2\|_2 \norm{\cNplusone^{(1+p)/2}\x_1}
    \norm{\cNplusone^{(1-p)/2}\x_2}\\
    |\la \xi_1,  A_3 \xi_2\ra |&
    \leq C  \| F \|_\text{op} \norm{\cNplusone^{(1+p)/2}\x_1}
    \norm{\cNplusone^{(1-p)/2}\x_2}\\
\end{split}
\ee
for all $p \in \bZ$ and $\xi_1, \xi_2 \in \cF^{\leq N}$. Moreover,  
\be
\label{eq:op2}
\begin{split}
    A_1^*A_1 + A_1A_1^* &\leq C\|J_1\|_2^2(\cN+1)^2\\
     A_2^*A_2 + A_2A_2^* &\leq C\|J_1\|^2_2 \|J_2\|^2_2 (\cN+1)^2 \,\\
     A_3^2 &\leq C \|F \|_\text{op}^2 (\cN+1)^2\,.
\end{split}
\ee
\end{lemma} 

{\it Remark.} The bounds in Lemma \ref{lm:aprioribounds} continue to hold true if we replace the operator $b, b^*$ (or the corresponding operator-valued distributions) with the projected operators $\tilde{b}, \tilde{b}^*$, introduced in (\ref{eq:bxxtildedef}) and used in (\ref{eq:generatorapprox}) to define the generator $\cJ_{2,N} (t)$. In fact, it is easy to see that switching from $b, b^*$ to $\tilde{b}, \tilde{b}^*$ corresponds to multiplying the operators $J_1, J_2, F$ with the orthogonal projection $\wt{q}_t = 1 - |\wt{\ph}_t \rangle \langle \wt{\ph}_t|$ on the right and/or on the left; this does not increase the norms $\| J_1 \|_2 , \| J_2 \|_2, \| F \|_\text{op}$ appearing on the r.h.s. of (\ref{eq:op1}), (\ref{eq:op2}). Furthermore, (\ref{eq:op1}), (\ref{eq:op2}) (and their proof) hold true also for operators $\wt{A}_1, \wt{A}_2, \wt{A}_3$ on the full Fock space $\cF$ (without truncation to $\cN \leq N$), defined like $A_1, A_2, A_3$, but with $b, b^*$ replaced by the standard creation and annihilation operators $a, a^*$ or by their projected version $\dbtilde{a}, \dbtilde{a}^*$, used in the definition of the limiting generator (\ref{eq:generatorinfapprox}). 

\begin{proof}[Proof of Prop. \ref{pr:boundsquadraticgenerator}]
With the notation introduced in the last two lines on the r.h.s. of \eqref{eq:generatorapprox}, we claim that $\norm{G_t}_{op}, \norm{H_t}_2\leq C\expt$. For most terms this follows easily from Lemma \ref{3.0.sceqlemma}, Prop. \ref{prop:propertiesphit} and Lemma \ref{lm:propeta}. In fact, all contributions (to $G_t$ and to $H_t$) arising from the kinetic part (\ref{eq:J2NK}) that involve derivatives of $p, r$ or $\mu$ have $L^2$-norm (and therefore also operator norm) less than $C\expt$. Also the term 
\[ \begin{split}  \int dx \, \tilde{b}^* (\nabla_x \eta_x) \tilde{b} (\nabla_x \eta_x)  &= \int dz dw \, \left[ \int dx \, \nabla_x \eta_t (z;x) \nabla_x \eta_t (x,w) \right] \tilde{b}^*_z \tilde{b}_w \\ &=: \int dz dw \, u_t (z,w) \tilde{b}^*_z \tilde{b}_w \end{split} \]
can be bounded with Lemma \ref{lm:aprioribounds}, since $\| u_t \|_\text{op} \leq \| u_t \|_2 \leq C e^{c|t|}$ by Lemma \ref{lm:propeta}.
To handle contributions arising from the potential part (\ref{eq:cJ2NV}), we observe that $N^3 (V_N f_{N,\ell}) (x-y) = N^3 V (N (x-y)) f_\ell (N(x-y))$ is the integral kernel of the differential operator $\widehat{V f_\ell} (i\nabla / N)$ on $L^2 (\bR^3)$ (with operator norm bounded by $\hat{V} (0)$), that $\g_t, p_t$ and $\s_t$ are also bounded operators on $L^2 (\bR^3)$ and that all off-diagonal terms (contributing to $H$) involve at least one factor of $p_t$ or $\s_t$ (with bounded Hilbert-Schmidt norm, uniformly in $N$). Finally, let us consider the contributions from the other terms on the r.h.s. of (\ref{eq:generatorapprox}) (line 2 to line 7). Most of them can be handled as above, with the estimates from Lemma \ref{3.0.sceqlemma}, Prop. \ref{prop:propertiesphit} and Lemma \ref{lm:propeta} (notice, in particular, that $N^3\l_\ell \leq C$). Some more attention is needed for the term involving $\nabla w_{N,\ell}$. Using Lemma \ref{3.0.sceqlemma}, we bound 
\be\label{eq:badterm}
\begin{split}
   &\abs{N\nabla \wlnxy [\nabla\ptx\pty - \ptx\nabla\pty]}\\
    &\leq \frac{C \chi(\abs{x-y}\leq \ell)}{\abs{x-y}^2} \left( 
    \abs{\nabla\ptx} \abs{\pty - \ptx} + \abs{\ptx} \abs{\nabla\ptx - \nabla\pty}
    \right)\\
    &\leq \frac{C \norm{\ptt}_{H^4} \chi(\abs{x-y}\leq \ell)}{\abs{x-y}} \left( 
    \abs{\nabla\ptx}+ \abs{\ptx} 
    \right)\,.\\
\end{split}
\ee
Thus, the $L^2$-norm of the l.h.s. is bounded, uniformly in $N$. This concludes the proof of the bounds $\norm{G_t}_{op}, \norm{H_t}_2\leq C\expt$. From Lemma \ref{lm:aprioribounds} (and from the remark after the lemma) we arrive at \eqref{eq:quadgeninequality1} and, using (\ref{eq:op2}), to  \eqref{eq:quadgeninequality1b} (the second term on the r.h.s. of \eqref{eq:generatorapprox}, the one proportional to $a_x^* a_x$, can be handled in the same way). Since 
\[ [\cN, \cJ_{2,N} (t) ] = \intxy \Big[ 2 H_t (x,y)\bxxstilde\byystilde - 2 
\overline{H}_t (x,y) \bxxtilde\byytilde\Big] \, \] we also conclude (\ref{eq:quadgeninequality3}). 

As for \eqref{eq:quadgeninequality2} and \eqref{eq:quadgeninequality2b}, we observe that contributions to the time-derivative $\dot{\cJ}_{2,N} (t)$ have the same form as contributions to $\cJ_{2,N} (t)$, either with a factor $\wt{\ph}_t$ replaced by $\dot{\wt{\ph}}_t$, or with one of the kernel $\eta_t, \gamma_t, \sigma_t, p_t, r_t$ replaced by $\dot{\eta}_t, \dot{\gamma}_t, \dot{\sigma}_t,\dot{p}_t, \dot{r}_t$, or with one operator $\tilde{b}$, $\tilde{b}^*$ replaced by its time-derivative (the projection depends on time). Using 
\begin{equation*}
\partial_t \bxxstilde =  -\overline{\dptx}b^*(\ptt) -\overline{\ptx}b^*(\dptt) ,
\end{equation*}
the similar formula for $\partial_t \tilde{b}_x$ and the bounds in Lemma \ref{3.0.sceqlemma}, Prop. \ref{prop:propertiesphit} and Lemma \ref{lm:propeta}, we conclude that $\dot{\cJ}_{2,N} (t)$ can be written as the sum of terms of the form (\ref{eq:A1A3})\footnote{In fact, from (\ref{eq:badterm}) we find a contribution proportional to $\| \dot{\wt{\ph}}_t \|_{H^4}$; this is the only term where control of the $H^6$-norm of $\wt{\ph}_t$ is needed.} (with two projected operators $\tilde{b}^\sharp$ or with one $\tilde{b}^\sharp$ and one $b^\sharp$). For this reason, Lemma \ref{lm:aprioribounds} also implies \eqref{eq:quadgeninequality2} and \eqref{eq:quadgeninequality2b} (again, the term proportional to $a_x^* a_x$ on the r.h.s. of (\ref{eq:generatorapprox}) can be handled similarly). 
\end{proof} 

With the help of Prop. \ref{pr:boundsquadraticgenerator}, we obtain well-posedness of the equation (\ref{eq:evolutionapprox}) (existence and uniqueness of the unitary quadratic evolution $\cU_{2,N}$) and control on the growth of kinetic energy and number of particles. Compared with the bounds obtained in Prop. \ref{prop:propenfulldyn} for the full fluctuation dynamics, here we can derive stronger estimates, controlling arbitrary moments of the number of particles operator $\cN$ and the second moment of the kinetic energy operator $\cK$. These improvements (which we can only show for $\cU_{2,N}$ and not for the full fluctuation dynamics $\cU_N$) will play a crucial role in the proof of Theorem \ref{th:main}. 

\begin{prop} \label{prop:wellposedness}
Under the same assumptions as in Theorem \ref{th:main}, let $\cJ_{2,N}$ be defined as in Eq. \eqref{eq:generatorapprox}. Then, $\cJ_{2,N} (t)$ generates a unique strongly continuous two-parameter group $\cU_{2,N}$ of unitary transformations satisfying the Schr\"odinger equation (\ref{eq:evolutionapprox}). For every $t,s \in \bR$, $\cU_{2,N} (t;s): \cF^{\leq N}_{\perp \wt{\ph}_s} \to\Fperpt$. Moreover, for every $k \in \bN$ there exists $C,c > 0$ such that 
\begin{equation}
	\label{eq:propagationquadratic}
	\begin{split}
		\la \cU_{2,N}(t;0) \x, \cNplusoneto{k}\, \cU_{2, N}(t;0) \x\ra 
		& \leq C \exp(c\exp(c\vert t \vert)) \la \x,\cNplusoneto{k} \x\ra\\
		\la \cU_{2,N}(t;0) \x, \cK^2\, \cU_{2, N}(t;0) \x\ra 
	& \leq C \exp(c\exp(c\vert t \vert)) \la \x,(\cK^2+ \cN^2 + 1) \x\ra\,.
		\end{split}
	\end{equation}
\end{prop}

{\it Remark.} From \cite[Lemma 3.10]{BCS} (extended trivially to the case $\beta=1$), we have 
\begin{equation}
    \label{eq:boundVN}
    \cV_N \leq \frac{C}{N}\big(\cK^2 + \cN^2\big)\,,
\end{equation}
which also implies $\cHN \leq C(\cK^2+\cN^2)$ and 
\begin{equation*}\label{eq:boundHN}
\cHNplusN\cNplusone 
\leq C(\cK^2+\cNplusone^2).
\end{equation*}

\begin{proof}
To prove the well-posedness, we proceed as in \cite[Theorem 7]{LNS}. Note that \eqref{eq:quadgeninequality1}, \eqref{eq:quadgeninequality2} and \eqref{eq:quadgeninequality3} are precisely the inequalities shown in  \cite[Lemma 9]{LNS} and needed to apply the abstract result in \cite[Theorem 8]{LNS}. 

To show that $\cU_{2,N} (t;s)$ maps $\cF^{\leq N}_{\perp \wt{\ph}_s}$ into $\Fperpt$, let us consider $\x \in \cF^{\leq N}_{\perp \wt{\ph}_s}$, which implies that $a(\wt{\ph}_s)\x=0$. We compute
	\[
	\begin{split}
		\frac{d}{dt}\norm{a(\ptt) \, \cU_{2,N}(t;s) \x}^2 
		= &\; 2 \textrm{Re} \la a(\ptt) \, \cU_{2,N}(t;s)\x, \frac{d}{dt}(a(\ptt) \, \cU_{2,N}(t;s)\x)\ra\\
		=&\; 2 \textrm{Im} \la a(\ptt) \, \cU_{2,N}(t;s) \x, (-a(i \dptt ) + a(\ptt)\cJ_{2,N}(t) ) \, \cU_{2,N}(t;s) \x\ra\\
		=&\; -2 \textrm{Im} \la \, \cU_{2,N}(t;s) \x, a^*(\ptt)a(i \dptt ) \, \cU_{2,N}(t;s) \x\ra\\
		&-i  \la \, \cU_{2,N}(t;s) \, \x,  [a^*(\ptt) a(\ptt), \cJ_{2,N}(t)] \, \cU_{2,N}(t;s) \x\ra = 0
	\end{split}
	\]
	where we used that $[a^*(\ptt) a(\ptt), \bxxstilde]=0$ and $[a(\ptt), \cK + \int dx N^3 (V_N \fln) * \abs{\ptt}^2 \axxs\axx] = a(i \dptt)$. 
	
Let us now show (\ref{eq:propagationquadratic}). From \eqref{eq:quadgeninequality3}, we find 
\[ \begin{split}
\Big|\frac{d}{dt}\la \cU_{2,N}(t;0) \x, \cNplusoneto{k} \cU_{2, N}(t;0) \x\ra \Big|
&=\big| \la \cU_{2,N}(t;0) \x, i [ \cNplusoneto{k}, \cJ_{2,N}(t)]\, \cU_{2, N}(t;0) \x\ra\big| \\
&\leq C\expt \la \cU_{2,N}(t;0) \x, \cNplusoneto{k} \cU_{2, N}(t;0) \x\ra \,.
\end{split} \]
The first estimate in (\ref{eq:propagationquadratic}) follows from Gr{\"o}nwall's Lemma. 

As for the second bound in (\ref{eq:propagationquadratic}), we first apply (\ref{eq:quadgeninequality1b}) to estimate 
\begin{equation}\label{eq:UK2U} \begin{split} \langle \cU_{2,N} &(t;0) \xi, \cK^2 \cU_{2,N} (t;0) \xi \rangle  \\ &\leq 2 \langle \cU_{2,N} (t;0) \xi , \cJ^2_{2,N} (t) \cU_{2,N} (t;0) \xi \rangle + 2 \langle \cU_{2,N} (t;0) \xi, (\cJ_{2,N} (t) - \cK)^2 \cU_{2,N} (t;0) \xi \rangle \\ &\leq 2 \langle \cU_{2,N} (t;0) \xi, \cJ^2_{2,N} (t)  \cU_{2,N} (t;0) \xi \rangle + C e^{c|t|} \langle \cU_{2,N} (t;0) \xi, (\cN+1 )^2 \cU_{2,N} (t;0) \xi \rangle \\ &\leq 2 \langle \cU_{2,N} (t;0) \xi, \cJ^2_{2,N} (t)  \cU_{2,N} (t;0) \xi \rangle + C e^{c e^{c|t|}} \langle \xi, (\cN+1 )^2 \xi \rangle \end{split}  \end{equation}
where, in the last inequality, we applied the first bound in (\ref{eq:propagationquadratic}). To control the first term on the r.h.s. of the last equation, we observe that 
\[ \begin{split} 
\frac{d}{dt} \langle \cU_{2,N} (t;0) \xi , &\cJ_{2,N}^2 (t)  \cU_{2,N} (t;0) \xi  \rangle  \\ &= \big\langle \cU_{2,N} (t;0) \xi , \big[ \cJ_{2,N} (t) \dot{\cJ}_{2,N} (t) + \dot{\cJ}_{2,N} (t) \cJ_{2,N} (t) \big]  \cU_{2,N} (t;0) \xi  \big\rangle \\ &\leq \big\langle \cU_{2,N} (t;0) \xi , \cJ^2_{2,N} (t)  \cU_{2,N} (t;0) \xi  \big\rangle  + \big\langle \cU_{2,N} (t;0) \xi ,\dot{\cJ}^2_{2,N} (t)  \cU_{2,N} (t;0) \xi  \big\rangle\,. \end{split} \]
With (\ref{eq:quadgeninequality2b}) and with the first bound in (\ref{eq:propagationquadratic}), we conclude that 
\[ \begin{split} 
\frac{d}{dt} \langle \cU_{2,N} &(t;0) \xi , \cJ_{2,N}^2 (t)  \cU_{2,N} (t;0) \xi  \rangle  \\ &\leq \big\langle \cU_{2,N} (t;0) \xi , \cJ^2_{2,N} (t)  \cU_{2,N} (t;0) \xi  \big\rangle  + C e^{c|t|} \big\langle \cU_{2,N} (t;0) \xi , (\cN+1)^2  \cU_{2,N} (t;0) \xi  \big\rangle \\  &\leq \big\langle \cU_{2,N} (t;0) \xi , \cJ^2_{2,N} (t)  \cU_{2,N} (t;0) \xi  \big\rangle  + C e^{c e^{c|t|}}\big\langle \xi , (\cN+1)^2 \xi  \big\rangle\,.\end{split} \]
Applying Gr{\"o}nwall's Lemma, we conclude that 
\[ \langle \cU_{2,N} (t;0) \xi , \cJ_{2,N}^2 (t)  \cU_{2,N} (t;0) \xi  \rangle \leq C e^{c e^{c|t|}} \langle \xi, (\cJ_{2,N}^2 (0) + \cN^2 +1) \xi \rangle\,. \]
With (\ref{eq:quadgeninequality1b}) (at time $t=0$), we find therefore 
 \[ \langle \cU_{2,N} (t;0) \xi , \cJ_{2,N}^2 (t)  \cU_{2,N} (t;0) \xi  \rangle \leq C e^{c e^{c|t|}} \langle \xi, (\cK^2 + \cN^2 +1) \xi \rangle\,. \]
Inserting in (\ref{eq:UK2U}), we obtain the desired bound. 
\end{proof}

Combining Prop. \ref{prop:propenfulldyn}, Prop. \ref{prop:cJNt} and Prop. \ref{prop:wellposedness}, we can now proceed with the proof of our first main theorem. 

\begin{proof}[Proof of Theorem \ref{th:main}]

As observed in the remark after Theorem \ref{th:main}, we have
\[ \| e^{-i H_N t} \psi_N - e^{-i \int_0^t \kappa_N (s) ds} U_{N,t}^* e^{B_t} \cU_{2,N} (t;0) \xi_N \| = \| \bar{\cU}_N (t;0) \xi_N - e^{-i \int_0^t \kappa_N (s) ds} \cU_{2,N} (t;0) \xi_N \|\,. \]
Next, we introduce cubic phases to pass from $\bar{\cU}_N$ to the new fluctuation dynamics $\cU_N$. To this end, we estimate 
\[ \begin{split} \| e^{-i H_N t} \psi_N - e^{-i \int_0^t \kappa_N (s) ds} U_{N,t}^* e^{B_t} &\cU_{2,N} (t;0) \xi_N \| \\ \leq \; &\| \bar{\cU}_N (t;0) e^{A_0} \xi_N - e^{-i \int_0^t \kappa_N (s) ds} e^{A_t} \cU_{2,N} (t;0) \xi_N \| \\ &+ \| e^{A_0} \xi_N - \xi_N \| + \| e^{A_t} \cU_{2,N} (t;0) \xi_N  - \cU_{2,N} (t;0) \xi_N \|\,. \end{split} \]
Writing 
\[ \|  e^{A_0} \xi_N - \xi_N \|^2 = 2 - 2 \text{Re } \langle \xi_N, e^{-A_0} \xi_N \rangle = \text{Re} \int_0^1 ds \,  \langle \xi_N, A_0 e^{-s A_0} \xi_N \rangle  \]
and estimating (recalling the definition (\ref{eq:defA})) 
\[ \begin{split}  &\big| \langle \xi_1, A_0 \xi_2 \rangle \big| \\
&\quad \leq C N^{-1/2} 
\int |\nu_0 (x,y)| \, \| a_x a_y \xi_1 \| \left[ \| a_x \xi_2 \| + \| (\cN+1)^{1/2} \xi_2 \|  \right] dx dy \\ 
&\quad \;+ C N^{-1/2} 
\int |\nu_0 (x,y)| \, \| a_x a_y (\cN+1)^{-1/2}\xi_2 \| \left[ \| a_x (\cN+1)^{1/2} \xi_1 \| + \| (\cN+1) \xi_1 \|  \right] dx dy \\ 
&\quad \leq C N^{-1/2} \big[ \sup_x \| \nu_{0,x} \|_2  + \| \nu_0 \|_2 \big] \, \| \cN \xi_1 \| \| (\cN+1)^{1/2} \xi_2 \| \end{split} \] 
we obtain, using the estimates $\sup_x \norm{\nu_{0,x}}_2 , \| \nu_0 \|_2 \leq C \sqrt{m} \leq C N^{-1/4}$ from Lemma \ref{lm:propnu},  
\[ \| e^{A_0} \xi_N - \xi_N \|^2 \leq C N^{-3/4} \int_0^1 ds  \| (\cN+1) \xi_N \| \| (\cN + 1)^{1/2} e^{-s A_0} \xi_N \|\,. \] 
With Lemma \ref{lm:gron-A}, we arrive at 
\begin{equation}\label{eq:eA-insert} \| e^{A_0} \xi_N - \xi_N \|^2 \leq C N^{-3/4}  \| (\cN+1) \xi_N \| \| (\cN+1)^{1/2} \xi_N \|\,.  \end{equation}
Similarly, using also Prop. \ref{prop:wellposedness}, we find  
\[ \| e^{-A_0}  \cU_{2,N} (t;0) \xi_N - \cU_{2,N} (t;0) \xi_N \|^2 \leq C e^{ce^{c|t|}} N^{-3/4} \| (\cN+1) \xi_N \| \| (\cN+1)^{1/2} \xi_N \|\,. \]
Hence, we conclude that 
\begin{equation}\label{eq:pr-main} \begin{split} \| e^{-i H_N t} \psi_N &- e^{-i \int_0^t \kappa_N (s) ds} U_{N,t}^* e^{B_t} \cU_{2,N} (t;0) \xi_N \| \\ &\leq \| \cU_N (t;0)  \xi_N - e^{-i \int_0^t \kappa_N (s) ds} \cU_{2,N} (t;0) \xi_N \| + C e^{ce^{c|t|}} N^{-3/8}\,.  \end{split} \end{equation} 

We now compute 
\[ \begin{split}
		\frac{d}{dt}\Big\|\cU_{N}(t;0) \x_N - &e^{-i\int_0^t ds\k_{N}(s)} \cU_{2,N}(t;0) \x_N\Big\|^2\\
		&=2\textrm{Im } \la \cU_{N}(t)\x_N, \big(\cJ_N(t)-\cJ_{2,N}(t)- \k_{N}(t)\big) e^{-i\int_0^t ds \k_{N}(s)} \cU_{2,N}(t)\x_N\ra\,. \end{split} \]
With Prop. \ref{prop:cJNt} and with the bound \eqref{eq:boundVN}, we obtain  
\[ \begin{split} 
\frac{d}{dt} &\Big\|\cU_{N}(t;0) \x_N - e^{-i\int_0^t ds\k_{N}(s)} \cU_{2,N}(t;0) \x_N\Big\|^2\\
		&\leq \big|\la \cU_{N}(t;0)\x_N, \cE_{\cJ_N}(t)\,  \cU_{2,N}(t;0)\x_N\ra\big|+ \| \cV_N^{1/2} \cU_N (t;0) \xi_N \| \| \cV_N^{1/2} \cU_{2,N} (t;0) \xi_N \| \\
		&\leq \big|\la \cU_{N}(t;0)\x_N, \cE_{\cJ_N}(t)\,  \cU_{2,N}(t;0)\x_N\ra\big|+ \frac{1}{\sqrt{N}} \| \cV_N^{1/2} \cU_N (t;0) \xi_N \| \| (\cK + \cN)  \,\cU_{2,N} (t;0) \xi_N \| \,.
	\end{split}
	\]
From Prop.  \ref{prop:cJNt}, we conclude that 
\[ \begin{split} 
\frac{d}{dt} \Big\| &\cU_{N}(t;0) \x_N - e^{-i\int_0^t \k_{N}(s) ds}  \cU_{2,N}(t;0) \x_N\Big\|^2\\
\leq\; &\frac{C e^{c|t|}}{N^{1/4}} \| (\cH_N + \cN + 1)^{1/2}  \cU_N (t;0) \xi_N \| \| (\cH_N + \cN^3 + 1)^{1/2} (\cN+1)  \, \cU_{2,N} (t;0) \xi_N \|  
\\ &+  \frac{1}{\sqrt{N}} \| \cV_N^{1/2} \cU_N (t;0) \xi_N \| \| (\cK + \cN)  \,\cU_{2,N} (t;0) \xi_N \|\,. 
\end{split} \]
Using $\cV_N \leq C \cK \cN$ (which follows from Sobolev inequality, since $\| N^2 V (N.) \|_{3/2} \leq \| V \|_{3/2} < \infty$), we arrive at 
\[ \begin{split} 
\frac{d}{dt} \Big\| \cU_{N}(t;0) \x_N - e^{-i\int_0^t ds\k_{N}(s)} &\cU_{2,N}(t;0) \x_N\Big\|^2\\
\leq\; &\frac{C e^{c|t|}}{N^{1/4}} \langle  \cU_N (t;0) \xi_N ,  (\cH_N + \cN + 1)  \cU_N (t;0) \xi_N \rangle \\ &+ \frac{C e^{c|t|}}{N^{1/4}} \langle \cU_{2,N} (t;0) \xi_N , (\cK^2 + \cN^6 + 1) \cU_{2,N} (t;0) \xi_N \rangle\,. 
\end{split} \]
With Prop. \ref{prop:propenfulldyn} and Prop. \ref{prop:wellposedness}, we obtain 
\[ \begin{split} 
\frac{d}{dt} \Big\| &\cU_{N}(t;0) \x_N - e^{-i\int_0^t \k_{N}(s) ds} \cU_{2,N}(t;0) \x_N\Big\|^2 \leq \frac{C e^{ce^{c|t|}}}{N^{1/4}} \langle \xi_N, (\cK^2 + \cN^6 + 1) \xi_N \rangle\,.  \end{split} \]
Integrating over $t$, using the assumption (\ref{eq:init-bd}) and combining with (\ref{eq:pr-main}), we find (\ref{eq:main}). 
\end{proof}

\section{Limiting quadratic evolution and proof of Theorem \ref{th:limitingdynamics}} 

In this section we show the well-posedness of the limiting Schr\"odinger equation (\ref{eq:limitingdynamics}), we control the growth of the number of particles w.r.t. the limiting quadratic evolution $\cU_{2,\infty}$ and we show the convergence of $\cU_{2,N}$ to $\cU_{2,\infty}$ in the limit $N \to \infty$, as stated in Theorem  \ref{th:limitingdynamics}.

\begin{prop}\label{prop:wp-infty} 
Under the same assumptions as in Theorem \ref{th:main}, let $\cJ_{2,\infty}$ be defined as in \eqref{eq:generatorinfapprox}. Then, $\cJ_{2,\infty} (t)$ generates a unique strongly continuous two-parameter group $\cU_{2,\infty}$ of 
unitary transformations satisfying (\ref{eq:limitingdynamics}). For every $t,s \in \bR$, $\cU_{2,\infty} (t;s): \cF_{\perp \ph_s} \to \cF_{\perp \ph_t}$, where $\ph_t$ denotes the solution of the limiting Gross-Pitaevskii equation (\ref{eq:GPtd}). Moreover, for every $k \in \bN$ there exists $C,c > 0$ such that 
\begin{equation}
	\label{eq:propag-infty}
	\begin{split}
		\la \cU_{2,\infty}(t;0) \x, \cNplusoneto{k}\, \cU_{2, \infty}(t;0) \x\ra 
		& \leq C \exp(c\exp(c\vert t \vert)) \la \x,\cNplusoneto{k} \x\ra		
		\end{split}
	\end{equation}
for all $\xi \in \cF_{\perp \ph}$. 
\end{prop} 

 \begin{proof} The proof is very similar to the proof of Prop. \ref{prop:wellposedness}. Combining Lemma \ref{lm:aprioribounds} (and the remark after it) with the estimates in Lemma \ref{lm:bds-infty}, one can prove that the limiting generator $\cJ_{2,\infty} (t)$ satisfies bounds analogous to (\ref{eq:quadgeninequality1}), (\ref{eq:quadgeninequality2}), (\ref{eq:quadgeninequality3}) in Prop. \ref{pr:boundsquadraticgenerator}. The well-posedness of (\ref{eq:limitingdynamics}), the fact that $\cU_{2,\infty} (t;s)$ maps $\cF_{\perp \ph_s}$ into $\cF_{\perp \ph_t}$ and the bound \eqref{eq:propag-infty} can then be shown arguing exactly as in the proof of Prop. \ref{prop:wellposedness}.
 \end{proof} 
 
 To prove the convergence of $\cU_{2,N}$ towards $\cU_{2,\infty}$, we bound the difference of the two generators. 
 Since $\cJ_{2,N} (t)$ is only defined on the truncated Fock space $\cF^{\leq N}$, our estimate is restricted to this space. 
 \begin{prop}\label{prop:cJN-cJinfty}
 Under the assumptions of Theorem \ref{th:main}, we have, for every $\xi_1, \xi_2 \in \cF^{\leq N}$ and for every $t \in \bR$, 
 \[ \begin{split} \Big| \big\langle \xi_1, \big( \cJ_{2,N} (t;0) &- \cJ_{2,\infty} (t;0) \big) \xi_2 \big\rangle \Big| \\ &\leq \frac{Ce^{ce^{c|t|}}}{\sqrt{N}}\left[  \| (\cN+1) \xi_1 \| \| (\cN+1)\xi_2 \| + \| (\cN+1)^{1/2} \xi_1 \| \| \cK^{1/2} \xi_2 \| \right]\,. \end{split} \]
 \end{prop} 
 
 \begin{proof}
 From (\ref{eq:generatorapprox}) and (\ref{eq:generatorinfapprox}), we write 
 \[ \begin{split}
\cJ_{2,N}(t)= &\; \cK + \int dx N^3 (V_N \fln * \abs{\ptt}^2) (x) \axxs\axx  \\ &+ \int dx dy \, N^3 (V_N f_{N,\ell}) (x-y) \wt{\ph}_t (x) \overline{\wt{\ph}}_t (y) \tilde{b}^*_x \tilde{b}_y  \\
        &+ \intxy \Big(G'_t (x,y) \bxxstilde\byytilde + H_t (x,y)\bxxstilde\byystilde + \overline{H}_t (x,y) \bxxtilde\byytilde\Big) \, , \\
\cJ_{2,\infty}(t)= &\; \cK + 8 \pi \frak{a} \int dx \, |\ph_t (x)|^2 a_x^* a_x + 8\pi \frak{a} \int dx \, |\ph_t (x)|^2  \dbtilde{a}^*_x \dbtilde{a}_x \\
        &+ \int dx dy \,  \Big(G'_{t,\infty} (x,y) \, \dbtilde{a}^*_x \dbtilde{a}_y + H_{t,\infty}  (x,y) \, \dbtilde{a}^*_x \dbtilde{a}^*_y + \overline{H}_{t,\infty}  (x,y) \, \dbtilde{a}_x \dbtilde{a}_y \Big)\,. \end{split} \]
Observe that, with respect to (\ref{eq:generatorapprox}), we extracted the operator on the second line from the $\tilde{b}^* (\gamma_x) \tilde{b} (\gamma_y)$-contribution in the last summand in (\ref{eq:cJ2NV}), denoting by $G'_t$ the difference between $G_t$ (as appearing in the last line of (\ref{eq:generatorapprox})) and this term (we isolate this term because it will require some additional care). In the expansion for $\cJ_{2,\infty}$, we extracted the corresponding term, proportional to $\dbtilde{a}^*_x \dbtilde{a}_x$, from $G_{t,\infty}$. Thus, we have 
\[ \cJ_{2,N} (t) - \cJ_{2,\infty} (t) = \text{I} + \text{II} + \text{III}  + \text{IV}\]
where
\[ \begin{split} 
\text{I} = \; &\int dx \Big[ N^3 (V_N f_{N,\ell} *|\wt{\ph}_t|^2) (x) - 8\pi \frak{a} |\ph_t (x)|^2 \Big] \, a_x^* a_x \\
\text{II} = \; &\int dx dy \, \Big[ \wt{q}_t  \, G'_t \, \wt{q}_t - q_t \, G'_{t,\infty} \, q_t \Big] (x,y) a_x^* a_y - \int dx dy \, \big[  \wt{q}_t \, G'_t \, \wt{q}_t \big] (x,y) \, a_x^* \, \frac{\cN}{N} \, a_y \\
\text{III} = \; &\int dx dy \,  \Big[ \wt{q}_t \, P_{t} \, \wt{q}_t - q_t \, P_{t,\io} \, q_t \Big] (x,y) a_x^* a_y - \int dx dy \, \big[\wt{q}_t \, P_t \, \wt{q}_t \big] (x,y) a_x^* \frac{\cN}{N} a_y \\
\text{IV} = \: &\int dx dy \Big[ (\wt{q}_t \otimes \wt{q}_t) H_t - (q_t \otimes q_t) H_{t,\infty} \Big] (x,y) a_x^* a_y^* \\ &\hspace{1cm} + \int 
dx dy \,  \big[ (\wt{q}_t \otimes \wt{q}_t) H_t \big] (x,y) \, a_x^* a_y^* \Big[ \sqrt{1-(\cN+1)/N} \sqrt{1- \cN/N} - 1\Big] + \text{h.c.} 
\end{split} \]
where we recall that $\wt{q}_t = 1 - |\wt{\ph}_t \rangle \langle \wt{\ph}_t|$ and $q_t = 1 - |\ph_t \rangle \langle \ph_t|$ and where we introduced the notation $P_{t} , P_{t,\io}$ for the operators with the integral kernels $P_{t} (x,y) = N^3 (V_N f_{N,\ell}) (x-y) \wt{\ph}_t (x) \overline{\wt{\ph}}_t (y)$, $P_{t,\io} (x,y) = 8\pi \frak{a} |\ph_t (x)|^2 \, \delta (x-y) $. 
By Lemma \ref{lm:aprioribounds}, we can bound
\[ | \langle \xi_1, \text{I } \xi_2 \rangle | \leq C \| N^3 (V_N f_{N,\ell} *|\wt{\ph}_t|^2)  - 8\pi \frak{a} |\ph_t|^2 \|_\infty \, \| (\cN+1)^{1/2} \xi_1 \| \| (\cN+1)^{1/2} \xi_2 \|\,. \]
For any $x \in \bR^3$, 
we have  
\[ \begin{split} \Big| N^3 (V_N f_{N,\ell} &*|\wt{\ph}_t|^2) (x)  - 8\pi \frak{a} |\ph_t (x)|^2 \Big| \\ \leq \; &\int dz \, V(z) f_\ell (z) \Big| |\wt{\ph}_t (x+z/N)|^2 - |\wt{\ph}_t (x)|^2 \Big| \\ &+ \Big| \int dz \, V(z) f_\ell (z) - 8\pi \frak{a} \Big| |\wt{\ph}_t (x)|^2 + 8\pi \frak{a} \Big|  |\wt{\ph}_t (x)|^2 - |\ph_t (x)|^2 \Big|\,. \end{split} \]
Estimating  
\[ \Big| |\wt{\ph}_t (x+z/N)|^2 - |\wt{\ph}_t (x)|^2 \Big| = \Big| \int_0^1 ds \frac{d}{ds} |\wt{\ph}_t (x+sz/N)|^2 \Big| \leq 2 |z| \| \nabla \wt{\ph}_t \|_\infty \| \wt{\ph}_t \|_\infty / N \]
and observing that $|  |\wt{\ph}_t (x)|^2 - |\ph_t (x)|^2 | \leq  ( \| \wt{\ph}_t \|_\infty + \| \ph_t \|_\infty) \, \| \wt{\ph}_t - \ph_t \|_\infty$ we conclude with Prop. \ref{prop:propertiesphit} and with Eq. (\ref{eq:Vfa0}) from Lemma \ref{3.0.sceqlemma} that 
\[  | \langle \xi_1, \text{I } \xi_2 \rangle | \leq \frac{C e^{ce^{c|t|}}}{N} \| (\cN+1)^{1/2} \xi_1 \| \| (\cN+1)^{1/2} \xi_2 \|\,. \]

Let us now consider the term $\text{II}$. Again with Lemma \ref{lm:aprioribounds}, we can estimate 
\[ \begin{split} | \langle \xi_1, \text{II } \xi_2 \rangle | \leq \; &C \| \wt{q}_t \, G'_t \, \wt{q}_t  - q_t \, G'_{t,\infty} q_t \|_\text{op} \| (\cN+1)^{1/2} \xi_1 \| \| (\cN+1)^{1/2} \xi_2 \| \\ &+ \frac{C}{N} \| \wt{q}_t \, G'_t \, \wt{q}_t \|_\text{op} \| (\cN+1) \xi_1 \| \| (\cN + 1) \xi_2 \|\,.\end{split}  \]
Since $\| \wt{q}_t - q_t \|_\text{op} \leq 2 \| \ph_t - \wt{\ph}_t \|_2$, we obtain, with Prop. \ref{prop:propertiesphit}, 
\[ \| \wt{q}_t \, G'_t \, \wt{q}_t  - q_t \, G'_{t,\infty} q_t \|_\text{op} \leq \frac{C e^{c e^{c|t|}}}{N} \| G'_t \|_\text{op} + \| G'_t - G'_{t,\infty} \|_\text{op}\,. \]
Going through the several contributions to $G'_t, G'_{t,\infty}$ in (\ref{eq:generatorapprox}) and (\ref{eq:generatorinfapprox}) (all diagonal terms) and applying the bounds in Prop. \ref{prop:propertiesphit} (part iv)), Lemma \ref{lm:propeta} and Lemma \ref{lm:bds-infty}, we find that $\| G'_t \|_\text{op} \leq C e^{c|t|}$ (as already discussed in the proof of Prop. \ref{prop:wellposedness}) and that 
\begin{equation}\label{eq:wtG-G} \| G'_t - G'_{t,\infty} \|_\text{op} \leq \frac{Ce^{ce^{c|t|}}}{\sqrt{N}}\,. \end{equation} 
In fact, to compare contributions from (\ref{eq:cJ2NV}) with the corresponding contributions in \eqref{eq:cJinftyV}, we often need to control the convergence of $N^3 V (N(x-y)) f_\ell (N(x-y))$ towards $8\pi \frak{a} \, \delta (x-y)$. Here, it is important to observe that, in all terms contributing to $G'_t$ (which are compared with terms in $G'_{t,\infty}$), the factor $N^3 V (N(x-y)) f_\ell (N(x-y))$ appears in convolution with a kernel $\sigma_x, \sigma_y, p_x, p_y$ (this is not the case for \text{III}; that's why this term has to be handled separately). To further illustrate this point, consider for example the difference $D_t - D_{t,\infty}$ contributing to $G'_t - G'_{t,\infty}$, with 
\[ \begin{split} D_t (y,z) &= \int dx \, N^3 V(N(x-y)) f_\ell (N(x-y)) \wt{\ph}^2_t (y) \sigma_{t} (x,z) , \\ D_{t,\io} (y,z) &= 8 \pi \frak{a} \, \wt{\ph}^2_t (y)  \sigma_{t} (y,z) \end{split} \]
(this difference arises from the term proportional to $\tilde{b}^* (\gamma_y) \tilde{b} (\sigma_x)$ in the second summand in (\ref{eq:cJ2NV})). With (\ref{eq:Vfa0}), we can bound
\begin{equation}\label{eq:DtDinf} \begin{split}  | D_t (y,&z) - D_{t,\io} (y,z)|  \\ = \; &\Big| \int dw V(w) f_\ell (w)  \wt{\ph}^2_t (y) \sigma_{t} (y+w/N , z) - 8 \pi \frak{a} \,  \wt{\ph}^2_t (y) \sigma_{t} (y,z) \Big| \\  \leq \; &Ce^{c|t|} \int dw V(w) f_\ell (w) \,   \big| \sigma_{t} (y+w/N, z) - \sigma_{t} (y,z) \big| + \frac{Ce^{c|t|}}{N}   |\sigma_{t} (y,z)| \\
\leq \; &\frac{Ce^{c|t|}}{N} \int dw  \int_0^1 ds \, V(w) f_\ell (w) |w|   | \nabla_1 \sigma_{t} (y+sw/N , z)| + \frac{Ce^{c|t|}}{N}  |\sigma_{t} (y,z)|  \end{split} \end{equation} 
which leads to 
 \[ \| D_t - D_{t,\io} \|_\text{op} \leq \| D_t - D_{t,\io} \|_2 \leq \frac{C e^{c|t|}}{N} \big[  \| \nabla_1 \sigma_{t} \|_2 + \| \sigma_{t} \|_2 \big] \leq \frac{Ce^{c|t|}}{\sqrt{N}}\,. \] 

With \eqref{eq:wtG-G}, we obtain 
\[ |  \langle \xi_1, \text{II } \xi_2 \rangle |  \leq \frac{Ce^{ce^{c|t|}}}{\sqrt{N}} \| (\cN+1) \xi_1 \| \| (\cN + 1) \xi_2 \|\,. \]
Similarly, we can also estimate 
\[ \begin{split} | \langle \xi_1 , \text{III } \xi_2 \rangle | \leq \; &\frac{Ce^{ce^{c|t|}}}{\sqrt{N}} \| (\cN+1) \xi_1 \| \| (\cN + 1) \xi_2 \| \\ &+ \Big| \int dz dx \,  V(z) f_{\ell} (z) |\ph_t (x)|^2 \, \langle a_x \xi_1 , \big( a_x - a_{x+z/N}\big) \xi_2 \rangle \Big|\,. \end{split}  \] 
To bound the last term, we proceed as in \cite[Lemma 5.2]{BCS}. We find 
\[ \begin{split} 
\Big| \int  dz dx \, V(z) f_{\ell} (z) &|\ph_t (x)|^2 \, \langle a_x \xi_1 , \big( a_x - a_{x+z/N}\big) \xi_2 \rangle \Big| \\ = \; &\Big| \int  dz dx \int_0^1 ds \, V(z) f_{\ell} (z) |\ph_t (x)|^2 \, \frac{d}{ds} \langle a_x \xi_1 , a_{x+sz/N} \xi_2 \rangle \Big| \\ \leq\; &\frac{1}{N}  \int  dz dx \int_0^1 ds \, V(z) f_{\ell} (z) |z| |\ph_t (x)|^2 \, \| a_x \xi_1 \|  \| \nabla_x a_{x+s z/N} \xi_2 \| \\ \leq \; &\frac{Ce^{c|t|}}{N} \| (\cN+1)^{1/2} \xi_1 \| \| \cK^{1/2} \xi_2 \|\,.  \end{split} \] 
This is the only contribution where the kinetic energy is needed (exactly because, in contrast with contributions in $G'_t$, here the difference $N^3 V (N(x-y)) f_\ell (N (x-y)) - 8\pi \frak{a} \delta (x-y)$ acts directly on the operators $a_x^* a_y$, without convolution; therefore, some regularity of $\xi_1, \xi_2$ is needed). 

Finally, to control the term $\text{IV}$, we bound
\[ \begin{split} 
|\langle \xi_1, \text{IV} \xi_2 \rangle | \leq \; &C \|  (\wt{q}_t \otimes \wt{q}_t) H_t - (q_t \otimes q_t) H_{t,\infty} \|_2 \| (\cN+1)^{1/2} \xi_1 \| \| (\cN+1)^{1/2} \xi_2 \| \\ &+ C \| (\wt{q}_t \otimes \wt{q}_t) H_t \|_2 \| (\cN+1) \xi_1 \| \| \big[\sqrt{1 - (\cN+1)/N} \sqrt{1- \cN/N} -1 \big] \xi_2 \|\,. \end{split} \]
Using that $\| \wt{q}_t - q_t \|_\text{op} \leq C e^{c e^{c|t|}} / N$, that 
\[  \| \big[\sqrt{1 - (\cN+1)/N} \sqrt{1- \cN/N} -1 \big] \xi_2 \|  \leq \frac{C}{N} \| (\cN+1) \xi_2 \| \]
and that, going through the several contributions to $H_t, H_{t,\infty}$ (the off-diagonal terms) in (\ref{eq:generatorapprox}) and (\ref{eq:generatorinfapprox}) and applying the bounds in Prop. \ref{prop:propertiesphit} (part iv)), Lemma \ref{lm:propeta} and Lemma \ref{lm:bds-infty}, $\| H_t \|_2 \leq C e^{c|t|}$, $\| H_t - H_{t,\infty} \|_2 \leq C e^{ce^{c|t|}} / \sqrt{N}$, we find that  
\[  |\langle \xi_1, \text{IV} \, \xi_2 \rangle | \leq \frac{C e^{c e^{c|t|}}}{\sqrt{N}} \| (\cN+1) \xi_1 \| \| (\cN+1) \xi_2 \| \]
which concludes the proof of the proposition. 
\end{proof}

We can now proceed with the proof of Theorem \ref{th:limitingdynamics}.
\begin{proof}[Proof of Theorem \ref{th:limitingdynamics}] 
First of all, we observe that 
\begin{equation}\label{eq:UUin1} \begin{split} \| \cU_{2,N} (t;0) \xi - \cU_{2,\infty} (t;0) \xi \|^2 &= 2 - 2 \text{Re } \langle \cU_{2,\infty} (t;0) \xi , \cU_{2,N} (t;0) \xi \rangle \\ &= 2 - 2 \text{Re } \langle \cU_{2,\infty} (t;0) \xi , \mathbbm{1} (\cN \leq N) \cU_{2,N} (t;0) \xi \rangle \end{split} \end{equation} 
because $\cU_{2,N} (t;0) \xi  = \mathbbm{1} (\cN \leq N) \cU_{2,N} (t;0) \xi$.  We compute 
\[ \begin{split} -i \frac{d}{dt}  \langle \cU_{2,\infty} (t;0) &\xi , \mathbbm{1} (\cN \leq N) \cU_{2,N} (t;0) \xi \rangle \\ &=  \langle \cU_{2,\infty} (t;0) \xi , \big[ \cJ_{2,\infty} (t) \mathbbm{1} (\cN \leq N) -  \mathbbm{1} (\cN \leq N)  \cJ_{2,N} (t) \big] \cU_{2,N} (t;0) \xi \rangle\,. \end{split} \]
While we cannot move $\cJ_{2,N} (t)$ to the left of the projection $\mathbbm{1} (\cN \leq N)$, we can move $\cJ_{2,\infty} (t)$ to its right, generating a commutator. Thus
\[ \begin{split} -i \frac{d}{dt}  \langle \cU_{2,\infty} (t;0) \xi , \mathbbm{1} (\cN \leq &N) \cU_{2,N} (t;0) \xi \rangle \\ = \; & \langle \cU_{2,\infty} (t;0) \xi , \mathbbm{1} (\cN \leq N) \big(\cJ_{2,\infty} (t) - \cJ_{2,N} (t) \big)   \cU_{2, N} (t;0) \xi \rangle \\ &+  \langle \cU_{2,\infty} (t;0) \xi , \big[ \cJ_{2,\infty} (t), \mathbbm{1} (\cN \leq N) \big] \cU_{2, N} (t;0) \xi \rangle\,. \end{split} \]
With Prop. \ref{prop:cJN-cJinfty} and recalling the expression in the last two lines of (\ref{eq:generatorinfapprox}) for the limiting generator $\cJ_{2,\infty} (t)$, we find 
\[ \begin{split} \Big| \frac{d}{dt}  \langle \cU_{2,\infty} (t;0) &\xi , \mathbbm{1} (\cN \leq N) \cU_{2, N} (t;0) \xi \rangle \Big|  \\ \leq \; & \frac{Ce^{ce^{c|t|}}}{\sqrt{N}} \| (\cN+1) \cU_{2,\infty} (t;0) \xi \| \| (\cK + \cN^2 + 1)^{1/2} \cU_{2,N} (t;0) \xi \|
\\ &+ \Big| \int dx dy H_\infty (x,y)|  \langle  \cU_{2,\infty} (t;0) \xi , \mathbbm{1} (N-2 \leq \cN \leq N)\, \dbtilde{a}_x \dbtilde{a}_y  \cU_{2,N} (t;0) \xi \rangle \Big|
\\ &+ \Big| \int dx dy \, H_\infty (x,y)  \langle  \cU_{2,N} (t;0) \xi , \dbtilde{a}^*_x \dbtilde{a}^*_y  \mathbbm{1} (N \leq 
\cN \leq N+2) \, \cU_{2,\infty} (t;0) \xi \rangle \Big| \,.
\end{split} \]
In the last two terms, we estimate $\mathbbm{1} (\cN \geq N-2) ,  \mathbbm{1}  (\cN \geq N) \leq C \cN/ N$ and we use Lemma \ref{lm:aprioribounds} in combination with $\| H_\infty \|_2 \leq C e^{c|t|}$. We obtain 
\[  \begin{split} \Big| \frac{d}{dt}  \langle \cU_{2,\infty} (t;0) &\xi , \mathbbm{1} (\cN \leq N) \cU_{2, N} (t;0) \xi \rangle \Big|  \\ \leq \; & \frac{Ce^{ce^{c|t|}}}{\sqrt{N}} \| (\cN+1) \cU_{2,\infty} (t;0) \xi \| \| (\cK + \cN^2 + 1)^{1/2} \cU_{2,N} (t;0) \xi \|\,.   \end{split} \]
From Prop. \ref{prop:wellposedness} and Prop. \ref{prop:wp-infty}, we conclude that 
\[  \begin{split} \Big| \frac{d}{dt}  \langle \cU_{2,\infty} (t;0) &\xi , \mathbbm{1} (\cN \leq N) \cU_{2, N} (t;0) \xi \rangle \Big| \leq  \frac{Ce^{ce^{c|t|}}}{\sqrt{N}} \| (\cN+1) \xi \| \| (\cK + \cN^2 + 1)^{1/2} \xi \| \,.  \end{split} \]
Integrating over $t$ and with the assumption (\ref{eq:exp-io}), we arrive at
\[ \Big| 1 - \langle \cU_{2,\infty} (t;0) \xi , \mathbbm{1} (\cN \leq N) \cU_{2,N} (t;0) \xi \rangle \Big| \leq  \frac{Ce^{ce^{c|t|}}}{\sqrt{N}} \,. \]
Inserting on the r.h.s. of (\ref{eq:UUin1}) proves the desired estimate. 
\end{proof}

\section{Central Limit Theorem: Proof of Theorem \ref{th:clt}}
\label{sec:CLT}

Following the remark after Theorem \ref{th:clt}, in this section we aim at proving that 
\begin{equation} \label{eq:clt-claim}
\begin{split}
\Big| \mathbb{E}_{\psi_{N,t}} g (\cO_{N,t}) &- \frac{1}{\sqrt{2 \pi} \| f_t \|} \int dx \, g(x) e^{-\frac{x^2}{2 \| f_t \|^2}} \Big| \\
&\leq C  e^{c e^{c |t|}} \int ds \, |\hat{g} (s)| \, (N^{-1/8} +N^{-1/2} \abs{s}^3 \| O \|^3+ N^{-1} s^4 \| O \|^4) 
\end{split}
\end{equation} 
for every $g \in L^1 (\bR)$ with $\hat{g} \in L^1 (\bR, (1+s^4) ds)$. 

For the initial wave function $\psi_N = U_{N,0}^* e^{B_0} e^B \Omega$ with $B_0$ defined as in (\ref{eq:defB}) and with $B$ given by \eqref{eq:Btau}, with $\tau \in (q_0 \otimes q_0) H^2 (\bR^3 \times \bR^3)$, we find that (\ref{eq:init-bd}) is satisfied, with $\xi_N = e^B \Omega$. Thus, Theorem \ref{th:main} provides the norm approximation 
\[ \| e^{-iH_N t} \psi_N - U_{N,t}^* e^{B_t} \cU_{2,N} (t) e^B \Omega \| \leq C e^{c e^{c|t|}} N^{-1/8} \, .\]
Writing 
\[ \mathbb E _{\psi_{N,t}} [g (\cO_{N,t})] = \la \psi_{N,t},g (\cO_{N,t}) \psi_{N,t}\ra =\int ds \,\hat g (s) \la \psi_{N,t} , e^{is \cO_{N,t}} \psi_{N,t}\ra \, , \]
setting 
\[
\wt\cO_{N,t} = \frac{1}{\sqrt N} \sum_{p=1}^N\big(O^{(p)} - \la \wt\ph_{t},O \wt\ph_{t}\ra\big)
\]
and observing that, by Prop. \ref{prop:propertiesphit}, $\| \cO_{N,t} - \wt\cO_{N,t} \|_\text{op} \leq C e^{c e^{c|t|}} /\sqrt{N}$, we can therefore estimate 
\begin{equation}\label{eq:clt-bd1} \begin{split}
\Big|\mathbb E _{\psi_{N,t}} [g (\cO_{N,t})]   - \int ds \, \hat g (s)   \la U_{N,t}^* e^{B_t}  &\cU_{2,N} (t) e^B \Omega,  e^{is \wt \cO_{N,t}} U_{N,t}^* e^{B_t} \cU_{2,N} (t) e^B \Omega \ra\Big|\\
    & \leq C e^{ce^{c|t|}}  \int |\hat{g} (s)| (N^{-1/8}+ N^{-1/2}|s| \| O \|) ds \,.
\end{split} \end{equation} 

Next, we conjugate the observable $e^{is \wt \cO_{N,t}}$ with the unitary operators defining the norm approximation. 
With the rules (\ref{eq:U-rules}), we find 
\begin{equation}\label{eq:UNt-act}  U_{N,t} \wt\cO_{N,t} U_{N,t}^* = \frac{1}{\sqrt{N}} d\Gamma (\wt{q}_t \, O' \wt{q}_t) + \phi (\wt{q}_t O \wt{\ph}_t) \end{equation}
with $O' = O - \langle \wt{\ph}_t , O \wt{\ph}_t \rangle$. Here we set $\phi (f) = b^* (f) + b (f)$, while $d\Gamma (R)$ denotes the second quantization of the one-particle operator $R$ and $\wt{q}_t = 1 - |\wt{\ph}_t \rangle \langle \wt{\ph}_t|$. 

When inserting in (\ref{eq:clt-bd1}), the contribution of the first term on the r.h.s. of (\ref{eq:UNt-act}) is small. 
Proceeding as in \cite[Step 1 in Proof of Theorem 1.1]{RS}, we arrive at 
\begin{equation}\label{eq:clt-bd2} \begin{split} \Big|\mathbb E _{\psi_{N,t}} [g (\cO_{N,t})]   - \int ds \, \hat g (s)   \la e^{B_t} &\cU_{2,N} (t) e^B \Omega,  e^{is \phi (\wt{q}_t O \wt{\ph}_t)} e^{B_t}  \cU_{2,N} (t) e^B \Omega \ra\Big|\\
    & \leq C e^{ce^{c|t|}}  \int |\hat{g} (s)| (N^{-1/8}+ N^{-1/2}|s|^3 \| O \|^3) ds \,.
\end{split} \end{equation} 
This bound relies on the control of the growth of the number of particles operator, which follows from Lemma \ref{lm:gron-B}, Lemma \ref{lm:gron-A}, Prop. \ref{prop:wellposedness} and from the estimate  
\begin{equation}\label{eq:phi-N}  \la \xi, e^{-i\phi(f)} \big(\cN + \a \big)^k e^{i\phi(f)}\xi \ra
    \leq C\la \xi, \big(\cN+\a + \|f\|^2\big)^k \xi \ra\end{equation} 
for the action of the modified Weyl operator $e^{i\phi (f)}$. From (\ref{eq:defd}) and (\ref{eq:d-bds}), the action of $e^{B_t}$ is given by
\[ \begin{split}
e^{-B_t} \phi(\wt{q}_t O \wt{\ph}_t)e^{B_t} =& \, \phi(h_{t}) + D
\end{split}
\]
with $h_t = \gamma_t (\wt{q}_t O \wt{\ph}_t) + \sigma_t (\overline{\wt{q}_t O \wt{\ph}_t})$ and with an error $D$ satisfying  
\[ \| D \xi \| \leq \frac{C}{N} \| (\cN+1)^{3/2} \xi \|\,. \]
Proceeding as in \cite[Step 2 in Proof of Theorem 1.1]{RS}, from (\ref{eq:clt-bd2}) we therefore arrive at 
\begin{equation}\label{eq:clt-bd3} \begin{split} \Big|\mathbb E _{\psi_{N,t}} [g (\cO_{N,t})]   &- \int ds \, \hat g (s)   \la  \cU_{2,N} (t) e^B \Omega,  e^{is \phi (h_t)} \cU_{2,N} (t) e^B \Omega \ra\Big| \\
    & \leq C e^{ce^{c|t|}}  \int |\hat{g} (s)| (N^{-1/8}+ N^{-1/2}|s|^3 \| O \|^3 + N^{-1} |s|^4 \| O \|^4) ds\,.
\end{split} \end{equation}

Before proceeding with the last two unitary conjugations, we replace now the field $\phi (h_t) = b^* (h_t) + b (h_t)$  with $\phi_a (h_t) = a^* (h_t) + a (h_t)$. To this end, we observe that 
\begin{equation}\label{eq:phiphia}
\abs{ \la \x_1,  (\phi(f)-\phi_a(f)) \x_2 \ra} 
\leq \frac{C\norm{f}}{N} \| (\cN+1)^{1/2} \xi_1 \| \| (\cN+1) \xi_2 \| \end{equation} 
for all $\x_1, \x_2 \in\Ftru$ (the operators $b, b^*$ are only defined on the truncated Fock space).
Thus, for $\xi \in \cF^{\leq N}$, we have, with the notation $\mathbbm{1}^{\leq N} = \mathbbm{1} (\cN \leq N)$, 
\[ \begin{split} \Big| \la \xi, &e^{is \phi(h_{t})}  \xi \ra  - \la \xi , e^{is \phi_a(h_{t})} \xi \ra \Big| \\ &= \Big| \int_0^s d\l \, \frac{d}{d\lambda}  \langle \xi , e^{i\l \phi (h_t)} \mathbbm{1}^{\leq N} e^{i (s-\l) \phi_a (h_t)} \xi \rangle \Big|  \\ &= \Big| \int_0^s d\l \,  \langle \xi , e^{i\l \phi (h_t)} \Big\{  (\phi (h_t) - \phi_a (h_t) ) \mathbbm{1}^{\leq N} + \big[ \mathbbm{1}^{\leq N}, \phi_a (h_t) \big] \Big\}  e^{i (s-\l) \phi_a (h_t)} \xi \rangle \Big|\,. 
\end{split} \]
Writing $[ \mathbbm{1}^{\leq N}, a (h_t) ] = a (h_t) \mathbbm{1} (N \leq \cN \leq N+1)$, $[ \mathbbm{1}^{\leq N}, a^* (h_t) ] = -  \mathbbm{1} (N \leq \cN \leq N+1) a^* (h_t)$, estimating $\mathbbm{1} (N \leq \cN \leq N+1) \leq \cN / N$ and applying (\ref{eq:phiphia}) and (\ref{eq:phi-N}) (and the analogous bound for the action of $e^{i \phi_a (f)}$), we conclude that 
\[   \Big| \la \xi, e^{is \phi(h_{t})}  \xi \ra  - \la \xi , e^{is \phi_a(h_{t})} \xi \ra \Big|  \leq \frac{C|s| \| h_t \|}{N}  \| (\cN + s^2 \| h_t \|^2)^{1/2} \xi_1 \| \| (\cN+ s^2 \| h_t \|^2) \xi_2 \|\,. \]
Inserting in (\ref{eq:clt-bd3}) we find, with Lemma \ref{lm:gron-B} and Prop. \ref{prop:wellposedness},
\begin{equation}\label{eq:clt-bd4} \begin{split} \Big|\mathbb E _{\psi_{N,t}} [g (\cO_{N,t})]   &- \int ds \, \hat g (s)   \la  \cU_{2,N} (t) e^B \Omega,  e^{is \phi_a (h_t)} \cU_{2,N} (t) e^B \Omega \ra\Big| \\
    & \leq C e^{ce^{c|t|}}  \int |\hat{g} (s)| (N^{-1/8}+ N^{-1/2}|s|^3 \| O \|^3 + N^{-1} |s|^4 \| O \|^4) ds\,. 
\end{split} \end{equation} 

Next, we apply Theorem \ref{th:limitingdynamics} to replace the quadratic evolution $\cU_{2,N}$ with its limit $\cU_{2,\infty}$. Moreover, we replace $h_t$ with  
\[ h_{\infty,t} = \gamma_{t,\infty} (q_t O \ph_t) + \sigma_{t,\infty} (\overline{q_t O \ph_t}) \]
where $\gamma_{t,\infty}, \sigma_{t,\infty}$ are defined as in (\ref{eq:ginfty}) and $q_t = 1 - |\ph_t \rangle \langle \ph_t|$. From Prop. \ref{prop:propertiesphit} and Lemma \ref{lm:bds-infty}, we find 
\[ \| h_t - h_{\infty ,t} \|_2 \leq C e^{c e^{c|t|}} \| O \| / \sqrt{N}\,. \]
From (\ref{eq:clt-bd4}), we therefore obtain 
\begin{equation*} \begin{split} \Big|\mathbb E _{\psi_{N,t}} [g (\cO_{N,t})]   
&- \int ds \, \hat g (s)   \la  \cU_{2,\infty} (t;0) e^B \Omega,  e^{is \phi_a (h_{\infty,t})} \cU_{2,\infty} (t;0) e^B \Omega \ra\Big| \\
    & \leq C e^{ce^{c|t|}}  \int |\hat{g} (s)| (N^{-1/8}+ N^{-1/2}|s|^3 \| O \|^3 + N^{-1} |s|^4 \| O \|^4) ds\,. 
\end{split} \end{equation*} 

The action of $\cU_{2,\infty} (t;0)$ on the operators $a, a^*$ appearing in $\phi_a$ is explicit and described by Prop. \ref{prop:Bogtrasf}. Setting $n_t = U(t;0) h_{\infty,t} + \overline{V(t;0) h_{\infty,t}}$, we find 
\begin{equation}\label{eq:clt-bd5} \begin{split} \Big|\mathbb E _{\psi_{N,t}} [g (\cO_{N,t})]   &- \int ds \, \hat g (s)   \la   e^B \Omega,  e^{is \phi_a (n_{t})} e^B \Omega \ra\Big| \\
    & \leq C e^{ce^{c|t|}}  \int |\hat{g} (s)| (N^{-1/8}+ N^{-1/2}|s|^3 \| O \|^3 + N^{-1} |s|^4 \| O \|^4) ds\,.  
\end{split} \end{equation} 
Finally, we need to compute the action of $e^B$. To this end, we first replace the operator $B$ in (\ref{eq:Btau}) with 
\[ B_a = \frac{1}{2} \int dx dy \left[ \tau (x;y) a_x^* a_y^* - \text{h.c.} \right] , \]
proceeding as we did above to replace $\phi (h_t)$ with $\phi_a (h_t)$ to show that $\| e^B \Omega - e^{B_a} \Omega \| \leq C/\sqrt{N}$. Then we use the explicit formula for the action of the Bogoliubov transformation $e^{B_a}$, which implies, setting $f_{t}= \cosh(\t) n_{t}+ \sinh(\t)\overline{n_{t}}$, that 
\[ \la e^{B_a} \Omega,  e^{is \phi_a (g_{t})} e^{B_a} \Omega \ra = \la \Omega,e^{i s \phi_a(f_{t})} \Omega\ra = 
\la \Omega,  e^{-s^2 \norm{f_t}^2/2}e^{i s a^*(f_t)} e^{i s a(f_t)}  \Omega\ra = e^{-s^2\norm{f_t}^2/2}\,. 
\] 
From (\ref{eq:clt-bd5}), we obtain 
\[\begin{split}
 \Big|\mathbb E _{\psi_{N,t}} [g (\cO_{N,t})]   &- \int ds \, \hat g (s)   e^{-s^2 \| f_t \|^2/2} \Big| \\
&\leq  C e^{ce^{c|t|}}  \int |\hat{g} (s)| (N^{-1/8}+ N^{-1/2}|s|^3 \| O \|^3 + N^{-1} |s|^4 \| O \|^4) ds  
\end{split}\] 
which immediately implies (\ref{eq:clt-claim}). The statement of Theorem \ref{th:clt} now follows by standard arguments (see, for example, \cite[Corollary 1.2]{BSS2}).

\section{Control of action of $A_t$ and proof of Lemma \ref{lm:gron-A}} 
\label{sec:gronA}

In this section, we consider the action of the cubic phase $e^{A_t}$ on number and energy of excitations. To this end, we compute commutators of $A_t$ with the Hamilton operator $\cH_N = \cK + \cV_N$.  
\begin{lemma}\label{lm:commutatorHA}
Recall the definition of $A_t$ in \eqref{eq:defA}, with parameter $M=m^{-1}=N^{1/2}$, and recall the notation $\cH_N = \cK+ \cV_N$, with $\cK$, $\cV_N$ the kinetic and potential energy operators on $\Fperpt$. On $\Fperpt\times\Fperpt$, we have
    \begin{equation}\label{eq:commutatorKA}
    [\cH_N, A_t]= -  \intxy  N^{5/2} \VNxy  \pty \bxxs\byys[\bgx + \bsxs] + \hc +\cE_{[\cH, A_t]}
    \end{equation} 
    where 
    \[
    \begin{split}
    &\abs{\la \x_1, \cE_{[\cH, A_t]} \x_2\ra} \leq C\expt N^{-1/4}
    \norm{(\cHN + \cN +1)^{1/2}\x_1}\norm{(\cHN + \cN +1)^{1/2}\cNplusoneh\x_2}\,.
    \end{split}
    \]
   Furthermore, 
    \begin{equation}\label{eq:KA2}
    \begin{split}
    &\abs{\la \x_1, [\cH_N , A_t] \x_2\ra}
    \\ &\hspace{1cm} \leq   C\expt \norm{\cHNplusNh\cNplusone^{n/2}\x_1}\norm{\cHNplusNh\cNplusone^{-n/2}\x_2}
    \end{split}
    \end{equation} 
    for all $n \in\ZZZ$.
\end{lemma}

{\it Remark.} To apply Lemma \ref{lm:commutatorHA} in the proof of Lemma \ref{lm:gron-A}, it is important that the total exponent of $\cH_N +\cN$ on the r.h.s. of  (\ref{eq:KA2}) is one (so that (\ref{eq:gron-A}) follows by Gr{\"o}nwall's Lemma). To reach this goal, we inserted the cutoff $\Theta (\cN)$ in the definition (\ref{eq:defA}) of $A_t$, similarly as done in \cite{NT} (in \cite[Lemma 5.7]{BSS1}, where the cubic phase does not have a cutoff, the exponent of $\cN$ increases). 

\begin{proof}
We proceed similarly as in the proof of \cite[Lemma 5.7]{BSS1}. We define $A_t^1$ as $A_t$ in (\ref{eq:defA}), but with $\wt{b}, \wt{b}^*$ replaced by $b, b^*$. Using $[\cK, \axxs]=-\Delta_x\axxs$ and that $\cK$ commutes with $\cN$, we have
\[
    \begin{split}
    [\cK, A_t] 
     =&\frac{\Th(\cN)}{\sqrt{N}}\intxy (-\Delta_x-\Delta_y) \nu_t(x,y) \bxxs\byys[\bgx + \bsxs] + \hc\\
    &-\frac{\Th(\cN)}{\sqrt{N}}\intxy \nabla_x \nu_t(x,y) \bxxs\byys[2 \nabla_x \bxx + 2 b(\nabla_x p_x) + b^*(\nabla_x \sxt)] + \hc\\
    &+\frac{\Th(\cN)}{\sqrt{N}}\intxy \nu_t(x,y) \bxxs \byys[ b(-\Delta_x p_x)+ b(\Delta p_x)] + \hc\\
    &+\frac{\Th(\cN)}{\sqrt{N}}\intxy \nu_t(x,y) \nabla_x \bxxs \byys b^*(\nabla_x \sxt) +\hc\\
    &+\frac{\Th(\cN)}{\sqrt{N}}\int dxdydz \nu_t(x,y) \bxxs \byys\nabla_z\bzzs  \nabla_z \s (z,x))] + \hc\\
    &+[\cK, A_t-A_t^1] 
    \eqqcolon \sum_{i=1}^5M_i + \hc + [\cK, A_t-A_t^1]\,.
    \end{split}
\]
With Lemma \ref{lm:propeta} and Lemma \ref{lm:propnu}, we find 
\[
\begin{split}
\abs{\la \x_1, M_3 \x_2\ra} 
&\leq C\expt \sqrt{\frac{m}{N}} \left( \norm{\Delta_1 p} +\norm{\Delta_2 p}\right)\norm{\cNplusoneh\x_1}\norm{\cNplusone\x_2}\\
&\leq C\expt N^{-3/4} \norm{\cNplusoneh\x_1}\norm{\cNplusone\x_2}\,.
\end{split}
\]
Writing $\nabla_x \nu_t(x,y) = - N \nabla_y \wmnxy \pty$ and integrating by parts, we can bound 
\begin{equation}\label{eq:M2-term}
\begin{split}
\abs{\la \x_1, M_2 \x_2\ra} 
&\leq C\expt \sqrt{\frac{m}{N}} 
\norm{\cHNplusNh\x_1}\\
&\hspace{2cm} \times\left(\norm{\cKh\cNplusoneh \x_2}
+\left( \norm{\nabla_2 p}+\norm{\nabla_2 \s}\right) \norm{\cNplusone\x_2}\right)\\
&\leq C\expt N^{-1/4} \norm{\cHNplusNh\x_1}\norm{\cHNplusNh\cNplusoneh\x_2}\,.
\end{split}
\end{equation} 
$M_4$ can be controlled analogously. Also $M_5$ satisfies the same bound, since $\norm{\nabla_1 \s}\leq C\expt \sqrt{N}$ by Lemma \ref{lm:propeta}. As for $M_1$, the scattering equation \eqref{eq:scatlN} yields 
\[
\begin{split}
(-\Delta_x-\Delta_y) \nu_t(x,y)
= &- N^3 \VNxy\fmnxy \pty
+ N^3 \l_m\fmnxy \chi_m (x-y)\\
&- 2 N \nabla \wmnxy \nabla \pty
+ N \wmnxy\Delta \pty.
\end{split}
\]
Let us denote the four contributions to $M_1$ corresponding to the four terms in the last identity by $M_{11},\dots , M_{14}$. Then $M_{14}$ is bounded similarly to $M_3$; $M_{13}$ satisfies the same bound as $M_2$ after integration by parts (as usual, we use Prop. \ref{prop:propertiesphit} to bound norms of $\ph_t$ and of its derivatives). Since $N^3 \l_m \leq C m^{-3}$ by \eqref{eq:lambdaell} and $0 \leq \fmnxy\leq 1$ we can bound
\[
\begin{split}
\abs{\la \x_1, M_{12} \x_2\ra} 
&\leq C\expt \frac{1}{\sqrt{N}m^3}
\left(\intxy \frac{m^2}{\abs{x-y}^2}\norm{\axx\ayy\cNplusoneto{-1/2}\x_1}^2 \right)^{1/2}\\
&\times\left(\intxy \chi_m (x-y) \norm{\cNplusoneto{1/2}[\bgx + \bsxs] \x_2}^2 \right)^{1/2}\\
&\leq C\expt \, N^{-1/4} \norm{\cHNplusNh \x_1}\norm{\cHNplusNh \cNplusoneh\x_2}
\end{split}
\]
where we used Hardy's inequality in the first integral and the choice $m = N^{-1/2}$. As for $M_{11}$, with $\Th = 1 + (\Th-1)$ we write 
 \begin{equation}\label{eq:M11} \begin{split} 
M_{11} &= - \Th(\cN) N^{5/2} \intxy  \VNxy\fmnxy \pty \bxxs\byys[\bgx + \bsxs] \\ 
& =  -  N^{5/2} \intxy  \VNxy\fmnxy \pty \bxxs\byys [\bgx + \bsxs] + \cE
\end{split} \end{equation} 
where, using $\Th(\cN)-1 \leq \mathbbm{1}(\cN\geq M/2)$ and Markov's inequality, we  can estimate 
\[
\begin{split}
\abs{\la \x_1, \cE \x_2\ra} 
&\leq C\expt \norm{\cVh \mathbbm{1}(\cN\geq M/2)\x_1}\norm{\cNplusoneh \mathbbm{1}(\cN\geq M/2)\x_2}\\
&\leq \frac{C\expt}{N^{1/4}} \norm{ (\cVN+\cN+1)^{1/2}\x_1}\norm{ (\cVN+\cN+1)^{1/2}\cNplusoneh\x_2}.
\end{split}
\]

Finally, we have to control the term arising from the difference $A_t-A_t^1$. We observe that, on $\Ftru\times\Ftru$, 
\be\label{eq:differenceprojectionA}   
\begin{split}
	&\bxxs\byys[\bgx + \bsxs]-\bxxstilde\byystilde[\bgxtilde + \bsxstilde]
	\\ &= \bxxs\byys \ptx b(\ptt)
	+b^*(\ptt)\overline{\pty}\bxxs[\bgxtilde + \bsxstilde]
    + b^*(\ptt)\overline{\ptx}\byystilde [\bgxtilde + \bsxstilde].
\end{split}
\ee
Here we used the fact that, because of the projection in the kernel $\eta_t(x,y)$ and the definition of $p_t$ and $\s_t$ in \eqref{eq:defcoshsinh}, we have
$
\la \ptt, \sxt\ra =0 = \la \ptt, p_x\ra.
$
Note that as quadratic forms on $\Fperpt\times\Fperpt$, we have, for any operator $D$, with the shorthand notation $B_x = \bgx + \bsxs$ and $\wt{B}_x = \bgxtilde + \bsxstilde$,  
\[
[D, \bxxs\byys\Bxx-\bxxstilde\byystilde\Bxxtilde ]
= \bxxs\byys \ptx [D, b(\ptt)]
+[D, b^*(\ptt)] \overline{\pty}\bxxs \Bxxtilde
+ [D, b^*(\ptt)] \overline{\ptx}\byystilde \Bxxtilde.
\]
Applying this with $D=\cK$, contributions arising from the commutator of this difference with $\cK$ can be controlled as before, using the bounds of Prop. \ref{prop:propertiesphit} for $\wt{\ph}_t$ and its derivatives. We conclude that 
\[
\begin{split}
\abs{\la \x_1, [\cK, A_t-A_t^1] \x_2\ra} 
\leq C\expt N^{-3/4} \norm{\cNplusoneh\x_1}\norm{\cNplusone\x_2}.
\end{split}
\]

Let us now consider $[\cV_N, A_t]$. With $[\cVN, \bxxs]= \int ds N^2 \VNxs \bxxs\asss\ass$ we can write 
\[
    \begin{split}
    [\cVN, A_t] 
    = &-\Th(\cN)N^{5/2} \intxy  \VNxy\wmnxy \pty \bxxs\byys [\bgx + \bsxs] +\hc\\
    &+\frac{\Th(\cN)}{\sqrt{N}}\int dxdyds \nu_t(x,y) N^2 \VNxs\bxxs\byys\asss\ass [\bgx + \bsxs]  + \hc\\
    &+\frac{\Th(\cN)}{\sqrt{N}}\int dxdyds \nu_t(x,y) N^2 \VNys \bxxs \byys\asss\ass [\bgx + \bsxs]  + \hc\\
    &-\frac{\Th(\cN)}{\sqrt{N}}\int dxdyds \nu_t(x,y) N^2 \VN (x-s) \bxxs \byys \asss\ass \bxx + \hc\\
    &-\frac{\Th(\cN)}{\sqrt{N}}\int dxdydzds \nu_t(x,y) N^2 \VN (z-s) \overline{p(z,x)} \, \bxxs \byys \asss\ass \bzz + \hc\\
    &+\frac{\Th(\cN)}{\sqrt{N}}\int dxdydzds \nu_t(x,y) N^2 \VN (z-s) \s(z,x)  \,\bxxs \byys \bzzs \asss\ass  + \hc\\
    &+[\cVN, A_t-A_t^1]\\
    \eqqcolon &\sum_{i=1}^6 N_i +\hc +[\cVN, A_t-A_t^1]\,.
    \end{split}
\] 
Proceeding as in (\ref{eq:M11}) to replace the cutoff $\Theta (\cN)$ by 1, we find 
\begin{equation}\label{eq:N1}
N_1= -N^{5/2} \intxy  \VNxy\wmnxy \pty \bxxs\byys\Bxx + \cE'
\end{equation} 
where 
\[
\begin{split}
\abs{\la \x_1, \cE' \, \x_2\ra} &\leq \frac{C\expt}{N^{1/4}} \norm{ (\cVN+\cN+1)^{1/2}\x_1}\norm{ (\cVN+\cN+1)^{1/2}\cNplusoneh\x_2}.
\end{split}
\]
Noticing that $[a_s^* a_s , b^* (\sigma_x)] = \sigma_t (x;s) \, b^*_s$,  and applying Lemma \ref{lm:propeta} to prove that  $|\sigma_t (x;s)| \leq C N$, we find  
\begin{equation}\label{eq:N2-term}
\abs{\la \x_1, N_2 \x_2\ra}\leq C\expt N^{-1/4} 
\norm{\cHNplusNh\x_1}\norm{\cHNplusNh\cNplusoneh\x_2}.
\end{equation}
The term $N_3$, $N_4$ can be bounded similarly. As for $N_5, N_6$, they can be estimated by Cauchy-Schwarz, using Lemma \ref{lm:propeta} to show that $\sup_z \norm{p_z}$, $ \sup_z \norm{\s_z} \leq C\expt$. We find 
\[ \begin{split} 
\abs{\la \x_1, (N_5+N_6) \x_2\ra} &\leq C\expt \frac{\sqrt{mM}}{N} \norm{\cHNplusNh\x_1}\norm{\cHNplusNh\cNplusoneh\x_2}\\ &\leq C\expt N^{-1} 
\norm{\cHNplusNh\x_1}\norm{\cHNplusNh\cNplusoneh\x_2}\,.
\end{split} \]
To bound $[\cVN, A_t-A_t^1]$ we argue as we did above to handle $[\cK, A_t - A_t^1]$, using the identity  \eqref{eq:differenceprojectionA}. Combining all the estimates above, we obtain (\ref{eq:commutatorKA}) (in particular, the large term on the r.h.s. of (\ref{eq:commutatorKA}) emerges summing the r.h.s. of (\ref{eq:M11}) and the r.h.s. of (\ref{eq:N1}). 

The second claim in the Lemma follows analogously, noticing that powers of $(\cN+1)$ can be moved freely from one norm to the other, that in the bounds (\ref{eq:M2-term}), (\ref{eq:N2-term}) we can use the cutoff to remove the operator $(\cN+1)^{1/2}$ applied to $\xi_2$, at the expense of an additional factor $M^{1/2} \leq C N^{1/4}$ and 
that also the main terms on the r.h.s. of (\ref{eq:M11}) and (\ref{eq:N1}) can be controlled (with   Cauchy-Schwarz) by the r.h.s. of (\ref{eq:KA2}). 
\end{proof} 

Now we are ready to show Lemma \ref{lm:gron-A}.
\begin{proof}[Proof of Lemma \ref{lm:gron-A}]
The proof is similar to the proof of \cite[Lemma 5.8]{BSS1}. To show (\ref{eq:gron-NA}), we observe that, since $\cN$ preserves the space $\cF^{\leq N}_{\perp \wt{\ph}_t}$, 
\begin{equation}\label{eq:comNA1}[ \cN, A_t ] = [ \cN , A_t^1 ] =  A_t^{1,\gamma} + 3 A_t^{1,\sigma} + \text{h.c.} \end{equation}
where, as above, $A_t^1$ is defined as $A_t$, but with $\tl{b}, \tl{b}^*$ replaced by $b, b^*$ and where 
\[ A_t^{1,\g}= \frac{\Th(\cN)}{\sqrt{N}}\intxy \nu_t(x,y) \bxxs\byys\bgx, \qquad A_t^{1,\sigma} = \frac{\Th(\cN)}{\sqrt{N}}\intxy \nu_t(x,y) \bxxs\byys\bsxs\,.
\]
From (\ref{eq:comNA1}), it is easy to check that $\pm [\cN, A_t] \leq C\expt (\cN+1)$. With Gr{\"o}nwall's Lemma, we obtain (\ref{eq:gron-NA}) (first with $k=1$ but then, using $[\cN^k, A_t]  = \sum_{j=1}^{k-1}\cN^j [ \cN, A_t ] \cN^{k-j-1}$, also for arbitrary $k \in \bN$). 

Let us now consider (\ref{eq:gron-A}). For $\x\in\Fperpt$ and $s\in [0,1]$ we define 
\[
f_\x(s)= \la \x, e^{-s A_t}  \cH_N \cNplusone^k e^{s A_t}\x\ra.
\]
We have
\begin{equation}\label{eq:f'xi} 
f_\x'(s)=  \left\langle \x, e^{-s A_t} \Big( 
[\cH_N, A_t] \cNplusone^k 
+\cH_N [\cNplusone^k ,A_t]
\Big)e^{s A_t}\x \right\rangle .
\end{equation} 
For $k=0$, Lemma \ref{lm:commutatorHA} implies, together with (\ref{eq:gron-NA}), that 
\begin{equation}\label{eq:gronf}
\abs{f_\x'(s)}\leq C\expt  \big( f_\x(s) + \la \xi , (\cN+1) \xi \ra \big) \,.
\end{equation} 
The desired bound follows therefore from Gr{\"o}nwall's Lemma. 

For $k \geq 1$, the contribution of the term proportional to $[\cH_N, A]$ is bounded similarly; using (\ref{eq:KA2}), with $\xi_1 = e^{s A_t} \xi$, $\xi_2 = (\cN+1)^{k} e^{s A_t} \xi$ and $n=k$, we find 
\[ \begin{split}  \la \x, e^{-s A_t}  [\cH_N, A_t] \cNplusone^k e^{s A_t}\x \rangle &\leq C e^{c|t|} \la e^{s A_t}\x, \cHNplusN\cNplusone^k e^{s A_t}\x\ra \\ &= C e^{c|t|} \, \big( f_\x (s) + \la \xi , (\cN+1)^{k+1} \xi \ra \big).
\end{split} \] 
To handle the second contribution on the r.h.s. of (\ref{eq:f'xi}), we use (\ref{eq:comNA1}). For $k=1$, we write 
\begin{equation}\label{eq:HcNA}
\cHN [\cN  ,A_t] = [\cHN, A_t^{1,\g} +3 A_t^{1,\s}] + \left( (A_t^{1,\g} +3 A_t^{1,\s}) \cHN +\hc\right)
+\cHN [\cN,A_t-A_t^1]\,.
\end{equation}

Contributions arising from the commutator $[\cHN, A_t^{1,\g} +3 A_t^{1,\s}]$ have been already estimated in the proof of Lemma \ref{lm:commutatorHA}. Thus 
\[
\abs{\la \x, e^{-s A_t} [\cHN, A_t^{1,\g} +3 A_t^{1,\s}]
e^{s A_t}\x\ra}
\leq C\expt f_\x(s).
\]

Contributions to the second term on the r.h.s. of (\ref{eq:HcNA}) have the form  
\begin{equation}\label{eq:AAAA}
\begin{split}
A_t^{1,\g}\cK =& \frac{\Th(\cN)}{\sqrt{N}}\intxy \nu_t(x,y) \bxxs\byys\int dz \nabla_z \azzs
\nabla_z \azz\bgx\\
&+\frac{\Th(\cN)}{\sqrt{N}}\intxy \nu_t(x,y) \bxxs\byys (-\Delta_x\bxx + b(-\Delta p_x)),\\
A_t^{1,\s}\cK =& \frac{\Th(\cN)}{\sqrt{N}}\intxy \nu_t(x,y) \bxxs\byys\bsxs \int dz \nabla_z \azzs
\cNplusoneto{-1}\cNplusone\nabla_z \azz ,\\
A_t^{1,\g}\cVN =& \frac{\Th(\cN)}{2\sqrt{N}}\intxy \nu_t(x,y) \bxxs\byys\int drdz N^2\VN (r-z)\arrs\azzs
\arr \azz\bgx\\
&+\frac{\Th(\cN)}{\sqrt{N}}\intxy \nu_t(x,y) \bxxs\byys \int drdz N^2\VN (r-z)\arrs
\arr \bzz \overline{\g(z,x)},\\ 
A_t^{1,\s}\cVN =& \frac{\Th(\cN)}{\sqrt{N}}\intxy \nu_t(x,y) \bxxs\byys\bsxs \int drdz N^2\VN (r-z)\arrs\azzs
\arr \azz. \end{split} \end{equation}
The two commutators contributing to $A_t^{1,\g} \cK$, $A_t^{1,\g} \cV_N$ can be bounded similarly to the terms $M_2, M_3$ and, respectively, $N_4,N_5$ in the proof of Lemma \ref{lm:commutatorHA}. All other terms can be bounded directly with Cauchy-Schwarz; for instance
\[
\begin{split}
&\Big| \Big\langle e^{s A_t}\x , \frac{\Th(\cN)}{\sqrt{N}}\intxy \nu_t(x,y) \, \bxxs\byys\int dz \nabla_z \azzs
\nabla_z \azz\bgx e^{s A_t}\x \Big\rangle \Big| \\
&\leq \frac{C\expt}{\sqrt{N}}\left(\int dxdydz \norm{\axx\ayy \nabla_z \azz \Th(\cN)e^{s A_t}\x}^2 \right)^{1/2}\\
&\hspace{4cm}
 \times\left(\int dxdydz N^2\wmnxy^2 \norm{\agx\nabla_z \azz e^{s A_t}\x}^2 \right)^{1/2}\\
&\leq C\expt \norm{\cHNplusoneh\cNplusoneh e^{s A_t}\x}^2.
\end{split}
\]
The hermitian conjugates of the terms in (\ref{eq:AAAA}) satisfy similar estimates. We handle the last term on the r.h.s. of (\ref{eq:HcNA}) analogously as we handled the terms proportional to $A_t - A_t^1$ in the proof of Lemma \ref{lm:commutatorHA}, using the identity (\ref{eq:differenceprojectionA}). Observe here that on $\Fperpt\times\Fperpt$ we have $\cH_N b^* (\ph_t)= [\cH_N, b^* (\ph_t)]$ allowing to recover a commutator which has been estimated before.
Thus, we find that (\ref{eq:gronf}) also holds true for $k=1$. For $k\geq 2$, we write 
\[
\begin{split}
 \cH_N [\cNplusone^k ,A_t] 
 &= \sum_{j=1}^k (\cN+1)^j \cHN [\cN ,A_t] (\cN+1)^{k-j-1} 
\end{split}
\]
and we argue similarly as we did for $k=1$, after appropriately pulling factors of $(\cN+1)^{1/2}$ through the commutator $[\cN, A_t]$. We obtain again the estimate (\ref{eq:gronf}). The desired bounds follows 
by Gr{\"o}nwall's Lemma.
\end{proof}

\section{Generator of fluctuation dynamics: proof of Prop. \ref{prop:cJNt}} 
\label{sec:proof-prop}

In this section, we study the properties of the generator $\cJ_N (t)$ of the full fluctuation dynamics (\ref{eq:fluct-dyn2}). We start from the expression (\ref{eq:genJ}). A first step in the proof of Prop.  \ref{prop:cJNt} consists in applying the rules (\ref{eq:U-rules}) to compute the generator 
\[ \cL_N (t) = (i\partial_tU_{N,t})U^*_{N,t} +  U_{N,t} H_N U^*_{N,t}\,. \]
We find $\cL_{N} (t) = \sum_{j=0}^{4}\cL_{N,t}^{(j)}$ where as quadratic forms on $\Fperpt\times\Fperpt$
	\begin{equation}
		\label{eq:cLNt}
	\begin{split}
		\cL_{N,t}^{(0)} &  = \frac{1}{2} \langle \wt{\ph}_t, [N^3V_N(1-2f_{N,\ell})*|\wt{\ph}_t|^2]\wt{\ph}_t\rangle (N-\cN) \\
		&\quad -\frac{1}{2} \langle \wt{\ph}_t, [N^3V_N*|\wt{\ph}_t|^2]\wt{\ph}_t\rangle (\cN +1)\frac{(N-\cN)}{N}\\
		\cL_{N,t}^{(1)} & = \sqrt{N} b([N^3V_Nw_{N,\ell}*|\wt{\ph}_t|^2] \wt{\ph}_t ) -\frac{\cN+1}{\sqrt{N}} b([N^3V_N*|\wt{\ph}_t|^2 ] \wt{\ph}_t) + \hc\\
		\cL_{N,t}^{(2)} & =\cK + \int dxdy\, N^3V_N(x-y)|\wt{\ph}_t(y)|^2\Big(b_x^*b_x -\frac{1}{N}a_x^*a_x\Big)\\
		&\quad + \int dxdy\, N^3V_N(x-y)\wt{\ph}_t(x)\bar{\wt{\ph}}_t(y)\Big(b_x^*b_y -\frac{1}{N}a_x^*a_y\Big)\\
		& \quad + \frac 12 \left[\int dxdy\, N^3V_N(x-y)\wt{\ph}_t(x)\wt{\ph}_t(y)\big(b_x^*b_y^* +\hc\big)\right]\\
		\cL_{N,t}^{(3)} & = \int dx dy N^{5/2}V_N(x-y)\wt{\ph}_t(y)b_x^*a_y^*a_x + \hc\\
		\cL_{N,t}^{(4)} & =  \cV_N\,, \\
	\end{split}
	\end{equation}
where we recall that $\cK$ and $\cV_N$ are the kinetic and the potential energy operators, as defined on $\cF_{\perp \wt{\ph}_t}^{\leq N}$ in (\ref{eq:KVN}). 

Next, we have to consider the effect of the quadratic conjugation with the generalized Bogoliubov transformation $e^{B_t}$. We define the renormalized generator 
\begin{equation}\label{eq:cGNt} \cG_N (t) =(i\partial_te^{-B_t})e^{B_t} +e^{-B_t}\cL_N (t) e^{B_t}\,. \end{equation} 
In order to describe the operator $\cG_N (t)$, we define 
\begin{equation}
	\label{eq:phasek1}
	\begin{split}
		\k_{\cG}(t) =& \frac{N}{2} \la \ptt, [N^3\VN(1-2\fln)*|\ptt|^2] \ptt \ra
		-\frac{1}{2} \la \wt{\ph}_t, (N^3 V_Nf_{N,\ell}\ast \vert \wt{\ph}_t^2\vert ) \wt{\ph}_t\ra \\
		&+ \frac{1}{2}  \intxy N^2 \VNxy \vert \la \g_y,\s_x \ra\vert^2  \\
		&+\int dx\, (N^3\VN*|\wt{\ph}_t|^2)(x)
		\langle \sxt, \sxt\rangle \\
		&+ \int dxdy\, N^3\VN(x-y)\wt{\ph}_t(x)\overline{\wt{\ph}_t(y)} \langle \sxt, \s_y\rangle \\
		&+\frac 12 \int dxdy\, N^3\VN(x-y)\wt{\ph}_t(x)\wt{\ph}_t(y)\big \langle\s_x,\g_y\rangle+ \hc + \|\nabla_1\s_t\|^2 \, \\
		&- \int_0^1ds \int dxdy\, \dot{\eta}_t(x,y)\la\s_x^{(s)},\g_y^{(s)}\ra +\hc\,
\end{split}\end{equation}
and the quadratic operator 
\begin{equation}
	\label{eq:quadgenfinal}
	\begin{split}
		\cG_{2,N}(t)=&\, \cG_{2,N}^{K}(t) + \cG_{2,N}^{V}(t)   \\
		& + N^3 \l_\ell\int dxdy\, \c_\ell(x-y)\wt{\ph}_t(x)\wt\ph_t(y)b^*_xb^*_y +\hc\\
		&+\frac12\int dxdy\, Nw_{N,\ell}(x-y)[\D\wt\ph_t(x)\wt\ph_t(y) +\wt\ph_t(x)\D\wt\ph_t(y)]b^*_xb^*_y +\hc\\
		&+ \int dxdy\, N\nabla w_{N,\ell}(x-y)[\nabla\wt\ph_t(x)\wt\ph_t(y) - \wt\ph_t(x)\nabla\wt\ph_t(y)]b^*_xb^*_y 
		+\hc\\
		&-\int_0^1ds \int dxdy\, \dot{\eta}_t(x,y)  \\
		&\hskip.2cm \times \big[ b^*(\g_x^{(s)})b^*(\g_y^{(s)})+b^*(\g_x^{(s)})b(\s_y^{(s)})+b(\s_x^{(s)})b(\s_y^{(s)})+ b^*(\g_y^{(s)})b(\s_x^{(s)})  \big] +\hc 
\end{split}\end{equation}
where 
\begin{equation}
	\label{eq:genquadraticK}
	\begin{split}
		\cG_{2,N}^{K}(t) =&\,\cK 
		+ \int dx \big[b^*(-\D_x p_x)b_x
		+ \frac{1}{2} b^*(\nabla_x p_x) b(\nabla_x p_x) 
		+ \bxxs b^*(-\D_x \mu_x)\\
		&+ b^*(-\D_xp_x)b^*(\eta_x) 
		+ b^*(p_x)b^*(-\D_x r_x) 
		+ b_x^* b^*(-\D_x r_x)\\
		&+\frac{1}{2} b^*( \nabla_x \eta_x) b(\nabla_x \eta_x)
		+ b^*(\eta_x)b(-\D_x r_x)
		+ \frac{1}{2} b^*(\nabla_x r_x)b(\nabla_x r_x) +\hc \big]\\
\end{split}
\end{equation}
and 
\begin{equation}
	\label{eq:genquadraticV}
	\begin{split}
		\cG_{2,N}^{V}(t) =&\,
		\frac 12 \int dxdy\, N^3(V_Nf_{N,\ell})(x-y)\wt{\ph}_t(x)\wt{\ph}_t(y)\big[b^*(p_{x})b^*_{y} + b^*(\g_{x})b^*(p_{y})\big] + \hc \\
		&+\frac 12 \int dxdy\, N^3(V_Nf_{N,\ell})(x-y)\wt{\ph}_t(x)\wt{\ph}_t(y)  \\
		&\hspace{1cm} \times \big[ b^*(\g_{y})b(\s_{x}) + b^*(\g_{x})b(\s_{y})+ b(\s_{x})b(\s_{y}) \big]+ \hc\\
		&+\int dx\, (N^3V_N*|\wt{\ph}_t|^2)(x)  \\
		&\hspace{1cm} \times \big(b^*(\gxt)b(\gxt) +b(\sxt)b(\gxt)+ b^*(\gxt)b^*(\sxt)+ b^*(\sxt)b(\sxt) \big)\\
		& +\int dxdy\, N^3V_N(x-y)\wt{\ph}_t(x)\overline{\wt{\ph}_t(y)}\\
		&\hspace{1cm} \times \big(b^*(\g_x)b(\g_{y}) + b(\sxt)b(\g_{y})+ b^*(\gxt)b^*(\s_{y})  +  b^*(\syt)b(\sxt)   \big)\,. \\
\end{split}\end{equation}

The following proposition establishes, up to negligible errors, the form of $\cG_N (t)$, in terms of (\ref{eq:phasek1}), (\ref{eq:genquadraticK}), (\ref{eq:genquadraticV}). 
\begin{prop}\label{prop:cGNt} 
Let $V \in L^3(\bR^3)$ non-negative, spherically symmetric, and  compactly supported.
 Let $\cG_{N}(t)$ be defined as in Eq. \eqref{eq:cGNt}. 
	Assume $\ell$ in \eqref{eq:defk} is small enough but of order one in $N$. Let 
	\begin{equation}
		\label{eq:cubic}
		\cC_{N,t} = \intxy N^{5/2} \VNxy\wt\ph_t(y) \bxxs\byys[\bgx + \bsxs]\, + \text{h.c.} 
	\end{equation}	
	Then on $\Fperpt\times\Fperpt$
	\begin{equation}\label{eq:cGNt-prop}
		\cG_{N}(t) = \k_{\cG}(t) +  \cG_{2,N}(t) + \cC_{N,t}+ \cVN + \cE_{\cG_N} (t),
	\end{equation}
	where the phase $\k_{\cG}(t)$, $\cG_{2,N}(t)$ and $\cV_N$ are given in Eq. \eqref{eq:phasek1},\eqref{eq:quadgenfinal}, \eqref{eq:cLNt} respectively, and the error term $\cE_{\cG_N} (t)$ satisfies
	\begin{equation}\label{eq:cEGN} \begin{split}|\langle \xi_1, \cE_{\cG_N} (t)\xi_2\rangle| & \leq \frac{C\expt}{\sqrt{N}} \norm{\cHNplusNh\x_1}\norm{(\cH_N + \cN^3 + 1)^{1/2} (\cN+1) \xi_2} \end{split}\end{equation} 
	for any $\xi_1,\xi_2 \in \Fperpt$. 		
\end{prop}  

Prop. \ref{prop:cGNt} is analogous to \cite[Prop. 2.5, part b]{BSS1} in the time-independent setting, where the solution of the Gross-Pitaevskii equation (\ref{eq:GPmod}) (entering the definition of $\cL_N (t)$ and $\cG_N (t)$) is replaced by the minimizer of the Gross-Pitaevskii energy functional \footnote{Notice that the energy functional used in \cite{BSS1} was the limiting Gross-Pitaevskii functional, defined in terms of the scattering length $\frak{a}$. In the present paper, on the other hand, we find it more convenient to work with the modified $N$-dependent Gross-Pitaevskii equation (\ref{eq:GPmod}); as a consequence, some of the expressions emerging in the computation of $\cG_N (t)$ are slightly different from the corresponding expressions in \cite{BSS1})}. The main difference is the fact that Prop. \ref{prop:cGNt} provides bounds for general matrix elements, while \cite[Prop. 2.5, part b]{BSS1} only deals with expectations. Since the two proofs are very similar, we only sketch the main steps in the proof of Prop. \ref{prop:cGNt}, referring to \cite{BSS1} for all details (we give some more details for the term $(i\partial_te^{-B_t})e^{B_t}$, which is absent in \cite{BSS1}).  

\begin{proof}[Sketch of proof of Prop. \ref{prop:cGNt}]  
We write $\cG_N (t) = (i\partial_te^{-B_t})e^{B_t} + \sum_{j=0}^4 \cG^{(j)}_N (t)$, with $\cG^{(j)}_N (t) = e^{-B_t} \cL^{(j)}_N (t) e^{B_t}$, for $j=0,\dots , 4$.  We have 
\begin{equation} \label{eq:cGN0} \begin{split}
	\cG_{N,t}^{(0)} &= \frac{N}{2} \la \ptt, [N^3\VN(1-2\fln)*|\ptt|^2] \ptt \ra - \frac{1}{2} \la \ptt, [N^3\VN*|\ptt|^2] \ptt \ra \\
	&\quad - \la \ptt, [N^3\VN\wln*|\ptt|^2] \ptt \ra \\
	&\hspace{0.5cm} \times\left( \int dx [\bgxs\bgx + \bsxs\bsx + \bgxs\bsxs +\bgx\bsx] +\norm{\s_t}^2\right) +\cE_{N,t}^{(0)}
\end{split} \end{equation} 
where 
\begin{equation*}
		\label{eq:errorcG0}
		\begin{split}
		\vert \la \x_1, \cE_{N,t}^{(0)}\, \x_2\ra \vert & \leq \frac{C}{\sqrt{N}}\norm{\cNplusoneh\x_1} \norm{\cNplusone\x_2}
\end{split}	
\end{equation*}
for all $\xi_1, \xi_2 \in \cF_{\perp \wt{\ph}_t}^{\leq N}$ and all $t \in \bR$. Details can be found in \cite[Cor. 4.3]{BSS1}.  
	
Proceeding as in \cite[Lemma 4.4]{BSS1}, we find 
\begin{equation}\label{eq:cGN1} 
\cG_{N,t}^{(1)} = 
\sqrt{N}\big[b(\g(h_{N,t})) + b^*(\s(\bar{ h}_{N,t})) +\hc\big] + \cE_{N,t}^{(1)}\,.
\end{equation}
with $h_{N,t} = (N^3V_Nw_{N,\ell}*|\wt{\ph}_t|^2) \wt{\ph}_t$ and 
\begin{equation*}
		\label{eq:errorcG1}
		\begin{split}
		|\langle \xi_1 , \cE_{N,t}^{(1)}\, \xi_2\rangle| &\leq  \frac{C}{\sqrt N} \|(\cN+1)^{1/2}\xi_1\|\|(\cN+1) \xi_2\|\,.\\
		\end{split}
	\end{equation*}

We decompose $\cG_{N}^{(2)} (t) = e^{-B_t} \cK e^{B_t} + \cG_N^{(2,V)} (t)$. Following \cite[Lemma 4.6]{BSS1}, we obtain 
\begin{equation}\label{eq:cGN2V}  
\begin{split} 
\cG_{N,t}^{(2,V)} &=\,	\int dx\, (N^3V_N*|\wt{\ph}_t|^2)(x)  \\
	&\hspace{.5cm} \times \big(b^*(\gxt)b(\gxt) +b(\sxt)b(\gxt)+ b^*(\gxt)b^*(\sxt)+ b^*(\sxt)b(\sxt) + \langle \sxt, \sxt\rangle  \big)\\
	&\quad+ \int dxdy\, N^3V_N(x-y)\wt{\ph}_t(x)\bar{\wt{\ph}}_t(y) \\
	&\hspace{.5cm} \times \big(b^*(\gxt)b(\gyt) +b(\sxt)b(\gyt)+ b^*(\gxt)b^*(\syt)  
	+ b^*(\syt)b(\sxt) + \langle \sxt, \syt\rangle  \big) \\
	&\quad+\frac 12 \int dxdy\, N^3\VN(x-y)\wt{\ph}_t(x)\wt{\ph}_t(y)  \\
	&\hspace{.5cm} \times \Big[b^*(\gyt)b^*(\gxt) + b^*(\gyt)b(\sxt)+ b^*(\gxt)b(\syt)+ b(\sxt)b(\syt) + \langle\sxt,\gyt\rangle(1-\cN/N) \\
	&\hspace{3cm}+ d^*_x\big(b^*(\gyt)+ b(\syt)\big) + \big(b^* (\gxt)+ b (\sxt)\big)d^*_y + \hc\Big] + \cE_{N,t}^{(2,V)} 
\end{split}
\end{equation} 
with 
\[ \begin{split} |\langle \xi_1, \cE_{N,t}^{(2,V)} \xi_2 \rangle | &\leq \frac{C}{N^{1/2}}e^{c|t|}\|(\cV_N + \cN+1)^{1/2} \xi_1\|\|(\cV_N + \cN+1)^{1/2} (\cN+1)^{1/2} \xi_2\| \end{split}  \]
where $d_x, d_x^*$ are the operator valued distributions associated with the fields $d_t$ defined in \eqref{eq:defd}. 
On the other hand, analogously to \cite[Lemma 4.8]{BSS1},  we conclude that  
\begin{equation}\label{eq:cGNK}  
\begin{split} \cG_{N,t}^{(K)} =& \,\cK + \int dxdy  \, (-\D_x \eta_t) (x,y) \, b^*_x b^*_y + \int dxdy \overline{(-\D_x\eta_t) (x,y)} \, b_x b_y \\
&+ \int dx \Big[ b^*(-\D_xp_x)b_x+  b^*_xb(-\D_xp_x)+\nabla_x b^*(p_x)\nabla_xb(p_x) \\
& \hspace{1cm} + b^*(-\D_xp_x)b^*(\eta_x) + b^*(p_x)b^*(-\D_xr_x) + b_x^*b^*(-\D_xr_x)\\
& \hspace{1cm} + b(-\D_xr_x)b_x + b(-\D_xr_x)b(p_x) + b(\eta_x)b(-\D_xp_x)\\
&\hspace{1cm}  +\nabla_x b^*(\eta_x)\nabla_xb(\eta_x)+ b^*(\s_x)b(-\D_xr_x)+ b^*(-\D_xr_x)b(\eta_x)\Big]\\
& + \int dx \big[b(\nabla_x\eta_x)\nabla_x d_x + \nabla_x d_x^* \, b^* (\nabla_x \eta_x) \big] + \|\nabla_1\s_t \|^2(1-\cN/N) + \frac 1N \|\nabla_1\eta_t\|^2\\
& +\frac{\|\nabla_1\eta_t\|^2}{N} \Big(\int dx \big[b^*(\g_x)b(\g_x)+b^*(\s_x)b(\s_x)+b^*(\g_x)b^*(\s_x)+b(\g_x)b(\s_x)\big]\Big)\\
&+ \frac{\|\nabla_1\eta_t\|^2}{N} \|\s_t \|^2+ \cE_{N,t}^{(\cK)}\,,\\
\end{split}\end{equation}
with 
\[ \begin{split} 
 |\langle \xi_1, \cE_{N,t}^{(K)} \xi_2 \rangle |  \leq  &\,
		\frac{C}{N^{1/2}}e^{c|t|} \|(\cK+\cN+1)^{1/2}\xi_1\| \|(\cK+\cN^3+1)^{1/2}(\cN+1) \xi_2\|\,.
	\end{split} \]

As for the cubic term, we find
\begin{equation}\label{eq:cGN3}  
\begin{split}
\cG_{N,t}^{(3)} = \intxy N^{5/2} \VNxy\ptt(y) \bxxs\byys[\bgx + \bsxs]\\
	- \sqrt{N}\big(b(\g ( h_{N,t})) + b^*(\s (\overline{ h_{N,t}})) \big)+\hc + \cE^{(3)}_{N,t},
\end{split}
\end{equation} 
with 
\[ \begin{split} |\langle \xi_1, \cE_{N,t}^{(3)} \, \xi_2 \rangle | \
	& \leq\frac{Ce^{c|t|}}{\sqrt{N}} 
	\norm{\cHNplusNh\x_1}\norm{\cHNplusNh\cNplusone\x_1}\,.
	\end{split} \]
The proof of (\ref{eq:cGN3}) is very similar to the proof of \cite[Lemma 4.9]{BSS1}. 

Furthermore, following the strategy used in \cite[Lemma 4.10]{BSS1}, we obtain 
\begin{equation}\label{eq:cGN4}
\begin{split}   
\cG_{N,t}^{(4)} = &\cV_N
	+ \frac{1}{2}  \intxy N^2 \VNxy \vert \la \g_y,\s_x \ra\vert^2  \left(1 +\frac{1}{N}- \frac{2\cN}{N}\right)\\
	&\quad+\frac{1}{2} \intxy N^2 \VNxy [\bgxs\bgys
	+\bgxs\bsy
	+\bgys\bsx\\
	&\hspace{7cm}
	+\bsx\bsy] k_t(x,y) + \hc\\
	&\quad+\frac{1}{2} \intxy N^2 \VNxy [
	\dxxs(\bgys+\bsy)
	+(\bgxs+\bsx)\dyys
	] k_t(x,y)+\hc\\
	&\quad+\intxy N^3 (V_N \wln^2)(x-y)\vert \wt\ph_t(x) \vert^2 \vert \wt\ph_t(y) \vert^2\\
	&\hspace{1cm}\times \left( \int du [\bgus\bgu + \bsus\bsu + \bgus\bsus +\bgu\bsu] +\norm{\s}^2 \right )\\
	&\quad+ \cE_{N,t}^{(4)}\,.
\end{split} 
\end{equation} 
with 
\[ \begin{split} |\langle \xi_1, \cE_{N,t}^{(4)}  \, \xi_2 \rangle |  \leq
	&\,\frac{C e^{c|t|}}{\sqrt{N}} \norm{\cHNplusNh\x_1}\norm{\cHNplusNh\cNplusone\x_1}\,.
\end{split} \]

Finally, we claim that 
\begin{multline}\label{eq:partialexpB}
(i\partial_te^{-B_t})e^{B_t} =  -\int_0^1ds \int dxdy\, \dot{\eta}_t(x,y)  \;\Big[b^*(\g_x^{(s)})b^*(\g_y^{(s)})+b^*(\g_x^{(s)}) b(\s_y^{(s)})+b(\s_x^{(s)})b(\s_y^{(s)}) \\+ b^*(\g_y^{(s)})b(\s_x^{(s)}) + \la\s_x^{(s)},\g_y^{(s)}\ra \Big] +\hc  + \cE_{\partial_t}\,, 
\end{multline}
where $\cE_{\partial_t}$ satisfies 
\begin{equation}\label{eq:partialtB} \begin{split}
    |\la \xi_1,\cE_{\partial_t} \xi_2\ra|
    \leq &\, \frac CNe^{c|t|}\|(\cN+1)^{1/2}\xi_1\|\|(\cN+1)^{3/2} \xi_2\|\,,\end{split}\end{equation} 
for any $\xi_1,\xi_2 \in \Fperpt$, $t \in \bR$, and $N \in \bN$ large enough.

To prove (\ref{eq:partialexpB}), we use \eqref{eq:defd} to expand   
\[ \begin{split}
(i\partial_te^{-B_t})e^{B_t} &= \int_0^1 ds \, e^{-sB_t} \, \big( i\partial_t B_t \big) e^{s B_t} \\ &= -\int_0^1ds \int dxdy\, \dot{\eta}_t(x,y) \big[b^*(\g_x^{(s)})b^*(\g_y^{(s)})+b^*(\g_x^{(s)})b(\s_y^{(s)})+b(\s_x^{(s)})b(\s_y^{(s)}) \\ &\hspace{5cm} + b^*(\g_y^{(s)})b(\s_x^{(s)}) + \la\s_x^{(s)},\g_y^{(s)}\ra \big] +\hc  + \cE_{\partial_t}
\end{split} \] 
where the superscript $s$ denotes that we consider cosh and sinh of $s\eta_t$ instead of $\eta_t$.
Here, $\cE_{\partial_t}$ is given by 
\[ \begin{split} \cE_{\partial_t} = \; &-\int_0^1ds \int dxdy\, \dot{\eta}_t(x,y) \\ &\times \big[d_{x,s}^* (b^* (\gamma_y^{(s)}) + b (\sigma^{(s)}_y)) + (b^* (\gamma^{(s)}_x) + b (\sigma^{(s)}_x)) d_{y,s}^* +d^*_{x,s}d^*_{y,s} - \la\s_x^{(s)},\g_y^{(s)}\ra\cN/N\big] \\ &+\hc\, \end{split} \]
where $d_{x,s}, d_{x,s}^*$ are the operator valued distributions associated with operators $d_t^{(s)}$, which are defined as $d_t$ in (\ref{eq:defd}), but with $\eta_t$ replaced by $s \eta_t$.  
With (\ref{eq:d-bds}) and with the bounds from Lemma \ref{lm:propeta}, we have
\[\begin{split}
    |\la \xi_1,\cE_{\partial_t} \xi_2\ra| \leq\,
    &\frac{C e^{c|t|}}{N}\intxy \,|\dot\eta_t(x,y)|\|(\cN+1)^{1/2} \xi_1\| \\ & \times \Big[ (1 + |\eta_t(x,y)|) \|(\cN+1)^{3/2}\xi_2\| + \|a_x (\cN+1)\xi_2 \| + \| a_x a_y(\cN+1)^{1/2}\xi_2 \|\Big] \\
    \leq & \frac{C e^{c|t|}}{N} \| (\cN+1)^{1/2} \xi_1 \| \| (\cN+1)^{3/2} \xi_2 \|
\end{split}\]
which concludes the proof of (\ref{eq:partialtB}). 

Collecting the contributions on the r.h.s. of (\ref{eq:cGN2V}), (\ref{eq:cGNK}) and (\ref{eq:cGN4}) containing the operators $d, d^*$, we define
\[ \begin{split}  \cD_{N,t} = \; &\int dx \, \big[ b (\nabla_x \eta_x) \nabla_x d_x + \text{h.c.} \big] \\ &+ \frac{1}{2} \int dx dy \, N^3 (V_N  f_{N,\ell}) (x-y) \wt{\ph}_t (x) \wt{\ph}_t (y) \\ &\hspace{4cm} \times \big[ d_x^* (b^* (\gamma_y) + b (\sigma_y)) + (b (\gamma_x) + b^* (\sigma_x)) d_y^* + \text{h.c.}  \big] \end{split} \]

Proceeding as in \cite[Section 4.5]{BSS1}, we obtain 
\begin{equation}\label{eq:cDNt} \begin{split}  \cD_{N,t} = \;&  \frac 12 \int dx dy N^3(V_Nf_{N,\ell} w_{N,\ell})(x-y) |\wt\ph_t(x)|^2|\wt\ph_t(y)|^2\\
 &\hspace{3cm}\times\int dv \big[b^*(\s_v)b^*(\g_v) +b(\s_v)b(\g_v) +2b^*(\s_v)b(\s_v)\big]\\
 & + \frac 12 \int dx dy N^3(V_N\wln \fln)(x-y)|\wt\ph_t(x)|^2|\wt\ph_t(y)|^2\|\s\|^2 + \cE_{N,t}^{(D)} 
 \end{split} \end{equation} 
with the error 
\[|\langle \xi_1,\cE_{N,t}^{(D)} \xi_2\rangle|\leq \frac CNe^{c|t|} \|(\cN+1)^{1/2}\xi_1\| \|(\cVN + \cN^2+1)^{1/2}(\cN+1)^{1/2}\xi_2\| \,.\]
To obtain the constant term \eqref{eq:phasek1}, we first recombine the following terms, appearing in \eqref{eq:cGN0}, \eqref{eq:cGNK} and \eqref{eq:cGN4}: 
\[
\frac1N\|\nabla_1\eta_t\|^2 +\frac{1}{2N}  \intxy N^2 \VNxy \vert \la \g_y,\s_x \ra\vert^2 - \frac{1}{2}\la\wt\ph_t, [N^3V_N*|\wt\ph_t|^2]\wt\ph_t\ra\,.
\]
Using the scattering equation \eqref{eq:scatlN} and the estimate $|\la\s_x,\g_y\ra -k_t(x,y)|\leq C|\wt\ph_t(x)||\wt\ph_t(y)|$, we end up with the second summand on the r.h.s. of \eqref{eq:phasek1} (up to an error of order $1/N$). Moreover, observing that 
\begin{equation}\label{eq:cancellations}
    \begin{split}
 \Big| &\frac{1}{2}\intxy N^3(V_Nw_{N,\ell}f_{N,\ell})(x-y) |\wt\ph_t(x)|^2|\wt\ph_t(y)|^2 + \frac{1}{N}\|\nabla_1\eta_t\|^2 \\
        &+ \intxy N^3(V_Nw_{N,\ell}^2)(x-y) |\wt\ph_t(x)|^2|\wt\ph_t(y)|^2   -\la\wt\ph_t, N^3V_Nw_{N,\ell}*|\wt\ph_t|^2\wt\ph_t\ra\Big|\leq \frac{C}{N}\,,
    \end{split}
\end{equation}
we conclude that the contributions proportional to $\| \s_t \|^2$ from Eq.  \eqref{eq:cGN0}, \eqref{eq:cGNK}, \eqref{eq:cGN4} and \eqref{eq:cDNt} cancel (again, up to an error of the order $1/N$); this leads to \eqref{eq:phasek1}. 

To derive (\ref{eq:quadgenfinal}), we use again (\ref{eq:cancellations}) to show that the sum of all terms  proportional to $\cN$ and all terms proportional to 
\[  \int du \big[ b^* (\g_u)b(\g_u) + b^* (\s_u) b(\s_u) + b^* (\g_u) b^* (\s_u) + b (\g_u) b(\s_u) \big] \]
appearing on the r.h.s. of \eqref{eq:cGN0}, \eqref{eq:cGN2V}, \eqref{eq:cGNK}, \eqref{eq:cGN4}, \eqref{eq:cDNt} produces a small error, which can be absorbed in $\cE_{\cG_N}$ (in \eqref{eq:cDNt}, we first decompose $2 \s_t\bar\s_t= \gamma^2_t + \s_t\bar\s_t - \mathbbm{1}$). Furthermore, in the terms on the first line on the r.h.s. of (\ref{eq:cGNK}), we split $\eta_t = \mu_t + k_t$ and we combine the contribution associated with $k_t$ with the contributions  
\[ \begin{split} 
&\frac{1}{2} \int dx dy N^3 V_N (x-y) \wt{\ph}_t (x) \wt{\ph}_t (y) b^*_x b^*_y \, ,  \qquad \frac{1}{2} \int dx dy \, N^2 V_N (x-y) k_t (x;y) b_x^* b_y^*  \end{split} \]
extracted from the third summands on the r.h.s. of (\ref{eq:cGN2V}) and on the r.h.s. of (\ref{eq:cGN4}) (expanding $\gamma_t = \mathbbm{1} + p_t$). Observing that 
\begin{equation*}\label{eq:kernqterms}
\begin{split}
	N \big(\D_x +\D_y\big)w_{N,\ell}(x-y) &\wt\ph_t(x)\wt\ph_t (y)
	+ N^3(V_Nf_{N,\ell})(x-y)\wt{\ph}_t(x)\wt{\ph}_t(y)  \\
	&= N[2\D w_{N,\ell}(x-y) +  N^2(V_Nf_{N,\ell})(x-y)]\wt\ph_t(x)\wt\ph_t(y) \\
	&\quad + Nw_{N,\ell}(x-y)[\D\wt\ph_t(x)\wt\ph_t(y) +\wt\ph_t(x)\D\wt\ph_t(y)]\\
	&\quad+ 2N\nabla w_{N,\ell}(x-y)[\nabla\wt\ph_t(x)\wt\ph_t(y) - \wt\ph_t(x)\nabla\wt\ph_t(y)]
	\end{split}\end{equation*}
and that, from the scattering equation (\ref{eq:scatlN}), 
\[\begin{split}N[2\D w_{N,\ell}(x-y) + N^2(V_Nf_{N,\ell})(x-y)]\wt\ph_t(x)\wt\ph_t(y) =&\, 2N^3 \l_\ell f_{N,\ell}\c_\ell(x-y)\wt\ph_t(x)\wt\ph_t(y) 
\end{split}\]
we conclude that 
\begin{equation}\label{eq:deltaeta} \begin{split} 
\int dx dy \, &\Big[ \frac12(-\Delta_x-\D_y) \eta_t (x;y) + \frac{1}{2}  N^3 V_N (x-y) \wt{\ph}_t (x) \wt{\ph}_t (y) + \frac{1}{2} N^2 V_N (x-y) k_t (x;y) \Big] \, b_x^* b_y^* \\ = \; &\int dx dy \, (-\Delta \mu_t) (x;y) b_x^* b_y^* + N^3 \lambda_\ell \int dx dy \, \chi_\ell (x-y) \, f_{N,\ell} (x-y) \wt{\ph}_t (x) \wt{\ph}_t (y) b_x^* b_y^* \\ &+ \frac12\int dx dy N w_{N,\ell} (x-y) \big[ \Delta \wt{\ph}_t (x) \wt{\ph}_t (y) + \wt{\ph}_t (x) \Delta \wt{\ph}_t (y) \big] b_x^* b_y^* \\ &+ 
\int N \nabla w_{N,\ell} (x-y) \big[ \nabla \wt{\ph}_t (x) \wt{\ph}_t (y) - \wt{\ph}_t (x) \nabla \wt{\ph}_t (y) \big] b_x^* b_y^* \,. \end{split} \end{equation}
In the second term on the r.h.s. we can write $f_{N,\ell} = 1 - w_{N,\ell}$ and check that the contribution of $w_{N,\ell}$ 
can be absorbed in the small error. Thus, combining (\ref{eq:cGN0}), (\ref{eq:cGN1}), (\ref{eq:cGN2V}), (\ref{eq:cGNK}), (\ref{eq:cGN3}), (\ref{eq:cGN4}) and (\ref{eq:partialexpB}) with (\ref{eq:cDNt}) and (\ref{eq:deltaeta}) and comparing with the definitions (\ref{eq:phasek1}), (\ref{eq:quadgenfinal}), (\ref{eq:cubic}), we obtain the claim of Prop. \ref{prop:cGNt}. 
 \end{proof}

With (\ref{eq:cGNt}), we can rewrite the generator $\cJ_N (t)$ defined in (\ref{eq:fluct-dyn}) as 
\begin{equation}\label{eq:cJN2}  \cJ_N (t) = (i\partial_t e^{-A_t}) e^{A_t} + e^{-A_t} \cG_N (t) e^{A_t} \end{equation}
To conclude the proof of Prop. \ref{prop:cJNt}, we need therefore to control the action of $A_t$ on the terms on the r.h.s. of (\ref{eq:cGNt-prop}). We already have some information about the action of $A_t$ on the Hamilton operator $\cH_N = \cK + \cV_N$, thanks to Lemma \ref{lm:commutatorHA}. The action of $A_t$ on the quadratic terms in $\cG_{2,N} (t)$ (excluding the kinetic energy operator $\cK$) is determined by the next lemma. 
\begin{lemma}\label{lm:cubicquadraticterms}
Let $A_t$ be defined as in (\ref{eq:defA}). Let $F\colon\RRR^3\times\RRR^3\to\mathbb{C}$. For any $\x_1, \x_2\in\Fperpt$ we have
\begin{equation}\label{eq:F2bd} \begin{split} \Big| \int dr ds F(r,s) & \left\langle \xi_1, \big[ b_r^* b^*_s , A_t \big] \xi_2 \right\rangle \Big| , \Big| \int dr ds \bar{F} (r,s)  \, \left\langle \xi_1, \big[ b_r b_s , A_t \big] \xi_2 \right\rangle \Big| \\ &\hspace{2cm} \leq C e^{c|t|} \| F \|_2 N^{-3/4} \| (\cN+1)^{1/2} \xi_1 \| \| (\cN+1) \xi_2 \|\,. \end{split} \end{equation} 
Moreover, assuming additionally that $F(r,s)=\overline{F(s,r)}$, we also have 
\begin{equation}\label{eq:Fopbd} \begin{split} 
\Big| \int & dr ds   F(r,s)   \left\langle \x_1,  [\brrs\bss, A_t] \x_2 \right\rangle \Big| \\ & \leq C\expt N^{-3/4} \min \Big( \norm{F}_2, \sup_s \int dr \abs{F(r,s)} \Big) \norm{\cNplusoneh\x_1}\norm{\cNplusone\x_2}. \end{split} \end{equation} 

\end{lemma}

Lemma \ref{lm:cubicquadraticterms} is very similar to \cite[Lemma 5.2]{BSS1}. The main differences are the presence of the cutoff $\Theta (\cN)$ in the definition (\ref{eq:defA}) of $A_t$ (which plays no role in the proof) and the fact that bounds in \cite[Lemma 5.2]{BSS1} only control commutators in expectation, while (\ref{eq:F2bd}), (\ref{eq:Fopbd}) control all matrix elements. However, the proof only requires straightforward adaptations. We skip the details. 

Finally, we need to control the action of $A_t$ on the cubic operator $\cC_{N,t}$ 
on the r.h.s. of (\ref{eq:cGNt-prop}). This is the aim of the next lemma. 
\begin{lemma}\label{lm:commutatorCNA}
Let $A_t$ be defined as in \eqref{eq:defA}, with parameters $M = m^{-1} = \sqrt{N}$ and $\cC_{N,t}$ as in (\ref{eq:cubic}). Furthermore, let 
\begin{equation}\label{eq:X1X2} \begin{split} \X_1 = &- 2 \intxy N^3 \VNxy\wmnxy \overline{\ptx}\pty\\
    &\hspace{1cm} \times
    \left(\bgys\bgx+ \bgys\bsxs + \bgx\bsy + \bsxs\bsy + \la \s_y,\s_x\ra\right) \\
    \X_2 = &-2 \int dx  \left( (N^3 V_N w_{N,m})\ast \abs{\ptt}^2 \right)(x)\\
    &\hspace{1cm} \times\left(\bgxs\bgx+ \bgxs\bsxs + \bgx\bsx + \bsxs\bsx + \norm{\s_x}^2 \right)\,.
    \end{split} \end{equation} 
Then, on $\Fperpt\times\Fperpt$, we have
    \begin{equation*}\label{eq:commutatorCNA}
    \begin{split}
    [\cC_{N,t}, A_t]= &\; \X_1+\X_2 +\cE_{[\cC_N, A_t]}
    \end{split}
    \end{equation*}
    where $\cE_{[\cC_N, A_t]}$ is such that
    \begin{equation}\label{eq:ECA}
    \begin{split}
    &\abs{\la \x_1, \cE_{[\cC_N, A_t]} \x_2\ra} \leq C\expt N^{-1/4} 
    \norm{\cHNplusNh\x_1}\norm{\cHNplusNh\cNplusoneh\x_2}\,. 
    \end{split}
    \end{equation} 
\end{lemma}

\begin{proof}
We write $\cC_{N,t} = \wt{\cC}_N + \wt{\cC}_N^*$, with 
\[ \wt{\cC}_N = N^{5/2} \int dx dy V_N (x-y) \wt{\ph}_t (y) \, b_x^* b_y^* [ b(\gamma_x) + b^* (\sigma_x)]\,. \]
We will also use the short-hand notation $B_x = b (\gamma_x) + b^* (\sigma_x)$. We have
\be\label{eq:commutatorcubicstart}
\begin{split}
[\cC_{N,t}, A_t] 
=\, & \frac{\Th(\cN)}{\sqrt{N}}\intxy \nu_t(x,y) [\widetilde{\cC_N}, \bxxs\byys\Bxx ]
+\frac{\Th(\cN)}{\sqrt{N}}\intxy \nu_t(x,y) [\widetilde{\cC_N}^*, \bxxs\byys\Bxx]\\
&+[\widetilde{\cC_N}+\widetilde{\cC_N}^*, \Th(\cN) ] \frac{1}{\sqrt{N}}\intxy \nu_t(x,y) \bxxs\byys\Bxx
 + \hc + [\cC_{N,t}, A_t-A_t^1]\,.
\end{split}
\ee
With the notation $W=1- \cN / N$ and $\tilde{W}=W- 1/N=1- (\cN+1)/N$ and with the commutation relations (\ref{eq:bcomm2}), we find 
\begin{equation*}
\begin{split} 
\label{commutatorbzzA}
[\bzz, \bxxs\byys\Bxx]\
    = \; &W\bxxs \d (z-y)\Bxx
    +W\byys \d (z-x)\Bxx
    -\frac{2}{N}\bxxs\ayys\azz\Bxx\\
    &+\bxxs\byys W \s_t (z,x)
    -\frac{1}{N}\bxxs\byys \asxs\azz \\ 
[\bzzs, \bxxs\byys\Bxx]
=  \;&-\bxxs\byys W \overline{\g_t (z,x)}
+\frac{1}{N} \bxxs\byys\azzs \agx \\ 
[\Brr, \bxxs\byys\Bxx]
    = \; &W\byys\overline{\g_t (x,r)}\Bxx
    +W\bxxs\overline{\g_t (y,r)}\Bxx
    -\frac{2}{N}\byys\axxs\agr\Bxx\\
    &+\bxxs\byys W (\la p_r,\s_x\ra -\la p_x,\s_r\ra )
    +\frac{1}{N}\bxxs\byys \asrs\agx
    -\frac{1}{N}\bxxs\byys \asxs\agr \\
    [\Brrs, \bxxs\byys\Bxx]
    = &\;
    W \overline{\s_t (x,r)}\byys\Bxx
    +W \overline{\s_t (y,r)}\bxxs\Bxx
    -\frac{2}{N}\bxxs\ayys\asr\Bxx
    -\bxxs\byys W\d (r-x)\\
    & -\bxxs\byys W \big(p_t (x,r) +\overline{p_t (r,x)} + \la p_x, p_r\ra - \la \s_r, \s_x\ra \big)\\
    &+\frac{1}{N}\bxxs\byys\agrs\agx 
    -\frac{1}{N}\bxxs\byys \asxs\asr.
\end{split}
\end{equation*}
A lengthy but straightforward computation leads to
\be\label{eq:commutatorcubicfourbs}
\begin{split}
    &[\brrs\bsss\Brr, \bxxs\byys\Bxx]\\
     &\;= -\bxxs\byys\bsss  \overline{p_t (r,x)} \tilde{W}\Brr
    -\brrs\bxxs\byys \tilde{W} \d(s-x) \Brr
    -\brrs\bxxs\byys \tilde{W} \overline{p_t (s,x)} \Brr
    +\frac{2}{N}\brrs  \bxxs\byys\asss \agx\Brr\\
    &\;\;+\brrs\bsss \byys\tilde{W}\overline{p_t (x,r)}\Bxx
    +\brrs\bsss \bxxs \tilde{W}\d (y-r)\Bxx
    +\brrs\bsss \bxxs \tilde{W}\overline{p_t (y,r)}\Bxx
    -\frac{2}{N}\brrs\bsss\byys\axxs\agr\Bxx\\
    &\;\;+\brrs\bsss\bxxs\byys W (\la p_r,\s_x\ra -\la p_x,\s_r\ra )
    +\frac{1}{N}\brrs\bsss\bxxs\byys \asrs\agx
    -\frac{1}{N}\brrs\bsss\bxxs\byys \asxs\agr\,.
\end{split}
\ee
Let us label by $R_1,\dots, R_{11}$ the contributions to $[\cC_{N,t}, A_t]$ associated with the terms in (\ref{eq:commutatorcubicfourbs}). We claim that these terms can all be included into the error 
$\cE_{[\cC_N, A_t]}$. Let us consider first terms with 6 creation and annihilation operators. As an example, 
using Prop. \ref{prop:propertiesphit}, Lemma \ref{lm:propeta} and the cutoff $\cN \leq M$, we bound 
\[
\begin{split}
\abs{\la\x_1, R_4 \x_2\ra}
&\leq \frac{C\expt}{N^{3/2}} \norm{\cVh\cNplusone\cutoffNlessM\x_1}\\
&\times\left( \int dxdydrds N^3\VNrs N^2\wmnxy^2 \norm{\agx\Brr\cutoffNlessM\x_2}^2\right)^{1/2} \\
&\leq C\expt \sqrt{m\frac{M^2}{N^3}} \norm{\cHNplusNh\x_1}\norm{\cHNplusNh\cNplusoneh\x_2}\,.
\end{split}
\]
Here we used that $a (\gamma_x) b^* (\sigma_r) = (1-\cN/N)^{1/2} a^* (\sigma_x) a_y + (1-\cN/N) \sigma_t (x,r)$ and the fact that $\| \sigma_t \|_2 \leq C$. The other terms with 6 creation and annihilation operators, namely $R_8, R_{10}$ and $R_{11}$, can be bounded in the same way. Next, we consider terms with a contraction, quartic in creation and annihilation operators. Let us start with 
\[
\begin{split}
\abs{\la\x_1, R_1 \x_2\ra}
&\leq \frac{C\expt}{N} \norm{\cNplusone\cutoffNlessM\x_1}\\
&\times\left( \int dxdydrds N^3\VNrs  N^2\wmnxy^2 \abs{p_t(r,x)}^2\norm{\cNplusoneh\Brr\x_2}^2\right)^{1/2} \\
&\leq C\expt \sqrt{m\frac{M}{N^2}} \norm{\cNplusoneh\x_1}\norm{\cNplusone\x_2}
\end{split}
\]
where we used Lemma \ref{lm:propeta} to bound $\sup_r \| p_r \|_2 \leq C e^{c|t|}$. The term $R_3$ is estimated in exactly the same way. Furthermore, 
\[
\begin{split}
\abs{\la\x_1, R_5\x_2\ra}
&\leq \frac{C\expt}{N^{1/2}} \left( \int dxdydrds N^2\VNrs \abs{p_t (r,x)}^2\norm{\cNplusoneto{-1/2}\ayy\arr\ass \cutoffNlessM\x_1}^2\right)^{1/2}\\
&\times \left( \int dxdydrds N^3\VNrs \abs{\nu_t(x,y)}^2 \norm{\cNplusoneh\Brr\x_2}^2\right)^{1/2} \\
&\leq C\expt \sqrt{\frac{m}{N}} \norm{\cHNplusNh\x_1}\norm{\cHNplusNh\cNplusoneh\x_2}\,,
\end{split}
\]
and $R_7$ can be bounded analogously. Also $R_9$ satisfies the same estimate, since $\abs{\la p_r, \s_x\ra}\leq\norm{p_r}\norm{\s_x}$. Finally, we control 
\[
\begin{split}
\abs{\la\x_1, R_2 \x_2\ra}
&\leq \frac{C\expt}{N^{1/2}} \left( \int dxdydr N^2\VN (r-x)\norm{\cNplusoneto{-1/2}\ayy\arr\axx \cutoffNlessM\x_1}^2\right)^{1/2}\\
&\times \left( \int dxdydr N^3\VN (r-x) N^2\wmnxy^2 \norm{\cNplusoneh\Brr\x_2}^2\right)^{1/2} \\
&\leq C\expt \sqrt{\frac{m}{N}} \norm{\cHNplusNh\x_1}\norm{\cHNplusNh\cNplusoneh\x_2}\,,
\end{split}
\]
and we observe that $R_6$ is essentially the same as $R_2$, after renaming variables. 

Let us now consider the second term on the r.h.s. of (\ref{eq:commutatorcubicstart}). Here, we have to compute the commutator $[\Brrs\bss\brr, \bxxs\byys\Bxx]$. The computations are more involved than in \eqref{eq:commutatorcubicfourbs}, because now there can be multiple contractions, leading to contributions that are 
quadratic or even constant in creation and annihilation operators. The main contributions are those where $b_s$ and $b_r$ are contracted with $b_x^*, b_y^*$. There are two such contributions. Assuming that $b,b^*$ satisfy canonical commutation relations (it is easy to check that the corrections are negligible), they are given by 
\begin{equation*}\label{eq:fully} \begin{split} &- \Theta (\cN) \int dx dy N^3 V_N (x-y) w_{N,m} (x-y) \overline{\wt{\ph}_t (x)} \wt{\ph} (y) B_y^* B_x \\ &- \Theta (\cN) \int dx (N^3 V_N w_{N,m} * |\wt{\ph}_t|^2) (x) B_x^* B_x + \text{h.c.}  \\ &= \X_1 + \X_2 + \cE \end{split}  \end{equation*} 
with the error 
\[ |\langle \xi_1, \cE \xi_2 \rangle | \leq C e^{c|t|} N^{-1/4} \| (\cN+1)^{1/2} \xi_1 \| \| (\cN+1) \xi_2\|\,, \]
needed to remove the cutoff $\Theta (\cN)$ (arguing similarly as in (\ref{eq:M11}), with the choice $M = N^{1/2}$). Terms involving contractions between $b_r, b_s$ and $B_x$ (or between $B_r^*$ and $b_x^*, b_y^*$) are smaller, because they produce factors of $\sigma_t (r,x)$ or $\sigma_t (s,x)$, which are in $L^2$. As an example, consider the term $\delta (s-x) \sigma_t (r,x) B_r^* b_y^*$ (where we contracted $b_s$ with $b_x^*$ and $b_r$ with $B_x$ (ignoring again corrections to the canonical commutation relations). It produces a contribution $S$ to  (\ref{eq:commutatorcubicstart}), which can be bounded by 
\[
\begin{split}
\abs{\la\x_1, S \x_2\ra}
&\leq C\expt 
\left( \int dxdydr N^3 V_N(r-x)  \norm{\cNplusoneto{-1/2}\ayy\Brr \x_1}^2\right)^{1/2}\\
&\hspace{1cm} \times\left( \int dxdydr N^3 V_N(r-x) |\nu_t (x,y)|^2 \right)^{1/2} \norm{\cNplusoneh\x_2} \\
&\leq C\expt N^{-1/4} \norm{\cNplusoneh\x_1}\norm{\cNplusoneh\x_2}
\end{split}
\]
where we used $\abs{\s_t (r,x)}\leq C  N \abs{\ptr\ptx}$ from Lemma \ref{lm:propeta} and $\| \nu_t \| \leq C \sqrt{m}$ from Lemma \ref{lm:propnu} (and the choice $m= N^{-1/2}$). Terms arising from $[\Brrs\bss\brr, \bxxs\byys\Bxx]$ containing 6 or 4 creation and annihilation operators can be bounded as we did above with $R_1, \dots , R_{11}$. 

The next term in \eqref{eq:commutatorcubicstart} has the form 
\be\label{eq:cubiccommutatorTheta}
\begin{split}
&[\widetilde{\cC_N}+\widetilde{\cC_N}^*, \Th(\cN) ] \frac{1}{\sqrt{N}}\intxy \nu_t(x,y) \bxxs\byys\Bxx +\hc\\
&= \intrs N^2\VNrs \pts (\Th( \cN -1)-\Th( \cN ))\brrs\bsss\bgr\intxy \nu_t(x,y) \bxxs\byys\Bxx\\
&+ \intrs N^2\VNrs \pts (\Th( \cN -3)-\Th( \cN ))\brrs\bsss\bsxs\intxy \nu_t(x,y) \bxxs\byys\Bxx\\
&+ \intrs N^2\VNrs \overline{\pts} (\Th( \cN +1))-\Th( \cN )\bgrs\bss\brr\intxy \nu_t(x,y) \bxxs\byys\Bxx\\
&+ \intrs N^2\VNrs \overline{\pts} (\Th( \cN +3))-\Th( \cN )\bsr\bss\brr\intxy \nu_t(x,y) \bxxs\byys\Bxx
+\hc\\
&=\sum_{i=1}^4T_i +\hc
\end{split}
\ee
The contribution $T_2$ is already normally ordered. It can simply be bounded by Cauchy-Schwarz. We find 
\[\begin{split}  |\langle \xi_1, T_2 \xi_2 \rangle | &\leq \int dx dy dr ds N^2 V_N (r-s) |\wt{\ph}_t (s)| |\nu_t (x,y)| \\ &\hspace{2cm} \times \| b_r b_s b (\sigma_x) b_x b_y (\cN+1)^{-3/2} \xi_1 \| \| B_x (\cN+1)^{3/2} \mathbbm{1}(\cN\leq M) \xi_2 \| \\ 
&\leq C e^{c|t|} \sqrt{\frac{m}{N}} \| \cV_N^{1/2} \xi_1 \| \| (\cN+1)^{2} \mathbbm{1}(\cN\leq M)\xi_1 \|  \\ &\leq C e^{c|t|} N^{-1/4} \| \cV_N^{1/2} \xi_1 \| \| (\cN+1) \xi_2 \| \end{split} \]
where we used $\sup_x \| \sigma_x \|_2 \leq C e^{c|t|}$ from Lemma \ref{lm:propeta} and $\sup_x \| \nu_x \|_2 \leq C \sqrt{m} e^{c|t|}$ from Lemma \ref{lm:propnu} (with $m = N^{-1/2}$). All other normally ordered term emerging from (\ref{eq:cubiccommutatorTheta}) can be treated analogously. On the other hand, terms involving commutators (produced through normal ordering) are closely related with the contributions discussed above from the first two terms on the r.h.s. of (\ref{eq:commutatorcubicstart}). Due to the presence of the differences $\Theta (\cN+1) - \Theta (\cN)$ (or similar), also the contributions where $b_s, b_r$ are contracted with $b_x^* b_y^*$ (arising from $T_3$ and $T_4$) are negligible, here (since $\norm{(\Theta (\cN+1) - \Theta (\cN))\xi}\leq C/M \norm{\mathbbm{1}(M/2\leq\cN\leq M)\xi}$, we can gain a factor $M^{-1}$, arguing similarly as in (\ref{eq:M11})). 

Finally, we deal with the commutator $[\cC_{N,t}, A_t-A_t^1]$ using the identity \eqref{eq:differenceprojectionA}. The resulting terms can be treated analogously as we did with the contributions to $[\cC_{N,t} , A_t^1]$ (but these terms are less singular and thus simpler to handle). They all satisfy the estimate (\ref{eq:ECA}). We skip the details. 
\end{proof}

We are now ready to proceed with the proof of Prop. \ref{prop:cJNt}. 
\begin{proof}[Proof of Prop. \ref{prop:cJNt}] 
Recall from (\ref{eq:cJN2}) that 
\begin{equation}\label{eq:cJN1} \cJ_N (t) = (i\partial_t e^{-A_t}) e^{A_t} + e^{-A_t} \cG_N (t) e^{A_t} \end{equation} 
where, by Prop. \ref{prop:cGNt},  
\[ \cG_{N}(t) = \k_{\cG}(t) +  \cG_{2,N}(t) + \cC_{N,t}+ \cVN + \cE_{\cG_N} (t), \]
with the eror $\cE_N (t)$ satisfying the bound \eqref{eq:cEGN}. From Lemma \ref{lm:gron-A}, we find 
\[ \begin{split} |\langle \xi_1, e^{-A_t} &\cE_{\cG_N} (t) e^{A_t} \xi_2 \rangle | \leq \frac{C e^{c|t|}}{\sqrt{N}}  \| (\cH_N + \cN+1)^{1/2} \xi_1 \| \| (\cH_N + \cN^3 + 1)^{1/2} (\cN+1) \xi_2 \|\,.  \end{split} \]

With Lemma \ref{lm:commutatorHA}, Lemma \ref{lm:commutatorCNA} and Lemma \ref{lm:cubicquadraticterms}, we claim that 
\begin{equation}\label{eq:propJ1}  e^{-A_t} (\cG_{2,N} (t) + \cC_{N,t} +\cV_N) e^{A_t} = \cG_{2,N} (t) + \frac{1}{2} \X_1 + \frac{1}{2} \X_2 + \cV_N + \cE  \end{equation} 
where $\Xi_1,\X_2$ are defined in Eq. \eqref{eq:X1X2} and 
   \[
    \begin{split}
    \abs{\la \x_1, \cE \x_2\ra}
    &\leq C\expt 
    N^{-1/4}
    \norm{\cHNplusNh\x_1}\norm{\cHNplusNh\cNplusoneh\x_2}\,.
    \end{split}
    \]
To prove (\ref{eq:propJ1}), we start by observing that 
\[ e^{-A_t} (\cG_{2,N} (t) - \cK) e^{A_t} = (\cG_{2,N} (t) - \cK) + \int_0^1 ds \, e^{-sA_t}  \big[ (\cG_{2,N} (t) - \cK) , A_t \big]  e^{sA_t} \]
Going through the terms in $\cG_{2,N} (t) - \cK$ in (\ref{eq:quadgenfinal}), we can check with Prop. \ref{prop:propertiesphit} and Lemma \ref{lm:propeta} that they all have one of the forms considered in Lemma \ref{lm:cubicquadraticterms}. It follows that 
\[ e^{-A_t} (\cG_{2,N} (t) - \cK) e^{A_t} = (\cG_{2,N} (t) - \cK) + \cE' \]
where, applying also Lemma \ref{lm:gron-A},  
\[ \begin{split}  | \langle \xi_1, \cE' \xi_2 \rangle | &= C e^{c|t|} N^{-3/4} \int_0^1 ds \, \| (\cN+1)^{1/2} e^{s A_t} \xi_1 \| \| (\cN+1) e^{sA_t} \xi_2 \|  \\ & \leq C e^{c|t|} N^{-3/4} \| (\cN+1)^{1/2} \xi_1 \| \| (\cN+1) \xi_2 \| \,.\end{split} \]

On the other hand, with Duhamel's formula, we can write
\[
\begin{split}
    e^{-A_t} \cHN  e^{A_t}
    =&\; \cHN + \int_0^1 ds e^{-s A_t}[\cHN,A_t]e^{s A_t}\\
    =&\; \cHN + \int_0^1 ds e^{-s A_t}(-\cC_{N,t} +\cE_{[\cHN,A_t]}) e^{s A_t}\\
    =&\; \cHN - \cC_{N,t} 
    + \int_0^1 ds e^{-s A_t}\cE_{[\cHN,A_t]} e^{s A_t} 
    \\
    &+\int_0^1 ds\int_0^s dr e^{-r A_t}(-\Xi_1-\Xi_2 +\cE_{[\cC_N,A_t]}) e^{r A_t}\\
    =&\; \cHN - \cC_{N,t} -\frac{1}{2} (\Xi_1+\Xi_2) 
    \\ &+ \int_0^1 ds e^{-s A_t}\cE_{[\cHN,A_t]} e^{s A_t} +  \int_0^1 ds\int_0^s dr e^{-r A_t}\cE_{[\cC_N,A_t]} e^{r A_t} \\ &- \int_0^1 ds\int_0^s dr \int_0^r d\tau \, e^{-\tau A_t} \big[\X_1 + \X_2 , A_t \big] e^{\tau A_t}\,. 
\end{split}
\]
Similarly,
\[
\begin{split}
    e^{-A_t} \cC_{N,t}  e^{A_t}
    = \; &\cC_{N,t}+ \int_0^1 ds e^{-s A_t}[\cC_{N,t} ,A_t]e^{s A_t}\\
    = \; &\cC_{N,t} + \int_0^1 ds e^{-s A_t}\big(\Xi_1+\Xi_2 +\cE_{[\cC_N,A_t]}\big) e^{s A_t}\\
    = \; &\cC_{N,t} +\Xi_1+\Xi_2 \\ &+ \int_0^1 ds e^{-s A_t}\cE_{[\cC_N,A_t]} e^{s A_t} + \int_0^1 ds \int_0^s dr \, e^{-r A_t} \big[ \X_1 + \X_2 , A_t \big] e^{r A_t} \,.
\end{split}
\]
Applying Lemma \ref{lm:commutatorHA}, Lemma \ref{lm:commutatorCNA} and Lemma \ref{lm:cubicquadraticterms} (noticing that the commutator of the quadratic operators $\X_1, \X_2$ with $A_t$ is a sum of terms that can be bounded with (\ref{eq:F2bd}), (\ref{eq:Fopbd})) and propagating the estimates through the cubic phase with Lemma \ref{lm:gron-A}, we arrive at (\ref{eq:propJ1}). 

As for the first term on the r.h.s. of (\ref{eq:cJN1}), we observe that, since $\sup_x \| \dot{\nu}_{t,x} \| \leq C e^{c|t|} N^{-1/4}$ from Lemma \ref{lm:propnu}, 
\begin{equation*}
	\begin{split}
	\abs{ \la\x_1, (i\partial_te^{-A_t})e^{A_t} \x_2\ra}
	&\leq \int_0^1ds  \abs{\la e^{s A_t} \x_1, [i\partial_t A_t]e^{s A_t} \x_2\ra} \\
	&\leq C\expt N^{-3/4} \norm{\cNplusoneh\x_1}\norm{\cNplusone\x_2}.
	\end{split}
\end{equation*}

We conclude that
\[ \cJ_N (t) = \kappa_\cG (t) + \cG_{2,N} (t) + \frac{1}{2} \Xi_1 + \frac{1}{2} \Xi_2 + \cE' \]
where 
\begin{equation}\label{eq:cE'-J} | \langle \xi_1, \cE'  \xi_2 \rangle | \leq C e^{c|t|} N^{-1/4} \norm{\cHNplusNh\x_1}\norm{(\cH_N + \cN^3 + 1)^{1/2} (\cN+1) \x_2}\,.\end{equation} 
Next, we observe that in the terms $\Xi_1, \Xi_2$, as defined in (\ref{eq:X1X2}), we can replace the parameter $m = N^{-1/2}$ with the fixed, $N$-independent, parameter $\ell \in (0;1)$, at the expense of small error. In fact, setting 
\begin{equation}\label{eq:X1X2'} \begin{split} \X'_1 = &- 2 \intxy N^3 \VNxy w_{N,\ell} (x-y) \overline{\ptx}\pty\\
    &\hspace{1cm} \times
    \left(\bgys\bgx+ \bgys\bsxs + \bgx\bsy + \bsxs\bsy + \la \s_y,\s_x\ra\right) \\
    \X'_2 = &-2 \int dx  \left( (N^3 V_N w_{N,\ell})\ast \abs{\ptt}^2 \right)(x)\\
    &\hspace{1cm} \times\left(\bgxs\bgx+ \bgxs\bsxs + \bgx\bsx + \bsxs\bsx + \norm{\s_x}^2 \right)
    \end{split} \end{equation} 
we find 
\begin{equation}\label{eq:mtoell} |\langle \xi_1, (\X_j - \X'_j) \xi_2 \rangle | \leq \frac{C e^{c|t|}}{\sqrt{N}} \| (\cN+1)^{1/2} \xi_1 \| \| (\cK + \cN+1)^{1/2}  \xi_2 \| \end{equation}
for $j=1,2$. To show (\ref{eq:mtoell}), we argue as we do in the proof of Prop. \ref{prop:cJN-cJinfty} to control the terms $\text{I}, \text{II}, \text{III}$ (in that Proposition, we control convergence of $N^3 V_N f_{N,\ell}$ towards a $\delta$-distribution, but the same argument implies convergence of $N^3 V_N f_{N,m}$ towards a $\delta$-distribution and therefore allows us to control the difference $N^3 V_N (f_{N,\ell} - f_{N,m}) = N^3 V_N (w_{N,m} - w_{N,\ell})$; note that, to handle $N^3 V_N f_{N,m}$ we need to use (\ref{eq:Vfa0}) with $\ell$ replaced by $m = N^{-1/2}$, which makes some of the estimates, like the one for the second term on the r.h.s. of (\ref{eq:DtDinf}), worse). 

Combining the operator $\cG_{2,N} (t)$, as defined in (\ref{eq:quadgenfinal}), with the terms $\Xi'_1, \Xi'_2$ from (\ref{eq:X1X2'}) we obtain, as quadratic form on $\cF_{\perp \wt{\ph}_t}^{\leq N}$ (so that the ``projected'' operators $\tl{b}, \tl{b}^*$ are the same as $b,b^*$), the operator $\cJ_{2,N} (t)$ in (\ref{eq:generatorapprox}), up to a small error due to the term on the seventh line of (\ref{eq:generatorapprox}), whose matrix elements can be bounded by  
\begin{equation}\label{eq:errJG} \Big| \int dx \, N^3 (V_N f_{N,\ell} * |\wt{\ph}_t|^2) (x) \la \xi_1, (a_x^* a_x - b_x^* b_x) \xi_2 \rangle \Big|  \leq C N^{-1}  \| (\cN+1)^{1/2} \xi_1 \| \| (\cN+1)^{3/2} \xi_2 \|\,. \end{equation} 
Remark that, despite its smallness on $\cF_{\perp \wt{\ph}_t}^{\leq N}$, we inserted this term in the definition (\ref{eq:generatorapprox}) to make sure that $\cU_{2,N} (t;s)$ maps $\cF_{\perp \wt{\ph}_s}^{\leq N}$ into 
$\cF_{\perp \wt{\ph}_t}^{\leq N}$. 
Absorbing (\ref{eq:errJG}), together with (\ref{eq:cE'-J}), into the error term $\cE_{\cJ_N} (t)$, we conclude the proof of Prop.  \ref{prop:cJNt}. 
\end{proof}

\appendix

\section{Properties of $f_\ell, \ph_t, \wt{\ph_t}, \eta_t, \eta_{\infty,t}, \nu_t$} 
\label{app:eta}

In this appendix, we collect some analytic properties of functions and kernels that are used throughout the paper to construct the approximation of the many-body dynamics. 

In the first lemma, whose proof can be found in \cite{BS,BBCS3,ESY0}, we consider the ground state solution of the Neumann problem (\ref{eq:scatl}) on the ball $|x| \leq N \ell$, with the normalization $f_\ell (x) = 1$ for $|x| = N \ell$. 
 
	\begin{lemma} \label{3.0.sceqlemma}
		Let $V \in L^3 (\bR^3)$ be non-negative, compactly supported and spherically symmetric. Fix $\ell > 0$ and let $f_\ell$ denote the solution of \eqref{eq:scatl}. For $N$ large enough the following properties hold true.
		\begin{enumerate}
			\item [i)] We have 
			\begin{equation}\label{eq:lambdaell} 
				\lambda_\ell = \frac{3\aa }{(\ell N)^3} \left(1 +\mathcal{O} \big(\aa  / \ell N\big) \right)\,.
			\end{equation}
			\item[ii)] We have $0\leq f_\ell, w_\ell\leq1$. Moreover there exists a constant $C > 0$ such that     
			\begin{equation} \label{eq:Vfa0} 
				\left|  \int  V(x) f_\ell (x) dx - 8\pi \aa   \right| \leq \frac{C \aa^2}{\ell N} \, 	\end{equation}
			for all $\ell \in (0;1/2)$ and $N \in \bN$.
			\item[iii)] There exists a constant $C>0 $ such that 
			\begin{equation}\label{3.0.scbounds1} 
				w_\ell(x)\leq \frac{C}{|x|+1} \quad\text{ and }\quad |\nabla w_\ell(x)|\leq \frac{C }{x^2+1}
			\end{equation}
			for all $x \in \bR^3$, $\ell \in (0;1/2)$ and all $N\in \bN$ large enough.   
		\end{enumerate}        
	\end{lemma}

Next, we consider solutions of the Gross-Pitaevskii equation (\ref{eq:GPtd}) and of the modified $N$-dependent 
Gross-Pitaevskii equation \eqref{eq:GPmod}. The proof of Prop. \ref{prop:propertiesphit} is essentially contained (up to straightforward changes) in \cite[Theorem 3.1]{BDS}, \cite[Prop. 4.2]{BS} and \cite[Prop. B.1]{BCS}.
	\begin{prop}\label{prop:propertiesphit}
		Let $V \in L^3(\bR^3)$ be a non-negative, spherically symmetric, compactly supported potential. Let  $\ph \in L^2(\bR^3)$ with $\|\ph\| =1$. Recall the scattering solution \eqref{eq:scatlN} which enters the modified Gross-Pitaevskii equation \eqref{eq:GPmod}; assume $\ell\in (0,1/2)$.
		\begin{itemize}
			\item[i)]  For $\ph \in H^1(\bR^3)$, there exist unique global solutions $t \to \ph_t$ and $t \to \wt{\ph}_t$ in  $\cC(\bR, H^1(\bR^3))$ of the Gross-Pitaevskii equation \eqref{eq:GPtd} and, respectively, of the modified Gross-Pitaevskii equation \eqref{eq:GPmod} with initial datum $\ph$. We have $\|\ph_t\|=\|\wt{\ph}_t\|=1$ for all $t\in \bR$. Furthermore, there exists a constant $C >0$ such that
			\begin{equation*}
				\label{eq:normH1phi}
				\|\ph_t\|_{H^1}, 	\|\tilde\ph_t\|_{H^1}\leq C\,.
			\end{equation*}
			\item[ii)]  If $\ph \in H^m(\bR^3)$ for some $m \geq 2$, then $\ph_t, \tilde\ph_t \in H^m(\bR^3)$ for every $t \in \bR$. Moreover there exist constants $C$ depending on $m$ and on $\|\ph\|_{H^m}$, and $c>0$ depending on $m$ and on $\|\ph\|_{H^1}$ such that for all $t \in \bR$ 
			\begin{equation*}
				\label{eq:normHmphi}
				\|\ph_t\|_{H^m}, 	\|\tilde\ph_t\|_{H^m}\leq Ce^{c|t|}\,.
			\end{equation*}
			\item[iii)]  Suppose $\ph \in H^4(\bR^3)$. Then there exist constants $C>0$ depending on $\|\ph\|_{H^4}$, and $c>0$ depending on $\|\ph\|_{H^1}$ such that for all $t \in \bR$ 
			\begin{equation*}\label{eq:normH2partialphi}
				\|\dot{\tilde\ph}_t\|_{H^2}, 	\|\ddot{\tilde\ph}_t\|_{H^2}\leq Ce^{c|t|}\,.
			\end{equation*}
			Furthermore, if $\ph \in H^6(\bR^3)$ there exist constants $C>0$ depending on $\norm{\ph}_{H^6}$, and $c>0$ depending on $\norm{\ph}_{H^1}$ such that for all $t \in \bR$ 
			\begin{equation*}\label{eq:normH4partialphi}
				\norm{\dot\ph_t}_{H^4}, \norm{\dptt}_{H^4} \leq C\expt.
			\end{equation*}
			\item[iv)] Suppose $\ph \in H^2(\bR^3)$. Then there exist constants $C,c_1,c_2 >0$ such that for all $t \in \bR$
			\begin{equation*}
				\label{eq:diffdynamics}
				\|\ph_t-\tilde\ph_t\| \leq CN^{-1}\exp(c_1\exp(c_2|t|))\,.
			\end{equation*}
			For $\ph\in H^6(\RRR^3)$ there are constants $C, c > 0$ such that
			\begin{equation*}\label{eq:differencephiH4}
			\norm{\ph_t-\ptt}_{H^4}\leq C N^{-1}\expexpt
			\end{equation*}
			and
            \begin{equation*}\label{eq:differencephidotH2}
			\norm{\dot\ph_t-\dptt}_{H^2}\leq C N^{-1}\expexpt.
			\end{equation*}
		\end{itemize}
\end{prop}

Recall now the definition (\ref{eq:defeta}), depending on the parameters $N,\ell$, of the kernel $\eta_t$ appearing in the generalized Bogoliubov transformation $e^{B_t}$ and the notation $\gamma_t = \cosh \eta_t$, $\sigma_t = \sinh \eta_t$. Furthermore, we set $p_t = \gamma_t - \mathbbm{1}$, $r_t = \sigma_t - \eta_t$ and $\mu_t = \eta_t - k_t$ (recall (\ref{eq:defk})). Several bounds for the operators $\eta_t, \gamma_t, \sigma_t, p_t, r_t$  (for their integral kernels) and for their time-derivatives are established in the next lemma, whose proof is a straightforward adaptation of \cite[Lemma 3.3 and 3.4]{BDS}, \cite[Lemma 4.3]{BS}, \cite[Appendix C]{BCS}. 
\begin{lemma}\label{lm:propeta}
		Let $\wt{\ph}_t$ be the solution of \eqref{eq:GPmod} with initial datum $\ph \in H^4(\bR)$. Let $w_\ell =1-f_\ell$ with $f_\ell$  the ground state solution of the Neumann problem \eqref{eq:scatl} and let $\ell \in (0;1/2)$. Let $k_t, \eta_t,\mu_t$ be defined as in \eqref{eq:defk},\eqref{eq:defeta}. Then there exist constants $C,c >0$ depending only on  $\|\ph\|_{H^4}$ (or lower Sobolev norms of $\ph$) and on $V$ such that the following bounds hold uniformly in $\ell$, for all $t\in \bR$.
		\begin{itemize}
			\item [i)] We have 
			\begin{equation*}
				\label{eq:boundeta}
				\|\eta_t\| \leq C\ell^{1/2}
			\end{equation*} 
			and also
			\[ \|\nabla_j \eta_t\| \leq C\sqrt N, \quad \|\nabla_j\mu_t\| \leq C 
   \]
			for $j =1,2$.
   With $\nabla_1\eta_t$ and $\nabla_2\eta_t$ we indicate the kernels $\nabla_x\eta_t(x;y)$ and $\nabla_y\eta_t(x;y)$, similar definitions hold for $\D_j\eta_t$, for $j=1,2$. 
			Let $\s_t, p_t, r_t$ be defined as in \eqref{eq:defcoshsinh} and after \eqref{eq:cJ2NV}, we obtain 
			\begin{equation*}
				\label{eq:boundsrp}\begin{split}
				\|\s_t\|, \|p_t\|,\|r_t\|, \|\nabla_j p_t\|, \|\nabla_j r_t\|\leq C\,,\\
				\norm{\D_j p_t}, \norm{\D_j r_t}, \norm{\D_j \mu_t}\leq C\expt,  \norm{\nabla_j \s_t}\leq C \expt \sqrt{N} \,.
			\end{split}
		\end{equation*}
			\item[ii)] for a.e. $x,y \in \bR^3$, $n \in \bN$, $n \geq 2$,  we have the pointwise bounds 
			\begin{equation*}
				\label{eq:boundeta2}
				\begin{split}
					|\eta_t(x;y)| &\leq \frac{C}{|x-y|+N^{-1}}|\wt\ph_t(x)|||\wt\ph_t(y)|\\
					|\mu_t(x;y)|, |p_t(x;y)|,|r_t(x;y)| &\leq C|\wt\ph_t(x)|||\wt\ph_t(y)|\,\\
					|\nabla_x\eta_t(x,y)| & \leq C \big(|\nabla\wt\ph_t(x)+|\wt\ph_t(x)|\big) |\wt\ph_t(y)|\Big(\frac{\chi(|x-y| \leq \ell)}{|x-y|^2}+1\Big)\,.
			\end{split}\end{equation*}
			\item[iii)] Moreover we have
			\begin{equation*}
				\label{eq:normetax}
				\sup_x \|\eta_x\|^2 ,\; \sup_x \|k_{t,x}\|^2,\;\sup_x \|\mu_{t,x}\|^2\leq C\|\tilde\ph_t\|_{H^2}\leq Ce^{c|t|}
			\end{equation*}
			where we indicate with $\sup_{x}\|\eta_x\|^2 = \sup_{x}\int |\eta_t(x;y)|^2dy$
            and 
			\[
			\|\s_{t,x}\|, \|p_{t,x}\|, \|r_{t,x}\| \leq Ce^{c|t|}\,.
			\]
			\item[iv)] For $j=1,2$ we have the following bounds for the time derivatives
			\begin{equation*}
				\label{eq:timederivative}
				\|\partial_t\eta_t\|, \|\partial_t^2\eta_t\| \leq Ce^{c|t|}\,,
			\end{equation*}
			and also 
			\[
			\|\partial_t\nabla_j\eta_t\|\leq C\sqrt Ne^{c|t|},\, \|\partial_t\nabla_j\mu_t\|\leq Ce^{c|t|}.
			\]
			Furthermore
			\begin{equation*}\label{eq:mixedderivatives}
			\|\partial_t\s_t\|,
			\|\partial_tr_t\|,
			\|\partial_tp_t\|, 
			\|\nabla_j\partial_t p_t\|,
			\|\nabla_j\partial_t r_t\|,
			\norm{\D_j\partial_t p_t},
			\norm{\D_j\partial_t r_t},
			\norm{\D_j\partial_t \mu_t}
			\leq Ce^{c|t|}\,.
			\end{equation*}
			\item[v)] For a.e. $x,y \in \bR^3$ we have the pointwise bounds
			\[|\partial_t\eta_t(x;y)| \leq C \Big[1+ \frac{1}{|x-y| + N^{-1}}\Big]\big[|\dot{\tilde\ph}_t(x)||\tilde\ph_t(y)| + |\tilde\ph_t(x)||\dot{\tilde\ph}_t(y)|+ |\tilde\ph_t(x)||\tilde\ph_t(y)|\big]\,,\]
			and
			\[\begin{split}
				&|\partial_t\mu_t(x;y)|,  |\partial_tr_t(x;y)|, |\partial_tp_t(x;y)|\\
				&\hspace{2cm}\leq C e^{c|t|}\big[|\dot{\tilde\ph}_t(x)||\tilde\ph_t(y)| + |\tilde\ph_t(x)||\dot{\tilde\ph}_t(y)|+ |\tilde\ph_t(x)||\tilde\ph_t(y)|\big]\,.
			\end{split}\]
			\item[vi)] Finally, we have
			\[
			\|\partial_t\eta_x\|, \|\partial_tk_{t,x}\|, \|\partial_t\mu_{t,x}\| \leq Ce^{c|t|}
			\]
			and
			\[
			\|\partial_t\s_{t,x}\|, \|\partial_tp_{t,x}\|, \|\partial_tr_{t,x}\| \leq Ce^{c|t|}\,.
			\]
\end{itemize}
\end{lemma}

While the kernels $\eta_t, \gamma_t, \sigma_t, p_t, r_t$ considered in the last lemma are used in the definition of the fluctuation dynamics $\cU_N$ and of its quadratic approximation $\cU_{2,N}$, the limiting quadratic evolution $\cU_{2,\infty}$ is defined in (\ref{eq:limitingdynamics}) in terms of limiting kernels $\eta_{\infty,t} , \gamma_{\infty,t}, \sigma_{\infty ,t}$, $p_{\infty,t}, r_{\infty,t}$. To show the well-posedness of $\cU_{2,\infty}$ and to compare it with $\cU_{2,N}$, we need some bound on these limiting objects. 
\begin{lemma}\label{lm:bds-infty} 
Let $w_{\infty, \ell}, \eta_{\infty,t}, \gamma_{\infty,t}, \sigma_{\infty,t}, p_{\infty,t} , r_{\infty,t}$ be defined as in \eqref{eq:defwio}-\eqref{eq:ginfty}. 
\begin{itemize}
\item[i)] The limiting kernels satisfy
\begin{equation*}\label{eq:normbounds}
\begin{split}
    &\norm{\eta_{t, \io}},
    \norm{\s_{ \io}},
    \norm{p_{\io}},
    \norm{r_{ \io}}\leq C\\
    &\norm{\D_j p_{ \io}},
    \norm{\D_j r_{\io}},
    \norm{\D_j \mu_{t, \io}}\leq C\expt\\
    &\norm{\dot\eta_{t, \io}},
    \norm{\dot\s_{ \io}},
    \norm{\dot p_{ \io}},
    \norm{\dot r_{ \io}}\leq C\expt\\
    &\norm{\D_j \dot p_{ \io}},
    \norm{\D_j \dot r_{ \io}},
    \norm{\D_j \dot \mu_{t, \io}}\leq C\expt\\
    &\norm{\ddot\eta_{t,\io}}\leq C\expt
\end{split}
\end{equation*}
for almost all $x\in\RRR^3$ and $j=1, 2$.
\item[ii)] Let $R > 0$ such that $V(x) = 0$ for $|x| > R$. Then, we have 
\begin{equation}
	\label{eq:differencewwioglobal}
\begin{split} 
	\abs{N \wln(x) - w_{\io,\ell} (x)} &\leq 
	\begin{cases}
	\frac{C}{\abs{x}} \quad &0 \leq \abs{x} \leq R/N \quad\\
	 \frac{C}{N\abs{x}} \quad &R/N \leq \abs{x} \leq \ell\\
	0 & \ell\leq \abs{x}
	\end{cases} \\
	\abs{N\nabla \wln(x) - \nabla w_{\io,\ell} (x)} &\leq 
	\begin{cases}
	\frac{C}{\abs{x}^2} \quad &0 \leq \abs{x} \leq R/N \quad\\
	 \frac{C}{N\abs{x}^2} \quad &R/N \leq \abs{x} \leq \ell\\
	0 & \ell\leq \abs{x}
	\end{cases}
\end{split} 
\end{equation}
\item[iii)] We have 
\begin{align}
  &\norm{\eta_t-\eta_{t, \io}},
  \norm{\dot\eta_t-\dot\eta_{t, \io}},
    \norm{\s-\s_{\io}},
    \norm{p-p_{ \io}},
    \norm{r-r_{ \io}}\leq \frac{C\expexpt}{\sqrt N}\nonumber\\
  \label{eq:differencelaplace}  &\norm{\D_2 p- \D_2 p_{ \io}},
    \norm{\D_2 r- \D_2 r_{ \io}},
    \norm{\D_2 \mu_t - \D_2 \mu_{t, \io}}\leq \frac{C\expexpt}{\sqrt N}\nonumber.
\end{align}
\end{itemize}
\end{lemma}

\begin{proof}
The bounds in i) follow from (\ref{eq:defwio}), which implies that $|w_{\infty,\ell} (x)| \leq C/|x|$, $|\nabla w_{\infty, \ell} (x)| \leq C / |x|^2$ and $\Delta w_{\infty, \ell} (x) = 4 \pi \frak{a} \delta (x) + 3\frak{a} \chi (|x| \leq \ell) / \ell^3$. 

The bounds (\ref{eq:differencewwioglobal}) are trivial for $|x| > \ell$. In the region $|x| < R/N$, they follow by combining the estimates for $w_{\infty,\ell}, \nabla w_{\infty, \ell}$ with (\ref{3.0.scbounds1}). In the region 
$R/N < |x| < \ell$, we apply the identity (see \cite{BBCS4,BBCS3}) 
\[ 
w_{N,\ell}(x) 
= 1- \frac{\sin(\sqrt{\l_{N,\ell}} \, (|x|-\ell))}{\sqrt{\l_{N,\ell}} \, |x|} -\frac{\ell}{|x|}\cos(\sqrt{\l_{N,\ell}} \, (|x|-\ell))
\]
and the corresponding expression for $\nabla w_{N,\ell}$. Using $\lambda_{N,\ell}  \ell^2 \simeq 1/ (N\ell) \ll 1$ and Taylor expanding $\sin$ and $\cos$, we obtain (\ref{eq:differencewwioglobal}).

The bounds in iii) follow from ii) and from (\ref{eq:diffdynamics}) (because the limiting kernels $\eta_{t,\infty}, \dot{\eta}_{t,\infty}, \sigma_\infty , p_\infty, r_\infty, \mu_{t,\infty}$ are defined like $\eta_t, \dot{\eta}_t, \sigma, p , r, \mu_t$, with $Nw_{N,\ell}$ and $\wt{\ph}_t$ replaced by $w_{\infty, \ell}$ and $\ph_t$. We leave the details to the reader.
\end{proof}

Finally, in the next lemma we collect some estimates for the kernel $\nu_t$ introduced in (\ref{eq:defnu}) entering the definition of the cubic phase $A_t$ in (\ref{eq:defA}). The proof is a straightforward adaptation of the proof of Lemma \ref{lm:propeta} above (with the parameter $\ell$ replaced now by $m$). 

\begin{lemma}\label{lm:propnu}
	Under the same conditions as in Lemma \ref{lm:propeta}, the kernel $\n_t$ defined as in Eq. \eqref{eq:defnu} satisfies the following bounds: 
	\[
	\norm{\nu_t}\leq C\sqrt{m}\,, \quad
	\sup_x \norm{\nu_{t,x}}\leq C \expt \sqrt{m}\,, \quad
	\norm{\nu_{t,y}}\leq C \abs{\pty} \sqrt{m}\,,
	\]
	and the pointwise bound
	\[
	\abs{\nu_t(x,y)}\leq \frac{C \abs{\pty}}{\abs{x-y}+N^{-1}}.
	\]
	Furthermore, the time derivative satisfies similar estimates:
	\[
	\norm{\dot{\nu}_t}\leq C\expt \sqrt{m}\,, \quad
	\sup_x \norm{\dot{\nu}_{t,x}}\leq C \expt \sqrt{m}\,, \quad
	\norm{\dot{\nu}_{t,y}}\leq C \abs{\dpty} \sqrt{m}\,.
	\]
	
\end{lemma}

\end{document}